\documentclass[%
reprint,
superscriptaddress,
 amsmath,amssymb,
 aps,
]{revtex4-2}

\usepackage{mathtools}
\usepackage{graphicx}
\usepackage{dcolumn}
\usepackage{bm}
\usepackage{physics}

\usepackage{amsmath}
\usepackage{amssymb}
\usepackage{amsthm}
\usepackage{graphicx}
\usepackage{algorithmic}
\usepackage{mathrsfs}
\usepackage{braket}
\usepackage{xcolor}
\usepackage{tikz}
\usepackage[roman]{complexity}
\usepackage[caption=true]{subfig}

\newtheorem{theorem}{Theorem}
\newtheorem{lemma}{Lemma}

\newtheorem{proposition}{Proposition}

\newtheorem{assumption}{Assumption}
\theoremstyle{remark}

\usepackage{qcircuit}
\usepackage{booktabs}
\usepackage{threeparttable}

\def\>{\rangle}
\def\<{\langle}

\def\poly{{\rm poly}}

\def\A{\mathcal{A}}
\def\B{\mathcal{B}}

\def\E{\mathbb{E}}
\def\H{\mathcal{H}}
\def\L{\mathcal{L}}
\def\O{\mathcal{O}}

\def\Q{\mathcal{Q}}
\def\R{\mathbb{R}}
\def\S{\mathcal{S}}
\def\Tr{{\rm Tr}}
\def\0{\bm{0}}
\def\1{\bm{1}}
\def\2{\bm{2}}

\def\bmt{\bm{\theta}}
\def\l\|{\left\|}
\def\r\|{\right\|}
\def\lmin{\lambda_{\min}}
\def\lmax{\lambda_{\max}}

\allowbreak

\usepackage{pdfpages}
\usepackage{pgffor}
\makeatletter
\AtBeginDocument{\let\LS@rot\@undefined}
\makeatother

\begin{document}

\title{The curse of random quantum data}

\author{Kaining Zhang}
\email{kzha3670@uni.sydney.edu.au}
\affiliation{School of Computer Science, Faculty of Engineering, University of Sydney, Australia}

\author{Junyu Liu}
\email{junyuliucaltech@gmail.com}
\affiliation{Pritzker School of Molecular Engineering, The University of Chicago, Chicago, IL 60637, USA}
\affiliation{Department of Computer Science, The University of Chicago, Chicago, IL 60637, USA}
\affiliation{Kadanoff Center for Theoretical Physics, The University of Chicago, Chicago, IL 60637, USA}
\affiliation{Department of Computer Science, The University of Pittsburgh, Pittsburgh, PA 15260, USA}

\author{Liu Liu}
\email{liuliubh@gmail.com}
\affiliation{School of Computer Science, Faculty of Engineering, University of Sydney, Australia}

\author{Liang Jiang}
\email{liang.jiang@uchicago.edu}
\affiliation{Pritzker School of Molecular Engineering, The University of Chicago, Chicago, IL 60637, USA}

\author{Min-Hsiu Hsieh}
\email{min-hsiu.hsieh@foxconn.com}
\affiliation{Hon Hai Quantum Computing Research Center, Taipei, Taiwan}

\author{Dacheng Tao}
\email{dacheng.tao@sydney.edu.au}
\affiliation{School of Computer Science, Faculty of Engineering, University of Sydney, Australia}

\date{\today}

\begin{abstract}
Quantum machine learning, which involves running machine learning algorithms on quantum devices, may be one of the most significant flagship applications for these devices. Unlike its classical counterparts, the role of data in quantum machine learning has not been fully understood. In this work, we quantify the performances of quantum machine learning in the landscape of quantum data. Provided that the encoding of quantum data is sufficiently random, the performance, we find that the training efficiency and generalization capabilities in quantum machine learning will be exponentially suppressed with the increase in the number of qubits, which we call ``the curse of random quantum data". Our findings apply to both the quantum kernel method and the large-width limit of quantum neural networks. Conversely, we highlight that through meticulous design of quantum datasets, it is possible to avoid these curses, thereby achieving efficient convergence and robust generalization. Our conclusions are corroborated by extensive numerical simulations.
\end{abstract}


\maketitle

\section{Introduction}

The attainment of quantum advantage~\cite{nature_arute2019quantum, science_abe8770} using noisy intermediate-scale quantum (NISQ) devices~\cite{quantum_Preskill2018} has significantly advanced research in various fields where quantum computing shows promise.
Among these, the variational quantum algorithm~(VQA)~\cite{nat_review_phys_cerezo2021} emerges as a prominent framework for harnessing NISQ devices, owing to its relatively modest requirements on gate noise and circuit connectivity. VQAs have been applied in diverse areas, including machine learning (so-called quantum machine learning)~\cite{pra_schuld2020circuit, nature_havlivcek2019sl, prl_schuld2019qml, prr_yuxuan2020powerpqc, ieee_samuel2020rl, nature_saggio2021rl, pra_heliang2021experimentqgan, prxq_yuxuan2021learnability}, numerical analysis~\cite{pra_lubasch2020vqanonlinear, pra_kubo2021vqasde, prxq_yongxin2021vqd, pra_hailing2021vqapoisson, pra_kyriienko2021nonlinearde}, quantum simulation~\cite{rmp_georgescu2014qs, quantum_yuan2019theoryofvariational, npj_mcardle2019variational, prl_endo2020_vqsgp, nature_neill2021qring, nature_mi2021time, science_randall2021mbltime, science_amita2021obs, science_semeghini2021topoliquid, science_satzinger2021topoorder}, and quantum chemistry~\cite{rmp_mcardle2020qcc, nc_peruzzo2014var, nature_kandala2017hardware, prx_hempel2018qcion, quantum_Higgott2019vqc, nc_grimsley2019adaptive, science_frank2020hf, prxq_tang2021adaptvqe, pra_delgado2021vqa}. The concepts of variational quantum circuits have also been utilized for practical demonstrations of quantum control and quantum error correction, thereby bridging current developments in noisy quantum computing towards future fault-tolerant quantum computing \cite{ni2023beating,sivak2023real}.

On the other hand, despite the existence of practical demonstrations of variational quantum algorithms, their theoretical foundation has not been completely established. People still lack confidence regarding when and how variational quantum circuits could be utilized, how to design them from first principles, and where theoretical quantum advantage could potentially exist in a generic setup \cite{liu2021representation,arxiv_liu2022analytic,liu2022laziness,arxiv_zhang2023dynamical,wang2023symmetric,li2023d}. An existing theoretical proposal, known as the quantum neural tangent kernel (QNTK), has proven helpful in understanding the dynamics of gradient descent in quantum machine learning, particularly in the realm of wide quantum neural networks where the number of training parameters and the dimension of the Hilbert space are considered large \cite{arxiv_liu2022analytic}. However, in scenarios where the number of training parameters is not comparable to the dimension of Hilbert spaces, the QNTK might become small, and training efficiency may be unsatisfactory, leading to what is termed the barren plateau problem \cite{nc_mcclean2018barren}. Conversely, a large number of training parameters will linearize the quantum neural network towards the frozen kernel limit, where an emergent fixed kernel could be used akin to linear regression. This limit linearizes the dynamics and obscures information about the quantum data. The effects of data, which appear to be significant in practical training performances, become less clear in the limit of a frozen kernel. Hence, an important question naturally arises: how are quantum data and its encoding related to the training efficiency and practical performance of quantum machine learning?

In this paper, we uncover a novel phenomenon related to the interplay between quantum machine learning performance and the distribution of quantum data, which we term \emph{the curse of random quantum data}. We discover that, similar to the barren plateau phenomenon in the gradient efficiency for training dynamics, across a typical class of quantum machine learning algorithms, when quantum data is encoded randomly, the algorithm's learning performance markedly declines. Specifically, the spectrum of the quantum neural tangent kernel diminishes in line with the dimension of the Hilbert space, and similarly, the reduction in the generalization error during training is also constrained by the Hilbert space's dimension. This phenomenon may not be restricted to any particular training strategy; it is observable in both the quantum kernel method and the variational quantum algorithm at or beyond the frozen kernel limit. Hence, we identify this as a new type of the curse of large dimensions of Hilbert spaces occurring within the realm of quantum data rather than conventional barren plateau statements in the conventional parameter space associated with the loss function landscape \cite{nc_mcclean2018barren,nc_cerezo2020cost, iop_2021Uvarovlocality, prx_pesah2020absence,qmi_skolik2020layerwise,arxiv_zhang2021towardtqnn,zhang2022gaussian}. Importantly, this curse can emerge irrespective of the training method employed, indicating that it is a fundamental aspect of quantum data space, affecting both kernel methods and gradient descent approaches alike. Furthermore, we discover that a carefully curated data distribution, as opposed to one that is uniformly random across the Hilbert space, can significantly mitigate the curse. This is demonstrated through direct examples where certain input state datasets and encoding schemes markedly enhance the performance of quantum machine learning algorithms, particularly affecting the spectra of the QNTK. Our theoretical examples are substantiated by numerical simulations of quantum machine learning on tasks including the quantum dynamics learning and the binary classification.

\section{Theoretical results}

\subsection{Background of quantum machine learning}

We begin by setting up notations and methods for quantum machine learning. The training set is denoted by $\A=\{a\}=\{(\rho_a, y_a)\}$, where each sample $a$ consists of the input state in the density matrix form $\rho_a$ and the label $y_a$. The test set is denoted by $\B=\{\hat{b}\}$. In the task of QML, we aim to approximate the label $y$ with the parameterized prediction function $z(\bmt)$. Here a common approach is to minimize the mean squared error on the training set as the loss function:
\begin{equation}\label{vqacr_qntk_loss_eq}
\L_{\A}(\bmt) = \frac{1}{2|\A|} \sum_{a\in\A} \left( z_a(\bmt) - y_a \right)^2 = \frac{1}{2|\A|} \left\| \bm{r}_{\A} (\bmt) \right\|_2^2,
\end{equation}
where 
\begin{equation}\label{vqacr_qntk_loss_i_eq}
\bm{r}_{\A}(\bmt):=\{r_a(\bmt)\}_{a\in\A} := \{z_a(\bmt) - y_a\}_{a\in\A}
\end{equation}
is the training residual vector. The predication function is given by
\begin{equation}\label{vqacr_main_backqml_za}
z_a(\bmt) = \Tr [O(\bmt) \rho_a].
\end{equation}

In this paper, we focus on two approaches for generating the predication function. The first is the quantum kernel method:
\begin{equation}\label{vqacr_main_backqml_qkm_z}
O_{\rm qkm}(\bmt) = \sum_{a\in\A} \theta_a \rho_a .
\end{equation}
The second approach is the quantum neural network using the observable $O$ and the variational quantum circuit $V(\bmt)$:
\begin{equation}\label{vqacr_main_backqml_qnn_o}
O_{\rm qnn}(\bmt) = {V(\bmt)}^\dag O V(\bmt) .
\end{equation}
For the QNN case, we train the parameter using gradient descent with some learning rate $\eta$, i.e.
\begin{equation}\label{vqacr_main_backqml_gd_theta}
\bmt(t+1) = \bmt(t) - \eta \nabla \L_{\A} (\bmt(t))	.
\end{equation}

Suppose the parameter $\bmt \in \R^D$. Let $\eta \rightarrow 0$, the gradient descent dynamics approximates to the continuous regime known as the gradient flow:
\begin{equation}\label{vqacr_qntk_gd_theta_eq}
\dot{\bmt} = - \nabla \L_{\A}(\bmt) = - \frac{1}{|\A|} J(\bmt) \bm{r}_{\A}(\bmt),
\end{equation}
where $J(\bm{\theta}) \in \R^{D\times |\A|}$ is the Jacobian matrix with entries $J_{d a}(\bm{\theta})=\frac{\partial z_a}{\partial \theta_d}(\bm{\theta})$.
We are interested in the convergence rate of the loss function, i.e.
\begin{align}
\dot{\L}_{\A}(\bmt) ={}&  \nabla \L_{\A}(\bmt)^T \dot{\bmt}  \notag \\ 
={}& - \frac{D}{|\A|^2} \bm{r}_{\A}(\bmt)^T K(\bmt) \bm{r}_{\A}(\bmt) \notag \\ 
\leq{}&  - \frac{2D}{|\A|} \lambda_{\min} \left[ K(\bmt) \right] {\L}_{\A}(\bmt), \label{vqacr_qntk_loss_descent_simple_4}
\end{align}
where
\begin{equation}\label{vqacr_qntk_eq}
K(\bm{\theta}) = \frac{1}{D} J(\bm{\theta})^T J(\bm{\theta})
\end{equation}
is the QNTK at the point $\bm{\theta}$. 
Eq.~(\ref{vqacr_qntk_loss_descent_simple_4}) shows that the loss decays linearly when the least eigenvalue of QNTK is lower bounded by some positive constant. In the rest of this paper, we will prove rigorous convergence guarantees based on further analyses to the least eigenvalue of QNTK.

\subsection{Generalization error with quantum data}

Since input states are engaged in the training of QML via quantum kernel or QNTK, their distributions could naturally affect training and generalization performances. Here, we provide some negative examples for datasets sampled from state 2-designs. Specifically, we consider the quantum kernel method and the quantum neural network approach separately in Theorems~\ref{vqacr_geb_qkm_haar_input_theorem_main} and \ref{vqacr_geb_qntk_haar_input_theorem_main}. Related proofs can be found in Appendix~\ref{vqacr_geb_qkm_haar_input}.
The generalization error has up to exponentially small improvement from the extremely overfitting situation, when the size of datasets is relatively small compared with the Hilbert dimension. 

\begin{theorem}\label{vqacr_geb_qkm_haar_input_theorem_main}
Suppose all $N$-qubit quantum states in the training and test datasets $\A$ and $\B$ are independently sampled from state 2-designs and the size of datasets is smaller than $2^{N/2}$. Let $\L_{\A}(t)$ and $\L_{\B}(t)$ be the training and the test loss function of the quantum kernel method (Eqs.~(\ref{vqacr_main_backqml_za}) and (\ref{vqacr_main_backqml_qkm_z})), respectively. Then, with high probability,
\begin{equation}
\mathop{\E} \L_{\B}(\infty) \gtrsim{} \mathop{\E} \mathcal{L}_{\B} (0) - \frac{|\A|}{2^{N-1}} \mathop{\E} \sqrt{ \mathcal{L}_{\A}(0) \mathcal{L}_{\B}(0) } ,
\end{equation}
where the expectation is taken under state $2$-designs for states in $\A$ and $\B$.
\end{theorem}

\begin{theorem}\label{vqacr_geb_qntk_haar_input_theorem_main}
Suppose all $N$-qubit quantum states in the training and test datasets $\A$ and $\B$ are independently sampled from state 2-designs. Let $\L_{\A}(t)$ and $\L_{\B}(t)$ be the training and the test loss function of the quantum neural network (Eqs.~(\ref{vqacr_main_backqml_za}) and (\ref{vqacr_main_backqml_qnn_o})), where the label is given as $y:=\Tr[O U \rho U^\dag]$ for a target unitary $U$ with the zero trace observable $O$. Then, under the frozen QNTK regime,
\begin{equation}\label{vqacr_geb_qntk_haar_input_theorem_eq_main}
\mathop{\E} \L_{\B}(\infty) \geq{} \left( 1 - \frac{|\A|}{2^{2N}-1} \right) \mathop{\E} \mathcal{L}_{\B} (0) ,
\end{equation}
where the expectation is taken under $2$-designs distributions for states in $\A$ and $\B$ and the target unitary $U$.
\end{theorem}

We remark that Theorem~\ref{vqacr_geb_qntk_haar_input_theorem_main} provides a novel result about learning quantum dynamics from the no-free-lunch framework perspective. 
Compared with existing results~\cite{Wang2024} that mainly focus on the lower bound with an absolute formulation, we provide a bound with a relative formulation for the test loss with trained parameters. As shown in Eq.~(\ref{vqacr_geb_qntk_haar_input_theorem_eq_main}), the bound is in terms of the initial test loss times a constant, which is exponentially close to $1$ with polynomially scaled training datasets. Therefore, the loss on the test set with $2$-design states only has exponentially small improvements from its initial value, unless the size of the training set exceeds an exponentially large threshold. However, training QNN with exponentially large datasets requires tremendous computational resources. Such kind of datasets are also unavailable in practice for large-scale problems.

Theorems~\ref{vqacr_geb_qkm_haar_input_theorem_main} and \ref{vqacr_geb_qntk_haar_input_theorem_main} reveal one critical issue of quantum machine learning from the data perspective, which we name as \textit{the curse of random quantum data}. For input states distributed uniformly in the Hilbert space, the corresponding QML task exhibits poor generalization performance. This phenomenon is universal and is independent of certain QML method employed. One intuitive explanation of the curse of random quantum data is that uniformly distributed quantum states have inherently bad learnability, unless the dataset has exponentially large size. Specifically, quantum representation learning is hard for datasets composed of nearly orthogonal states without prior knowledge. However, practical QML datasets do not grow exponentially with increasing qubits in general. Therefore the distribution of input states given as quantum data or quantum encoding would play an essential role, which may rule out the potential of benign generalization. 

\begin{figure*}[t]
\centering
\subfloat[]{
\includegraphics[width=.48\linewidth]{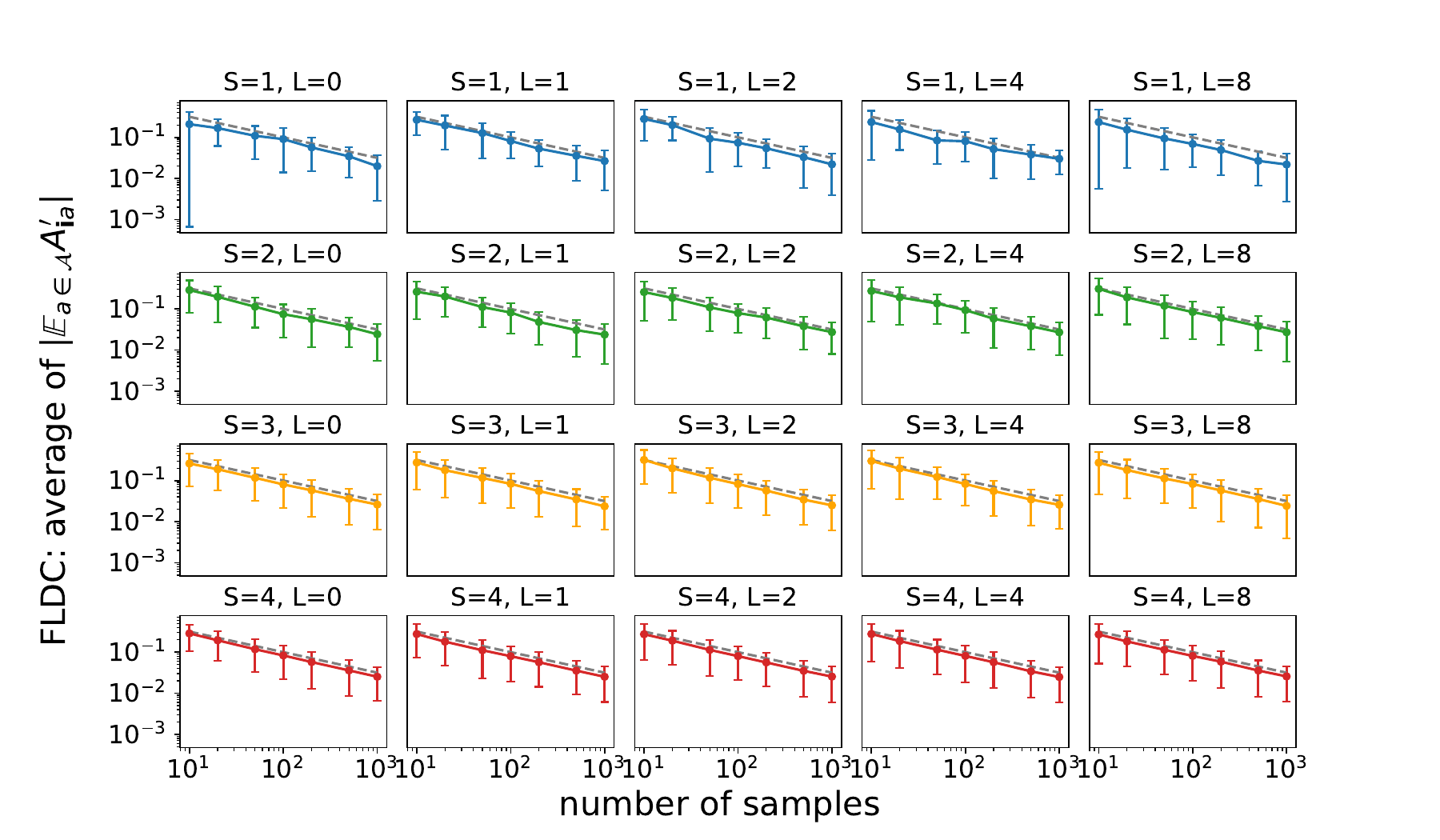}  
\label{vqacr_fig_Aai_fldcs_mean0}
}
\subfloat[]{
\includegraphics[width=.48\linewidth]{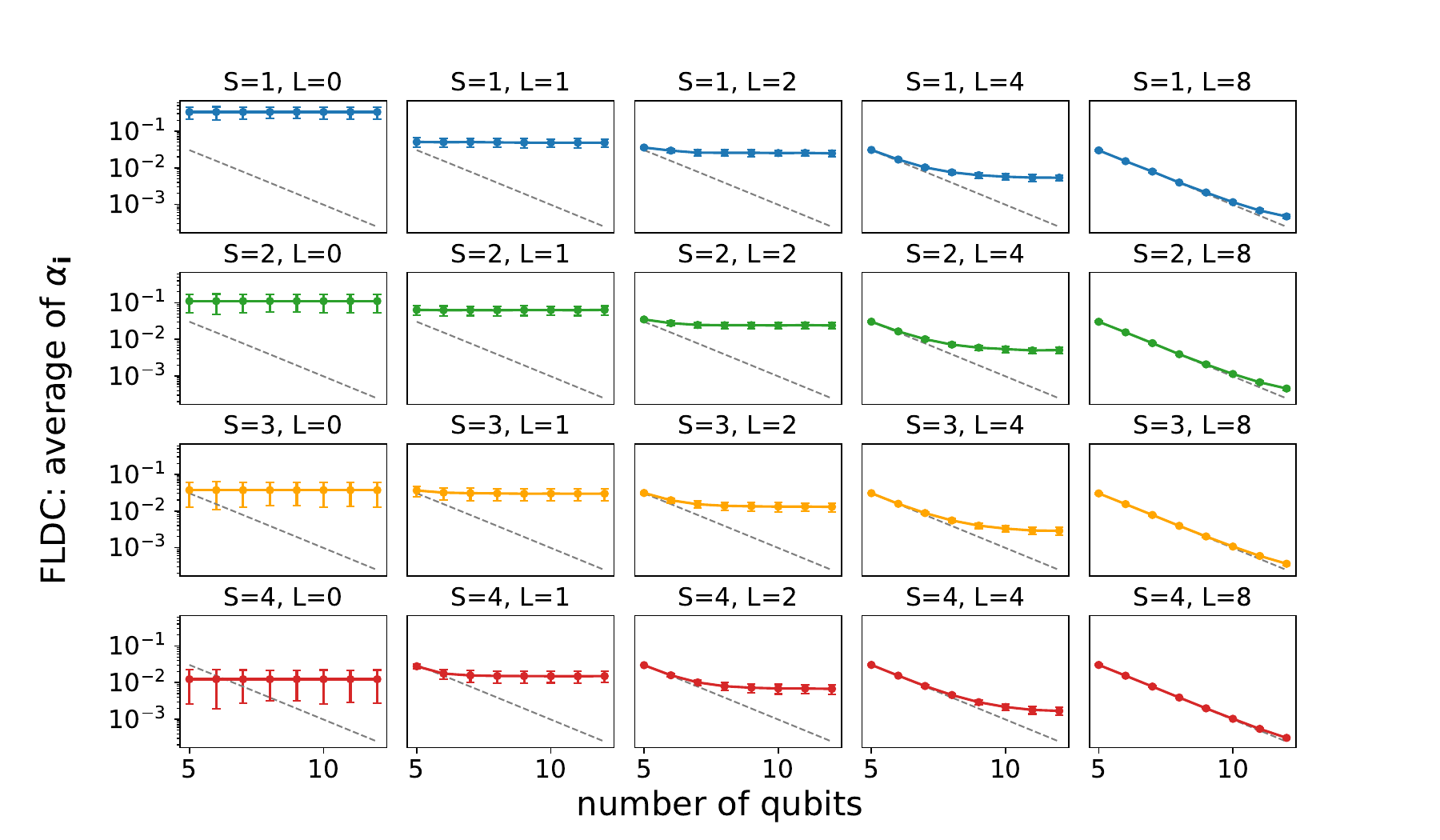}  
\label{vqacr_fig_Aai_fldcs_alpha}
}
\caption{Scaling of Pauli decomposition coefficients of finite local-depth circuit states over $S\in \{1,2,3,4\}$ and circuit layer $L\in \{0,1,2,4,8\}$. Figure~\ref{vqacr_fig_Aai_fldcs_mean0} shows the average of  $ |\E_{a\in\A} A'_{\bm{i}a} | $ in Eq.~(\ref{vqacr_main_Ttest_mean}) w.r.t. all $\|\bm{i}\|_1=S$ with 1D connectivity, where $|\A| \in \{10, 20, 50, 100, 20, 500, 1000\}$. The grey dashed line plots $|\A|^{-0.5}$. Figure~\ref{vqacr_fig_Aai_fldcs_alpha} shows the average of $\alpha_{\bm{i}}$ in Proposition~\ref{vqacr_prop_state_distribution_local} w.r.t. all $\|\bm{i}\|_1=S$ with 1D connectivity, where $N \in \{5, 6, 7, 8, 9, 10, 11, 12\}$, and each $\alpha_{\bm{i}}$ is calculated by taking the average over $|\A|=1000$ samples. The grey dashed line plots the Haar limit. The error bar in all figures shows the standard deviation over variables with different $\bm{i}$. }
\label{vqacr_fig_Aai_fldcs}
\end{figure*}

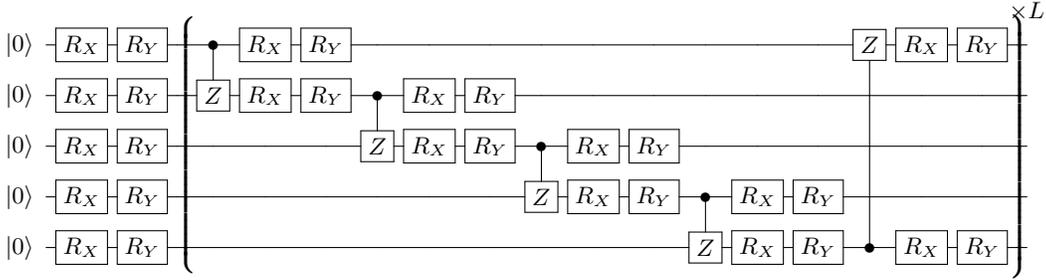
\begin{figure*}
\centerline{
\Qcircuit @C=0.4em @R=0.7em {
& & & & & & & & & & & & & & & & & & & & & \times L \\
\lstick{|0\>} & \gate{R_X} & \gate{R_Y} & \qw & \qw & \ctrl{1} & \gate{R_X} & \gate{R_Y} & \qw & \qw & \qw & \qw & \qw & \qw & \qw & \qw & \qw & \gate{Z} & \gate{R_X} & \gate{R_Y} & \qw & \qw \\
\lstick{|0\>} & \gate{R_X} & \gate{R_Y} & \qw & \qw & \gate{Z} & \gate{R_X} & \gate{R_Y} & \ctrl{1} & \gate{R_X} & \gate{R_Y} & \qw & \qw & \qw & \qw & \qw & \qw & \qw & \qw & \qw & \qw & \qw \\
\lstick{|0\>} & \gate{R_X} & \gate{R_Y} & \qw & \qw & \qw & \qw & \qw & \gate{Z} & \gate{R_X} & \gate{R_Y} & \ctrl{1} & \gate{R_X} & \gate{R_Y} & \qw & \qw & \qw & \qw & \qw & \qw & \qw & \qw \\
\lstick{|0\>} & \gate{R_X} & \gate{R_Y} & \qw & \qw & \qw & \qw & \qw & \qw & \qw & \qw & \gate{Z} & \gate{R_X} & \gate{R_Y} & \ctrl{1} & \gate{R_X} & \gate{R_Y} & \qw & \qw & \qw & \qw & \qw \\
\lstick{|0\>} & \gate{R_X} & \gate{R_Y} & \qw & \qw & \qw & \qw & \qw & \qw & \qw & \qw & \qw & \qw & \qw & \gate{Z} & \gate{R_X} & \gate{R_Y} & \ctrl{-4} & \gate{R_X} & \gate{R_Y} & \qw & \qw 
\gategroup{2}{6}{6}{6}{2.0em}{(}
\gategroup{2}{20}{6}{20}{1.0em}{)} 
}
}
\caption{Finite local-depth circuits with $L$ blocks for $N=5$.}
\label{vqacr_fldcs_example}
\end{figure*}

\subsection{Spectrum of QNTK with quantum data}

We have shown the \textit{learnability} issue from uniformly random quantum data, i.e., datasets with $2$-design states have bad generalization performances. Here we present another \textit{trainability} issue induced by uniformly random quantum data from the optimization perspective in QNN. 

\begin{theorem}\label{vqacr_lemma_qntk_spectrum_qntk_haar_main}
Let $\A\{(\rho_a, y_a)\}$ be the training datasets with independent $N$-qubit $2$-design states $\rho_a$. Denote by $z_a=\Tr[OV(\bmt)\rho_a V(\bmt)^\dag]$ the QNN with the observable $O$ and the variational circuit $V(\bmt)=\prod_{d=D}^{1} \exp[-iH_d \theta_d/2] W_d$, where the hamiltonian $H_d$ is Hermitian unitary and $W_d$ is the fixed unitary. Then there exists a constant $C>0$ such that
\begin{equation} \label{vqacr_lemma_qntk_spectrum_qntk_haar_main_eq}
\lmax \left[ K (\bmt) \right] \leq{} \frac{\|O\|_F^2}{4^{N}} \left( 1 + \frac{C|\A|}{2^N \delta} \sqrt{\ln |\A|} \right)
\end{equation}
with probability at least $1-\delta$.
\end{theorem}

As shown in Theorem~\ref{vqacr_lemma_qntk_spectrum_qntk_haar_main}, the QNTK has exponentially small spectrum with increased qubits $N$. Since the convergence rate is controlled by the spectrum of QNTK, we should avoid using uniformly distributed input states in the training of QNN. One appropriate substitute is a biased state distribution. For convenience, we first show the state $2$-design distribution in the form of Pauli decomposition. Let $A\in\R^{4^N\times|\A|}$ be the matrix that stores the Pauli decomposition coefficients of states in $\A$, i.e.
\begin{align}
\rho_a = \sum_{\bm{i}\in\{0,1,2,3\}^{\otimes N}} \frac{1}{2^N} A_{\bm{i}a} \sigma_{\bm{i}}, \label{vqacr_main_pauli_decomposition_rho_a}
\end{align}
where the basis $\sigma_{\bm{i}}:=\sigma_{i_0} \otimes \cdots\otimes \sigma_{i_N}$ is the tensor product of Pauli matrices. When each $\rho_a$ is sampled independently from a $2$-design distribution,
\begin{align}
\mathop{\E}_{\text{2-design}} A_{\bm{i} a} ={}& 0 , \label{vqacr_main_EAA_2design_eq0} \\
\mathop{\E}_{\text{2-design}} A_{\bm{i} a} A_{\bm{j} b} ={}& \frac{1}{2^N+1} \delta_{\bm{i}\bm{j}} \delta_{ab} . \label{vqacr_main_EAA_2design_eq1}
\end{align}
for all $\bm{i},\bm{j} \neq \bm{0}$. 
We recap that the Pauli coefficient coefficients of $2$-designs states has the following two properties: 1) zero mean and covariance; 2) unbiased variance with the value $(2^N+1)^{-1}$. Therefore, each element of $A$ except $A_{\bm{0}a}$ has a small size around $\O(2^{-N/2})$, which technically contributes to exponentially small generalization improvement and QNTK. To overcome this issue, we introduce the state distribution in Proposition~\ref{vqacr_prop_state_distribution_local}. The proposed distribution is generalized from Eqs.~(\ref{vqacr_main_EAA_2design_eq0}) and (\ref{vqacr_main_EAA_2design_eq1}), such that coefficients have biased variances on different Pauli bases.

\begin{proposition}
\label{vqacr_prop_state_distribution_local}
States $\rho_a$ in the dataset $\A$ are generated independently from the same distribution, such that for all $\bm{i}, \bm{j} \neq \bm{0}$, 
\begin{align}
\E A_{\bm{i} a} ={}& 0 , \notag \\
\E A_{\bm{i} a} A_{\bm{j} b} ={}& \alpha_{\bm{i}} \delta_{\bm{i}\bm{j}} \delta_{ab} . \notag
\end{align} 
\end{proposition}

Here we present an example that lead to the coefficient distribution in Proposition~\ref{vqacr_prop_state_distribution_local}. We consider states from finite local-depth circuits (FLDC)~\cite{PhysRevLett.132.150603}, where the parameter in circuits are independently sampled from the uniform distribution on $[0,2\pi]$. Numerical evidence is provided in Figure~\ref{vqacr_fig_Aai_fldcs}. 
Denote by $\mu_{\bm{i}}$ the expectation of each Pauli decomposition coefficient $A_{\bm{i}a}$ over the distribution of $a$. Firstly, we demonstrate that the expectation $\mu_{\bm{i}}$ is zero. 
Due to the central limit theorem, the statistical term 
\begin{align}
\left| \mathop{\E}_{a \in \A} A'_{\bm{i}a} \right| :={}& \left| \frac{ \mathop{\E}_{a \in \A} A_{\bm{i}a}}{ \mathop{\rm SD}_{a \in \A} A_{\bm{i}a} } \right|  \lesssim |\mu_{\bm{i}}| + \frac{1}{\sqrt{|\A|}} |g| , \label{vqacr_main_Ttest_mean}
\end{align}
where $g$ is a standard normal distributed variable. As illustrated in Figure~\ref{vqacr_fig_Aai_fldcs_mean0}, the average of the statistical term in Eq.~(\ref{vqacr_main_Ttest_mean}) decays inversely with the square root of $|\A|$ when it grows. Therefore each $\mu_{\bm{i}}=0$ for $S \in \{1, 2, 3, 4\}$.
Next, we reveal the scaling of the variance $\alpha_{\bm{i}}$ in terms of the qubit number and the block layer. As shown in Figure~\ref{vqacr_fig_Aai_fldcs_alpha}, when the qubit number increases, 
the variance remains large when the number of circuit layer is small. With growing circuit layers, the variance converges to the Haar limit $1/(2^N+1)$ for $S\in \{1,2,3,4\}$. In summary, input states prepared from FLDC with finite circuit layers are expected to have large Pauli decomposition coefficients on local bases with small $S$.

Datasets with the property in Proposition~\ref{vqacr_prop_state_distribution_local} have the potential to induce large QNTK. Here we provide one example in Lemma~\ref{vqacr_lemma_qntk_least_eigen_gaussian_matrix_main}. The least eigenvalue of the QNTK at $\bmt=\bm{0}$ has a lower bound that is proportional to the least variance of coefficients corresponding to quantum gates qubits. Therefore, we expect that an appropriate state distribution should have variances at least $\O(1/\poly(N))$ for local Pauli bases.

\begin{lemma}
\label{vqacr_lemma_qntk_least_eigen_gaussian_matrix_main}
Let the observable be $O=\sum_{k=1}^{N} o_k Z_k$, where each $Z_k$ is the product of Pauli matrices that is $Z$ on the $k$-th qubit and $I$ on other qubits. Denote by $H_d$ the Hamiltonian of the $d$-th quantum gate in the circuit $V(\bmt)=\prod_{d=D}^{1} \exp[-iH_d \theta_d/2]$ that acts on qubits $\S_d \subseteq \{1,2,\cdots,N\}$. We assume that $H_d$ is randomly sampled from the set $\H:=\{X,Y,Z\}^{\otimes S}$. Let $\Q:=\{\S_d\}$. Let $\rho_{a}=\sum_{\bm{p} \in \{0,1,2,3\}^{\otimes N}} \frac{1}{2^N} A_{\bm{p}a} \sigma_{\bm{p}}$ be the Pauli decomposition of state $\rho_a$, where $A$ follows Proposition~\ref{vqacr_prop_state_distribution_local}. Denote by $A_{\Q}$ the matrix with dimension $((3^S-1)|\Q|, |\A|)$ formed by vertically concatenating matrices $A_{\Q,w}$ for $w \in \{0,1,\cdots,S-1\}$. 
Denote by ${\rm Var}_{\Q} := {\rm diag}(\alpha_{\bm{i}})$, $\alpha_{\max}:=\max_{\bm{i}} \alpha_{\bm{i}}$, and $\alpha_{\min}:=\min_{\bm{i}} \alpha_{\bm{i}}$, where $\bm{i}$ corresponds to the index of rows of $A_{\Q}$. Denote by $A'_{\Q} := {\rm Var}_{\Q}^{-1/2} A_{\Q}$, $ A'_{\Q,\max}:=\frac{1 }{\sqrt{(3^S-1)|\Q|}} \max_{a\in\A} \left\| [A'_{\Q}]_{\cdot,a} \right\|_2 $ and $ A'_{\Q,\min}:=\frac{1 }{\sqrt{(3^S-1)|\Q|}} \min_{a\in\A} \left\| [A'_{\Q}]_{\cdot,a} \right\|_2 $.
Then when $|\A|$ is sufficiently large, there exists two constants $C_1, C_2 >0$ such that for any $\delta, k \in (0,1)$, the following holds with probability at least $1-\delta$
\begin{align}
(1-k) \Delta_S \alpha_{\min} A_{\Q, \min}'^2 \leq{}& \lambda \left[K(\bm{0}) \right] \notag \\
\leq{}& (1+k) \|O\|_2^2 \alpha_{\max} A_{\Q, \max}'^2 , \label{vqacr_lemma_qntk_least_eigen_gaussian_matrix_main_eq_1}
\end{align}
for $D\geq D_0 \ln \frac{D_0}{\delta}$ with $D_0 = C_2 \frac{ \|O\|_2^4 \alpha_{\max}^2 A_{\Q, \max}'^2 }{k^2 \Delta_S^2 \alpha_{\min}^2 A_{\Q, \min}'^4} |\A| \ln \frac{|\A|}{\delta} $ and $S \geq{} \log_3 \left[ 1 + \frac{C_1^2  |\A|^2 \ln |A| }{ k^2 A_{\Q, \min}'^4 |\Q| \delta^2 } \right]$, 
where $\Delta_{S} := \min_{\bm{g}\in \{0,\pm 1\}^{\otimes N}, 1 \leq \|\bm{g}\|_1 \leq S} (\bm{g}^T \bm{o} )^2 $.
\end{lemma}

We remark that Proposition~\ref{vqacr_prop_state_distribution_local} could be applied for designing quantum encoding strategies. One example that Proposition~\ref{vqacr_prop_state_distribution_local} holds is the qubit embedding~\cite{lloyd2020quantum} for random vectors. Here, we encode the vector $\bm{x} \in \R^N$ into the state 
\begin{equation}\label{vqacr_qubit_embedding_eq}
|\psi(\bm{x})\>	 = \bigotimes_{n=1}^{N} e^{-iZ\pi x_n/2} \frac{|0\>+|1\>}{\sqrt{2}} 
\end{equation}
and each entry of $\bm{x}$ is normalized in $[-1,1]$. 
An example of QNNs with qubit embedding is presented in Lemma~\ref{vqacr_lemma_qntk_least_eigen_random_sphere_main}, where the least eigenvalue of the QNTK at $\bmt=\bm{0}$ has the lower bound around $\O(2^{-S})$.   

\begin{lemma}\label{vqacr_lemma_qntk_least_eigen_random_sphere_main}
We follow conditions and notations in Lemma~\ref{vqacr_lemma_qntk_least_eigen_gaussian_matrix_main} except the Hamiltonian set $\H := \{X,Y\}^{\otimes S}$. Let $\rho_a$ be the density matrix of the qubit embedding of the vector $\bm{x}_{a}$ via Eq.~(\ref{vqacr_qubit_embedding_eq}), where each entry of $\bm{x}_{a}$ is an independent random variable distributed uniformly in $[-1,1]$. Then when $|\A|$ is sufficiently large, there exists two constants $C_1, C_2 >0$ such that for any $\delta, k \in (0,1)$, the following holds with probability at least $1-\delta$
\begin{equation}\label{vqacr_lemma_qntk_least_eigen_random_sphere_main_eq}
(1-k)\frac{\Delta_S}{ 2^S} \leq{} \lambda \left[K(\bm{0}) \right] \leq{} (1+k)\frac{\|O\|_2^2}{2^S},
\end{equation}
for $D\geq D_0 \ln \frac{D_0}{\delta}$ with $D_0 = C_2 \frac{|\A| \|O\|_2^4 }{k^2 \Delta_S^2 } \ln \frac{|\A|}{\delta} $ and $S \geq 2+ \log_2 \frac{C_1^2 |\A|^2 \ln |\A|}{k^2 |\Q|\delta^2}$, 
where $\Delta_{S} := \min_{\bm{g}\in \{0,\pm 1\}^{\otimes N}, 1 \leq \|\bm{g}\|_1 \leq S} (\bm{g}^T \bm{o} )^2 $.
\end{lemma}

Both bounds in Lemma~\ref{vqacr_lemma_qntk_least_eigen_gaussian_matrix_main} and Lemma~\ref{vqacr_lemma_qntk_least_eigen_random_sphere_main} holds with mild conditions on the number of gates in the circuit and the number of qubits involved for each gate, i.e. $D=\tilde{\O}(|\A| \ln^2 |\A|)$ and $S=\tilde{\O}(\ln |\A|)$. 
Derivations of Theorem~\ref{vqacr_lemma_qntk_spectrum_qntk_haar_main} and Theorems~\ref{vqacr_lemma_qntk_least_eigen_gaussian_matrix_main} and \ref{vqacr_lemma_qntk_least_eigen_random_sphere_main} are provided in Appendix~\ref{vqacr_qntk_app_lmin_qntk}.

\begin{figure*}[t]
\centering
\subfloat[]{
\includegraphics[width=.48\linewidth]{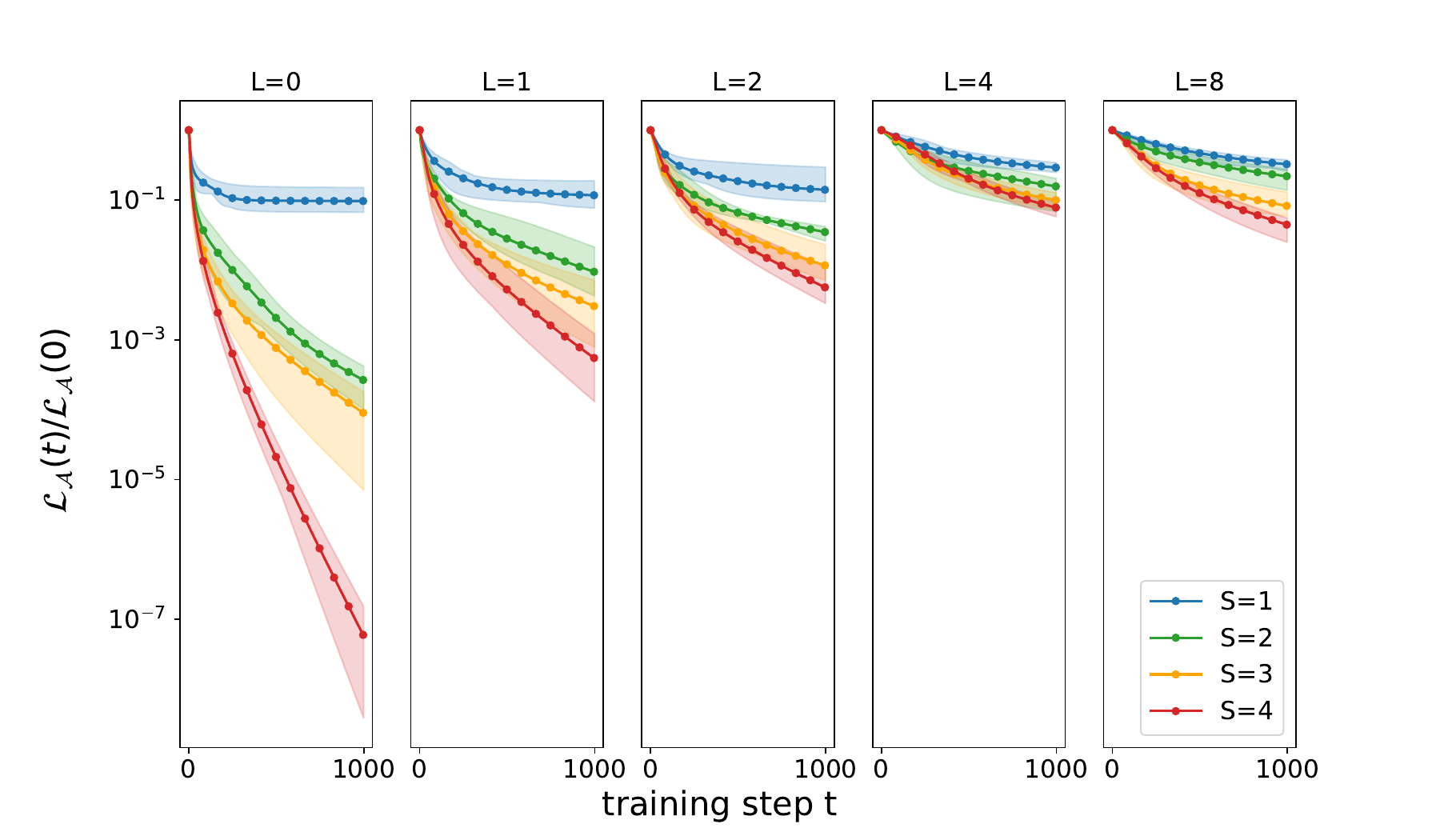}  
\label{vqacr_fig_qdl_fldc_LA}
}
\subfloat[]{
\includegraphics[width=.48\linewidth]{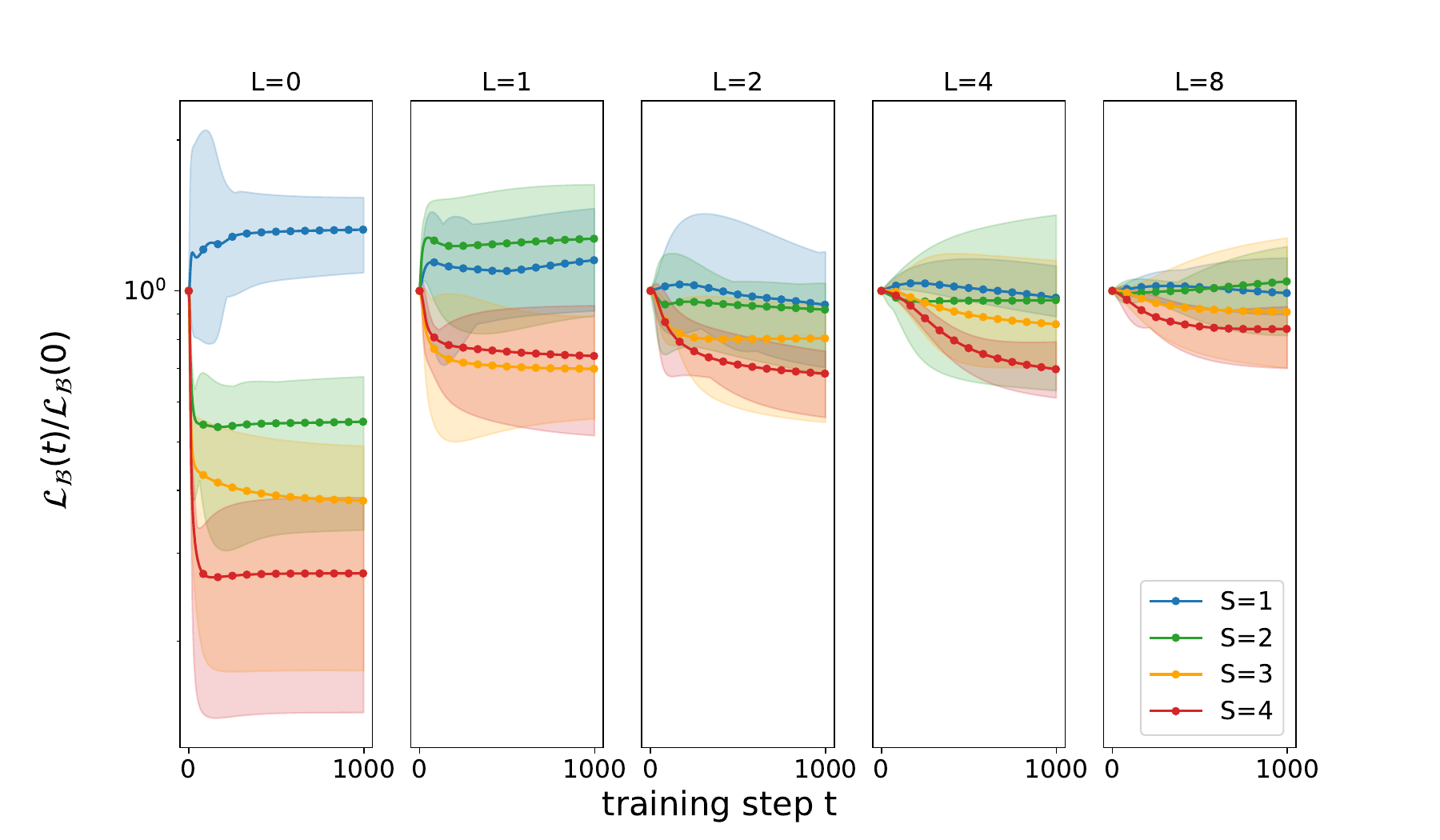}  
\label{vqacr_fig_qdl_fldc_LB}
}\\
\subfloat[]{
\includegraphics[width=.48\linewidth]{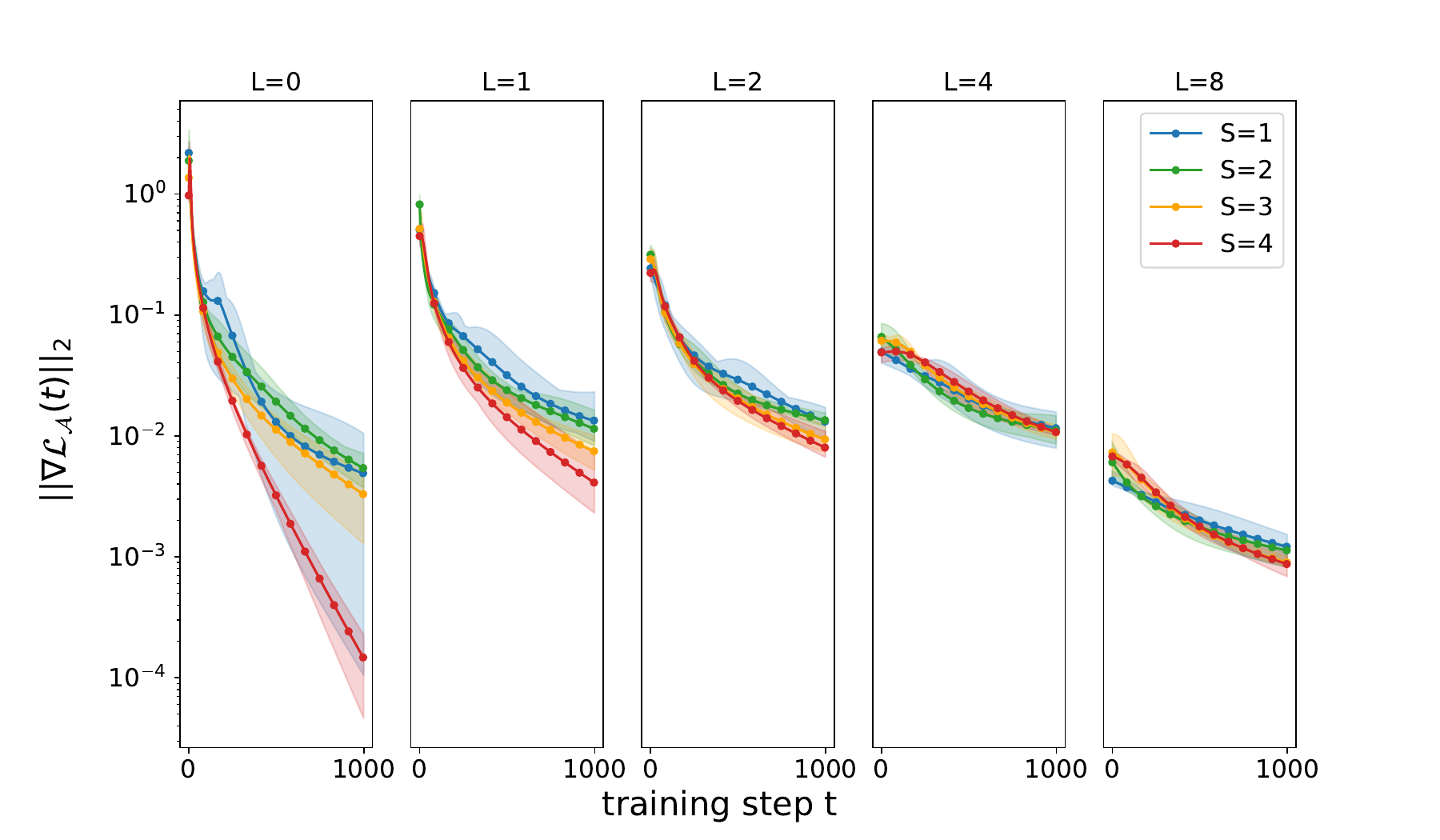}  
\label{vqacr_fig_qdl_fldc_gradnorm}
}
\subfloat[]{
\includegraphics[width=.48\linewidth]{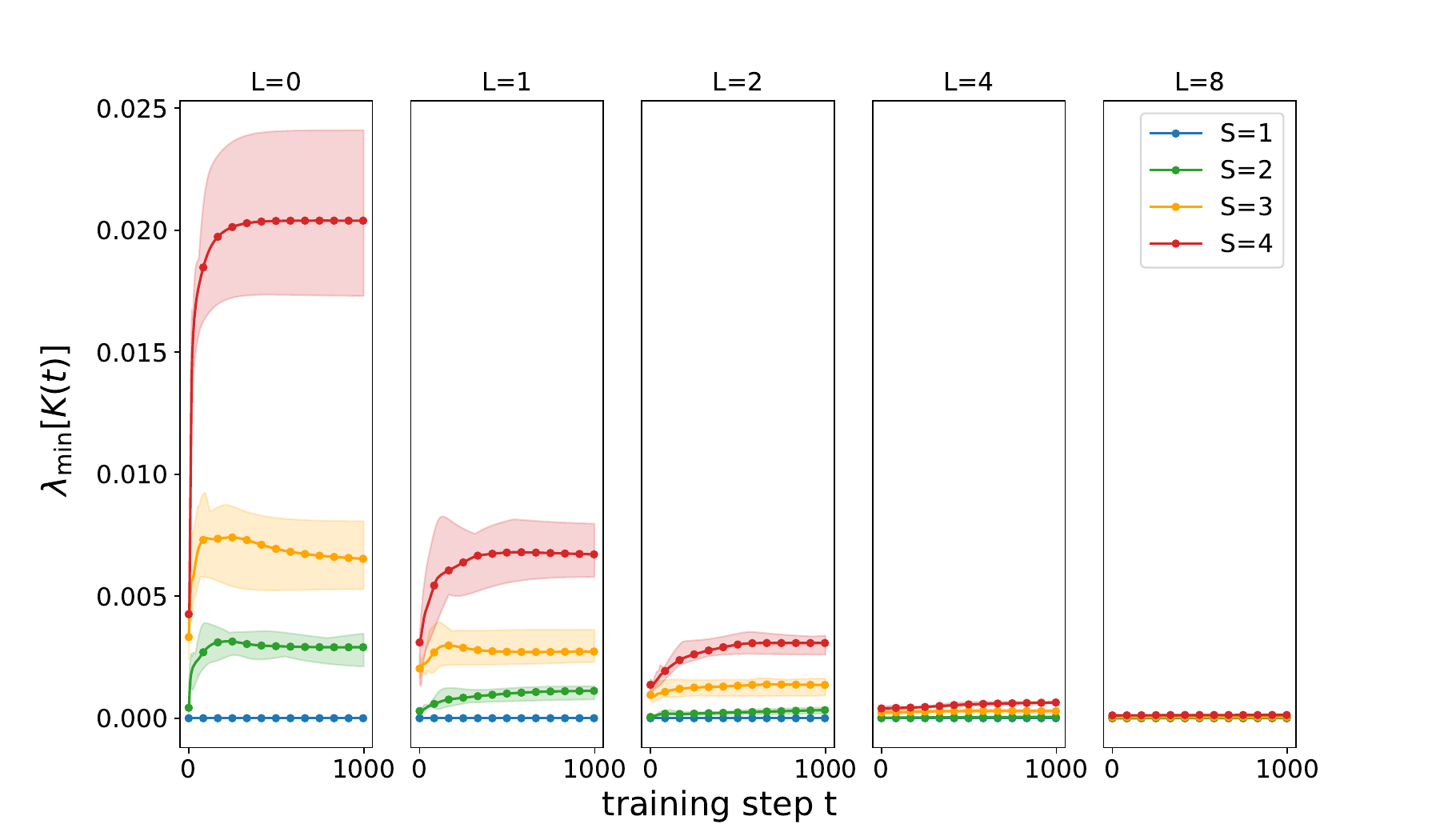}  
\label{vqacr_fig_qdl_fldc_lminK}
}
\caption{Numerical results of QDL with FLDC input states, where $(N,D)=(12,240)$ and $L \in \{0,1,2,4,8\}$. Figures~\ref{vqacr_fig_qdl_fldc_LA} and \ref{vqacr_fig_qdl_fldc_LB} show the relative loss of the training and the test dataset during the training, respectively. Figure~\ref{vqacr_fig_qdl_fldc_gradnorm} shows the $\ell_2$-norm of the gradient for the loss function. Figure~\ref{vqacr_fig_qdl_fldc_lminK} shows the least eigenvalue of the QNTK. Each solid line denotes the average of $5$ rounds of simulations with independent circuits and parameters. }
\label{vqacr_fig_qdl_fldc_L}
\end{figure*}

\subsection{Convergence analysis of training QNNs}
\label{vqacr_rate}

Here we provide two examples of QNNs with rigorous convergence analysis, where the input state in the training dataset satisfies Proposition~\ref{vqacr_prop_state_distribution_local}. The optimization of the loss function in Eq.~(\ref{vqacr_qntk_loss_eq}) exhibits proven linear convergence during gradient descent training with high probability, when the locally smooth assumption of the Jacobian matrix in Assumption~\ref{vqacr_lemma_ntk_local_Jacobian_stability_xyz} holds. 
Technically, we employ Lemmas~\ref{vqacr_lemma_qntk_least_eigen_gaussian_matrix_main} and \ref{vqacr_lemma_qntk_least_eigen_random_sphere_main} to bound the largest and the least eigenvalue of the initial QNTK. Thus we obtain necessary conditions given as the lower bound of the number of parameters $D$ and the number of local qubits $S$, such that $D$ gate Hamiltonians are sampled independently from $\{X,Y,Z\}^{\otimes S}$ or $\{X,Y\}^{\otimes S}$. These Lemmas are employed with Assumption~\ref{vqacr_lemma_ntk_local_Jacobian_stability_xyz} for deriving the convergence results in Theorems~\ref{vqacr_theorem_gd_xyz_lazy_theta_and_K} and \ref{vqacr_theorem_gd_xy_lazy_theta_and_K} with proof in Appendix~\ref{vqacr_qntk_app_qnn_linear_converge}. For the general case with Proposition~\ref{vqacr_prop_state_distribution_local}, the linear convergence rate requires $S\geq\tilde{\O}(\ln |\A|)$ and $D\geq\tilde{\O}(|\A|\ln^2|\A|)$. For the qubit embedding case in Eq.~(\ref{vqacr_qubit_embedding_eq}), the linear convergence rate requires $S\geq\tilde{\O}(\ln |\A|)$ and $D\geq\tilde{\O}(|\A|^5 \ln|\A|)$. We remark that the condition on $D$ is smaller than previous results that scale $\O(2^N)$~\cite{arxiv_larocca2021theory} when the size of dataset is polynomial to the qubit number. Apart from the linear convergence regime, we also obtain some by-products including the lazy training property shown as the small parameter changes and the stability of the QNTK during the training.

\begin{assumption}[Locally smooth condition]\label{vqacr_lemma_ntk_local_Jacobian_stability_xyz}
Let $J(\bm{\theta})\in \R^{D \times |\A|}$ be the Jacobian matrix of a QNN. Then for any constant $C_0>0$, we call $\mathcal{Z}(C_0, R_0, \bmt)$ a locally smooth region if for any $\bmt'$ with $\|\bmt'-\bmt\|_2 \leq R_0/\sqrt{D}$, the following holds:
\begin{equation}\label{vqacr_lemma_ntk_local_Jacobian_stability_eq}
\left\| J(\bm{\theta}') - J(\bm{\theta}) \right\|_2 \leq{} \sqrt{C_0 D} \| \bmt' -\bmt \|_2 .
\end{equation} 
\end{assumption}

\begin{theorem}[Linear convergence of QNN training w.r.t. Proposition~\ref{vqacr_prop_state_distribution_local}]\label{vqacr_theorem_gd_xyz_lazy_theta_and_K}
We follow conditions and notations in Lemma~\ref{vqacr_lemma_qntk_least_eigen_gaussian_matrix_main}. 
Denote $\lmin':={} \frac{2}{3} \Delta_S \alpha_{\min} A_{\Q, \min}'^2$ and $\lmax' :={} \frac{4}{3} \|O\|_2^2 \alpha_{\max} A_{\Q, \max}'^2$. Let $\L_{\A}(t)$ be the loss function at the $t$-th step. Suppose that there exists four constants $c_0, c_1, c_2, c_3>0$ and $\delta \in (0, 1)$, such that Assumption~\ref{vqacr_lemma_ntk_local_Jacobian_stability_xyz} holds for the parameter $\bmt(0)=\bm{0}$ with the constant $c_0$ and $R_0 := \frac{5}{2} {\lmin'^{-1}} \sqrt{2|\A|\lmax' \L_{\A}(0)} $ for every $S \geq{} \log_3 \left[ 1 + c_3 \frac{ |\A|^2 \ln |A| }{ A_{\Q, \min}'^4 |\Q| \delta^2 } \right]$ some $D \geq \max (D_1 \ln \frac{D_1}{\delta}, D_2 )$, where 
\begin{align*}
D_1 ={}& \left( \frac{\lmax' }{\lmin' A_{\Q, \max}' } \right)^2  c_1 |\A| \ln \frac{|\A|}{\delta} ,  \\
D_2 ={}& \frac{\lmax'^2 }{ \lmin'^4 } c_2 |\A| \L_{\A}(0) .
\end{align*}
Then the following holds with probability at least $1-\delta$ when applying gradient descent with learning rate $\eta={\eta_0}|\A|/{D}$ with $\eta_0 \in (0, \lmin'^{-1} )$,
\begin{align}
{}& \L_{\A} (T) \leq{} \left( 1 - \frac{1}{2} \eta_0 \lmin' \right)^{2T} \L_{\A} (0) , \label{vqacr_theorem_gd_xyz_lazy_theta_and_K_eq1} \\
{}& \left\| \bmt(T) - \bmt(0) \right\|_2  \leq \frac{R_0}{\sqrt{D}} . \label{vqacr_theorem_gd_xyz_lazy_theta_and_K_eq2}
\end{align}
\end{theorem}

\begin{theorem}[Linear convergence of QNN training w.r.t. qubit embedding]\label{vqacr_theorem_gd_xy_lazy_theta_and_K}
We follow conditions and notations in Lemma~\ref{vqacr_lemma_qntk_least_eigen_random_sphere_main}. Let $\L_{\A}(t)$ be the loss function at the $t$-th step.
Suppose that there exists four constants $c_0, c_1, c_2, c_3>0$ and $\delta \in (0, 1)$, such that Assumption~\ref{vqacr_lemma_ntk_local_Jacobian_stability_xyz} holds for the parameter $\bmt(0)=\bm{0}$ with the constant $c_0$ and $R_0 := \frac{5}{2} { \|O\|_2 }{\Delta_S^{-1} } \sqrt{6 \times 2^S |\A| \L_{\A}(0)} $ for some $S \geq 2+ \log_2 \frac{ c_3 |\A|^2 \ln |\A|}{ |\Q|\delta^2}$ and $D \geq \max (D_1 \ln \frac{D_1}{\delta}, D_2 )$, where 
\begin{align*}
D_1 ={}& c_1 \frac{\|O\|_2^4 }{ \Delta_S^2 } |\A| \ln \frac{|\A|}{\delta}, \\
D_2 ={}& c_2 \frac{\|O\|_2^4 }{ \Delta_S^4  } 2^{2S} |\A| \L_{\A}(0) .
\end{align*}
Then the following holds with probability at least $1-\delta$ when applying gradient descent with learning rate $\eta={\eta_0}|\A|/{D}$ with $\eta_0 \in (0, 3 \times 2^{S-1} \Delta_S^{-1} )$,
\begin{align}
{}& \L_{\A} (T) \leq{} \left( 1 - \eta_0  \frac{\Delta_S}{3\times 2^S} \right)^{2T} \L_{\A} (0) , \label{vqacr_theorem_gd_xy_lazy_theta_and_K_eq1} \\
{}& \left\| \bmt(T) - \bmt(0) \right\|_2  \leq \frac{R_0}{\sqrt{D}} . \label{vqacr_theorem_gd_xy_lazy_theta_and_K_eq2}
\end{align}
\end{theorem}

\section{Numerical Results}

Here we provide some numerical evidence that suggests the separated performances of training QNNs with different input states, regarding the convergence rate and the generalization behavior.

\subsection{Quantum dynamics learning}
\label{vqacr_exper_qdl}

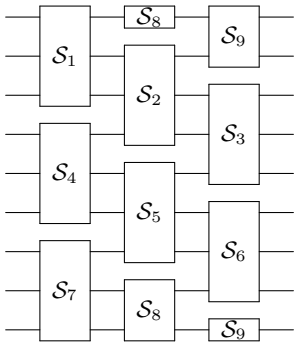
\begin{figure}
\centerline{
\Qcircuit @C=1.4em @R=.7em {
\lstick{} & \multigate{2}{\mathcal{S}_1} & \multigate{0}{\mathcal{S}_8} & \multigate{1}{\mathcal{S}_9} & \qw \\
\lstick{} & \ghost{\mathcal{S}_1} & \multigate{2}{\mathcal{S}_2} & \ghost{\mathcal{S}_9} & \qw \\
\lstick{} & \ghost{\mathcal{S}_1} & \ghost{\mathcal{S}_2} & \multigate{2}{\mathcal{S}_3} & \qw \\
\lstick{} & \multigate{2}{\mathcal{S}_4} & \ghost{\mathcal{S}_2} & \ghost{\mathcal{S}_3} & \qw \\
\lstick{} & \ghost{\mathcal{S}_4} & \multigate{2}{\mathcal{S}_5} & \ghost{\mathcal{S}_3} & \qw \\
\lstick{} & \ghost{\mathcal{S}_4} & \ghost{\mathcal{S}_5} & \multigate{2}{\mathcal{S}_6} & \qw \\
\lstick{} & \multigate{2}{\mathcal{S}_7} & \ghost{\mathcal{S}_5} & \ghost{\mathcal{S}_6} & \qw \\
\lstick{} & \ghost{\mathcal{S}_7} & \multigate{1}{\mathcal{S}_8} & \ghost{\mathcal{S}_6} & \qw \\
\lstick{} & \ghost{\mathcal{S}_7} & \ghost{\mathcal{S}_8} & \multigate{0}{\mathcal{S}_9} & \qw
}
}
\caption{An assignment of $\mathcal{Q}=\{\mathcal{S}_1,\cdots,\mathcal{S}_{N}\}$ in the hardware-efficient manner for $(N,S)=(9,3)$.}
\label{vqacr_Q_Sd_he_circuit_main}
\end{figure}

The first task is to learn quantum dynamics, where the aim is to approximate the measurement performance of the target unitary using variational quantum circuits. Given the target unitary $U$, the observable $O$, and an ensemble of input states $\{(\rho_a, y_a)\}_{a \in \A}$, we consider the optimization of the loss defined in Ref.~\cite{Wang2024}, i.e. 
\begin{align}
\L_{\rm QDL} := \frac{1}{2|\A|} \sum_{a \in \A} \left( \Tr[O V(\bmt) \rho_a V(\bmt)^\dag ] - y_a \right)^2 ,
\end{align}
with the label 
$y_a := \Tr [O U \rho_a U^\dag] $. 
In the experiment, we set the qubit number $N=12$.
The observable is set to be
\begin{equation}\label{vqacr_experiment_observable_eq}
O = \sum_{n=1}^{N} c_n Z_n
\end{equation}
with independent normal variables $c_n \sim \mathcal{N}(0,1)$. 
The target unitary is set to be the time evolution of the Heisenberg model~\cite{bonechi1992heisenberg} with open boundary conditions 
\begin{align}
U = \exp \left[-it\sum_{n=1}^{N} X_{n} X_{n+1} + Y_{n} Y_{n+1} + Z_{n} Z_{n+1} \right] , \notag
\end{align}
where we use $t=1$. 
Input states are generated from the state $|0\>$ after finite local-depth circuits (Fig~\ref{vqacr_fldcs_example}) with $L$ blocks. The size of both training and test sets are $40$.
Hamiltonians of variational quantum gates are sampled from the set $ \{X, Y, Z\}^{\otimes S}$, and the deployment of quantum gates assumes the 1D connectivity with $|\Q|=N$. The variational circuit part consists of $\{5, 7, 10, 14, 20\}$ layers of gates in Figure~\ref{vqacr_Q_Sd_he_circuit_main}. We initialize parameters as $\bmt(0)=\bm{0}$, which is consistent with the analysis in Theorem~\ref{vqacr_theorem_gd_xyz_lazy_theta_and_K}. 
During the training stage, we use gradient descent with learning rate $\eta=N/D$ for $L \in \{0,1,2,4\}$. For the case $L=8$, the initial gradient is at least an order of magnitude smaller that that with other $L$ values, as shown in Figure~\ref{vqacr_fig_qdl_fldc_gradnorm}. Therefore, we adopt a larger learning rate $\eta=10N/D$ when $L=8$ to guarantee the convergence of the training set loss.

\begin{figure}[t]
\centering
\includegraphics[width=.96\linewidth]{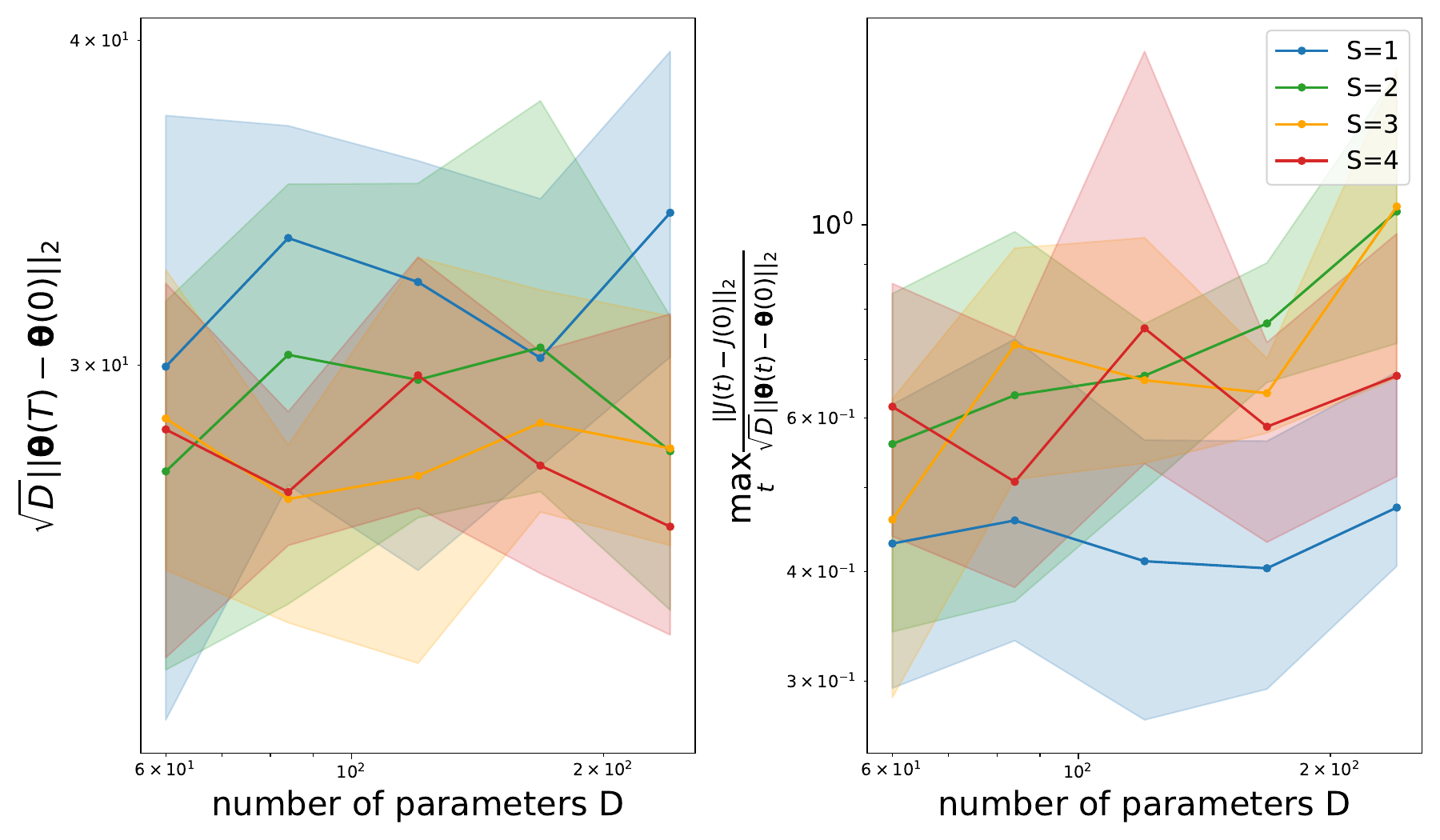}  
\caption{Numerical results of QDL with FLDC input states, where $(N,L)=(12,1)$ and $D \in \{60, 84, 120, 168, 240\}$. The left figure illustrates the $\ell_2$-norm distance in the parameter space from the initial to the final step. The right figure illustrates the largest coefficient in Eq.~(\ref{vqacr_lemma_ntk_local_Jacobian_stability_eq}) during the training. Each solid line denotes the average of $5$ rounds of simulations with independent circuits and parameters. }
\label{vqacr_fig_qdl_fldc_D_thetaandJ}
\end{figure}

\begin{figure*}[t]
\centering
\subfloat[]{
\includegraphics[width=.48\linewidth]{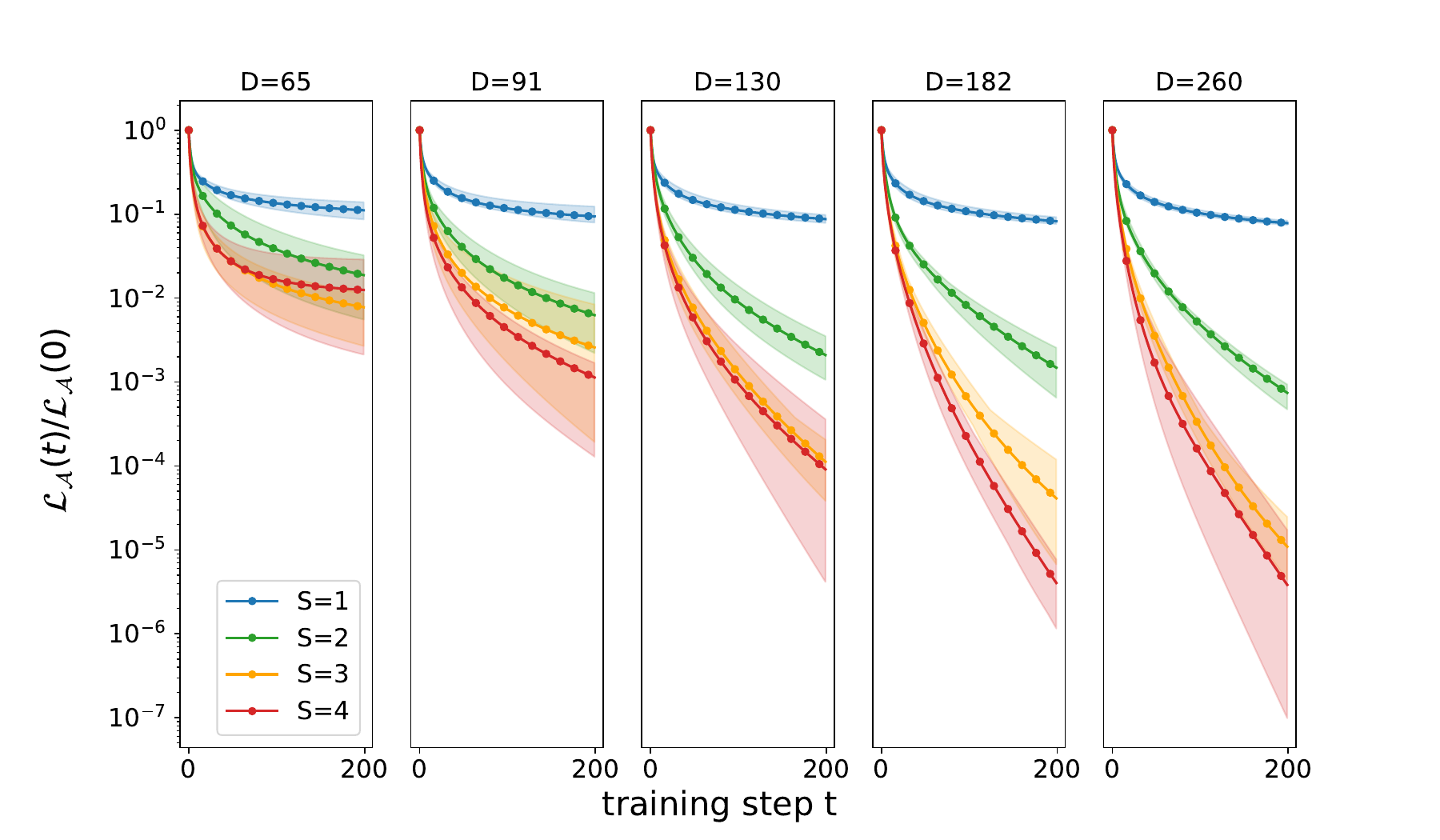}  
\label{vqacr_fig_bc_wine_LA}
}
\subfloat[]{
\includegraphics[width=.48\linewidth]{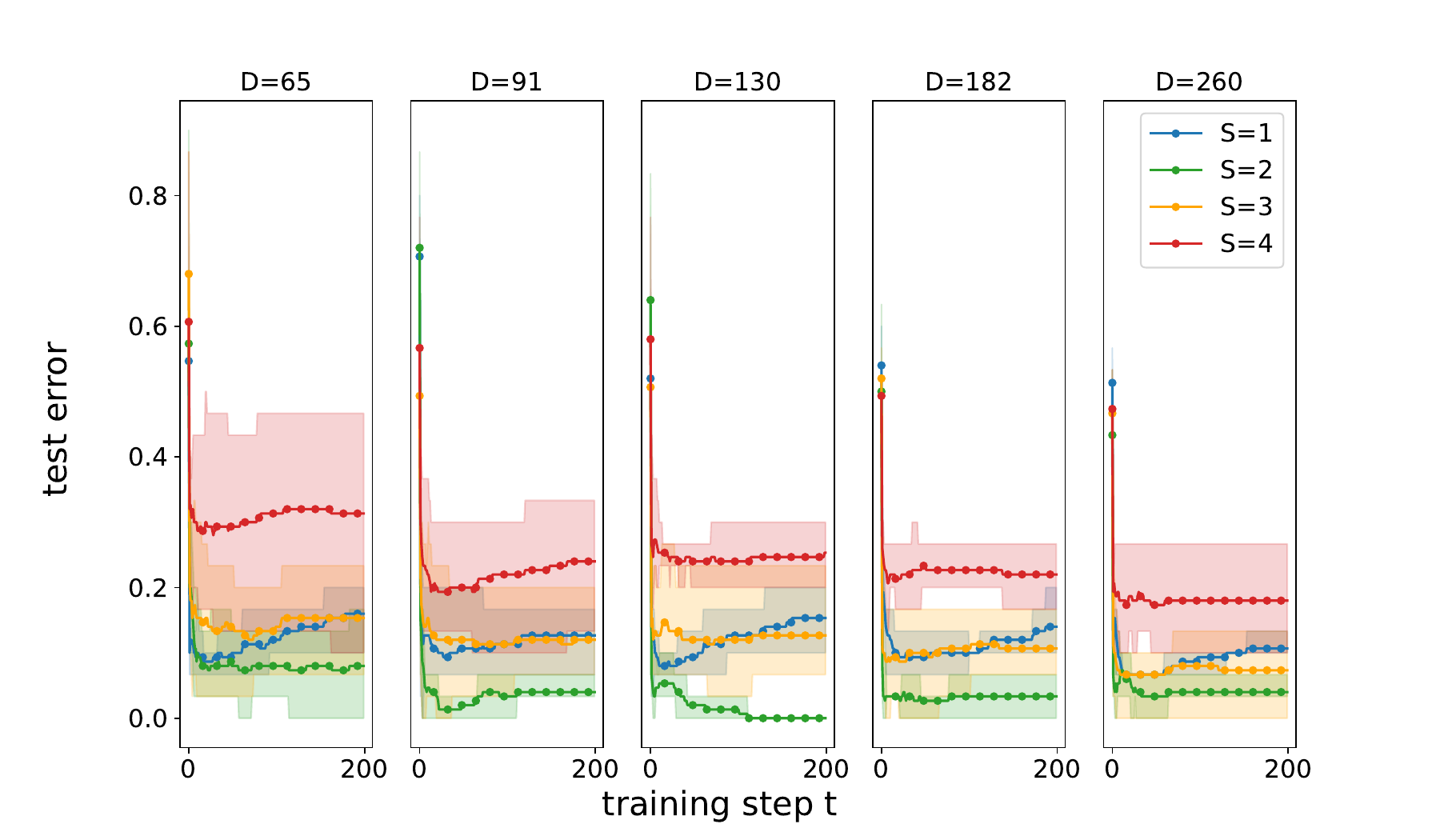}  
\label{vqacr_fig_bc_wine_eB}
}\\
\subfloat[]{
\includegraphics[width=.48\linewidth]{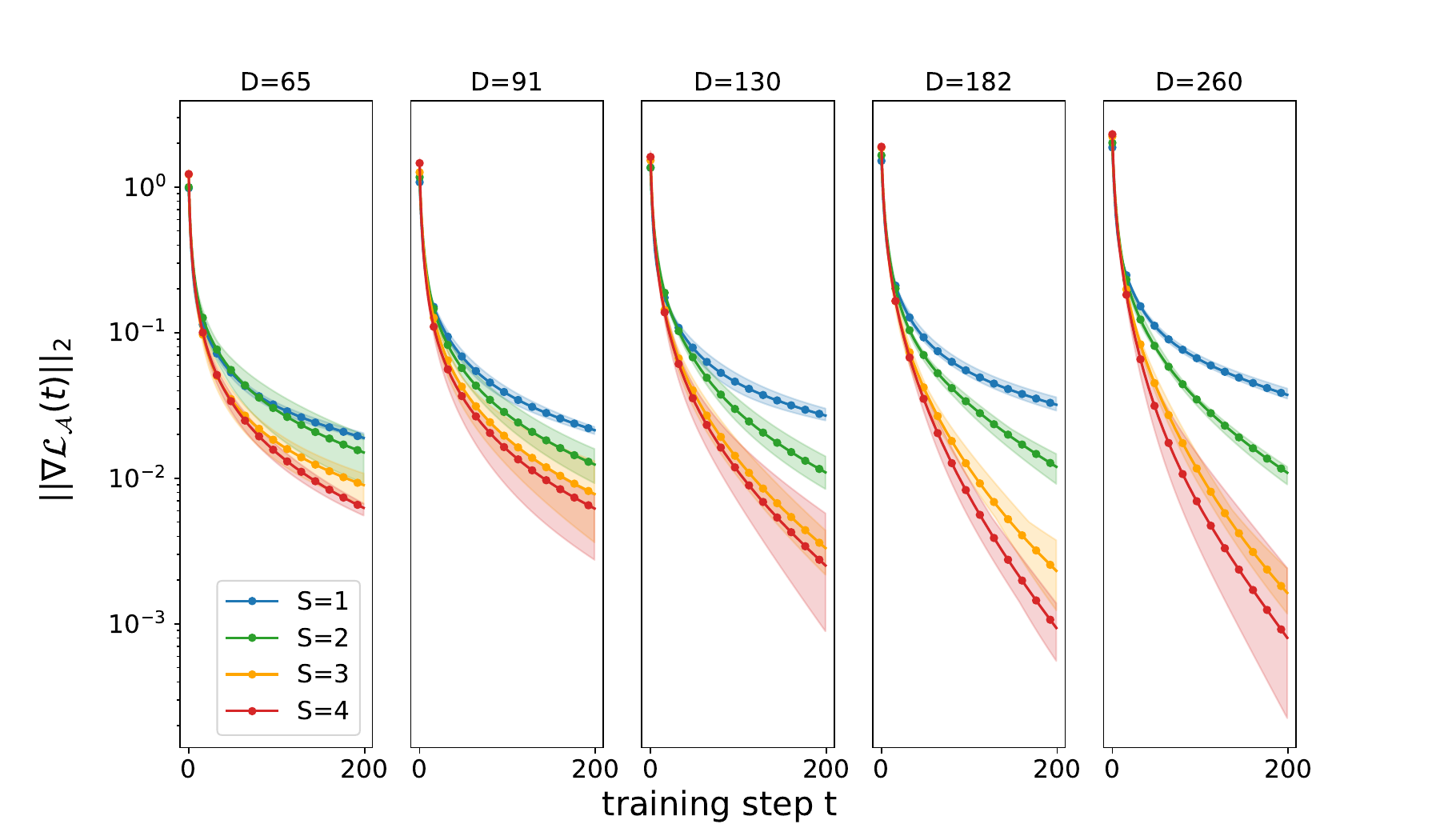}  
\label{vqacr_fig_bc_wine_gradnorm}
}
\subfloat[]{
\includegraphics[width=.48\linewidth]{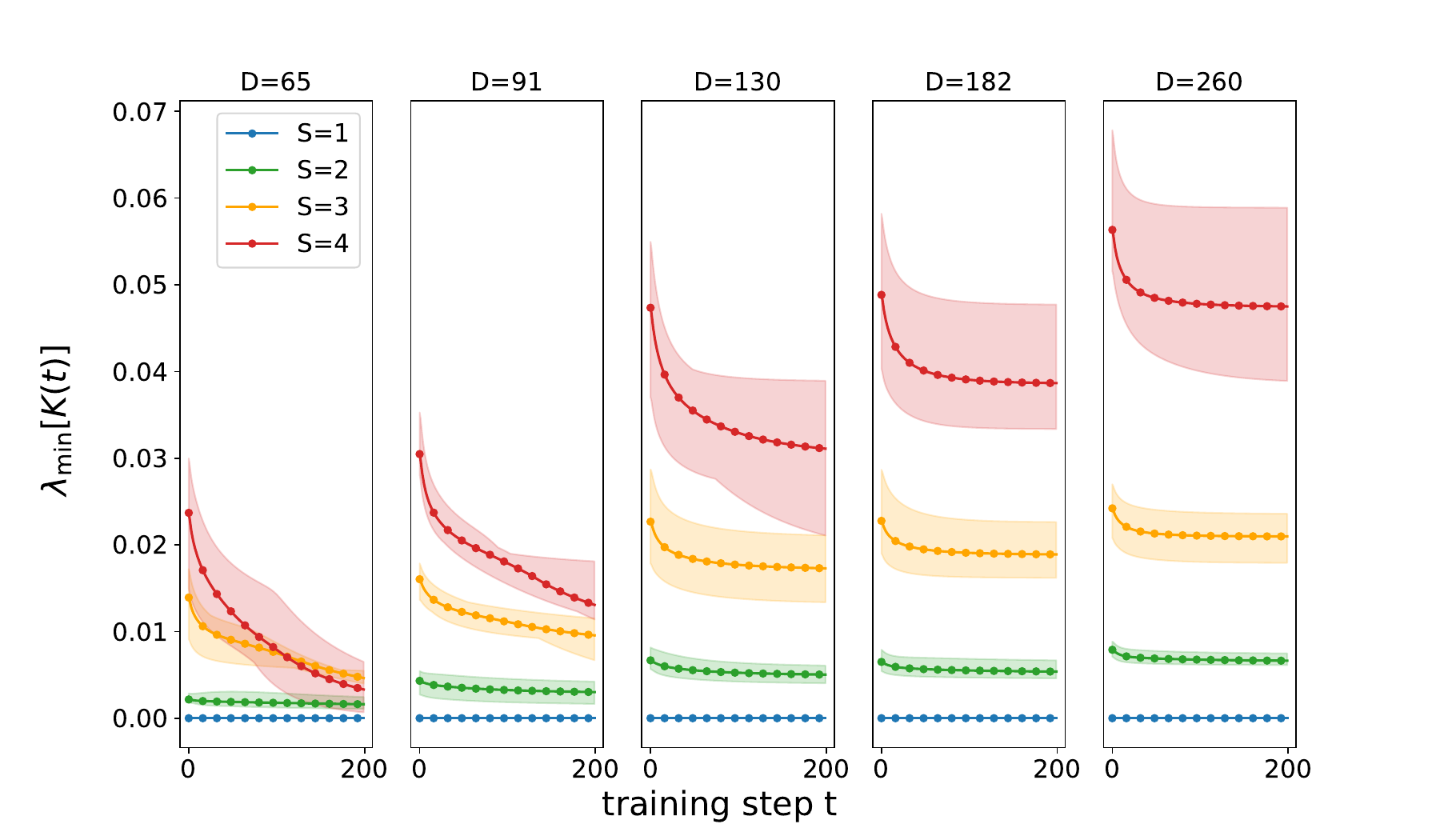}  
\label{vqacr_fig_bc_wine_lminK}
}
\caption{Numerical results of the binary classification with the wine dataset, where $N=13$ and $D \in \{65, 91, 130, 182, 260\}$. Figures~\ref{vqacr_fig_bc_wine_LA} and \ref{vqacr_fig_bc_wine_eB} show the training loss the test error during the training, respectively. Figure~\ref{vqacr_fig_bc_wine_gradnorm} shows the $\ell-2$-norm of the gradient of the training loss. Figure~\ref{vqacr_fig_bc_wine_lminK} shows the least eigenvalue of the QNTK during the training. Each solid line denotes the average of $5$ rounds of simulations with independent circuits and parameters. }
\label{vqacr_fig_bc_wine}
\end{figure*}

We exhibit numerical results in Figures~\ref{vqacr_fig_qdl_fldc_L} and \ref{vqacr_fig_qdl_fldc_D_thetaandJ}. Figures~\ref{vqacr_fig_qdl_fldc_LA} and \ref{vqacr_fig_qdl_fldc_LB} plot the training and the test loss over their initial values during the training. The convergence rate decreases towards the Haar limit as the input layer $L$ increases, while the test loss in Figure~\ref{vqacr_fig_qdl_fldc_LB} illustrates the separation of the generalization capability using different input states. In general, the test loss has worse performance with larger $L$. Meanwhile, the convergence performance is also influenced by the choice of parameterized gates, which is marked by different $S$ values. For the case $S=1$, the loss converges to non-zero values, implying the insufficient hypothesis space of the single-qubit unitary product. With larger $S$, the training loss exhibits increasingly pronounced and rapid linear convergence, which is further corroborated by the larger $\lmin[K]$ shown in Figure~\ref{vqacr_fig_qdl_fldc_lminK}.
Finally, we evaluate $D^{-1/2}$ and $D^{1/2}$ dependences in the frozen parameter result in Eq.~(\ref{vqacr_theorem_gd_xyz_lazy_theta_and_K_eq2}) and the locally smooth assumption in Eq.~(\ref{vqacr_lemma_ntk_local_Jacobian_stability_eq}). The results are shown in Figure~\ref{vqacr_fig_qdl_fldc_D_thetaandJ} where we employ VQCs with different depths, such that the number of parameters has values $D \in \{60, 84, 120, 168, 240\}$. Both dependences are more evident as $S$ increases to $4$.
In summary, we have verified theoretical results in Theorem~\ref{vqacr_theorem_gd_xyz_lazy_theta_and_K}.

\subsection{Binary classification}
\label{vqacr_experiment_bc_wines} 

In the second task, we evaluate the performance of binary classification using variational quantum circuits on the wine dataset~\cite{Dua2019}. This dataset comprises 138 instances from three different cultivars, with each instance featuring 13 measurements of various chemical components. For our experiment, we choose two wine types and divide the dataset into a training and a test set, each containing 30 samples, where samples from different classes $y=\pm 1$ are balanced in both sets. 
We use qubit embedding in Eq.~(\ref{vqacr_qubit_embedding_eq}) to encode the dataset information into $N=13$ qubits quantum states. 
Variational quantum circuits are generated randomly in accordance with all-to-all connectivity, where Hamiltonians are sampled from $\{X, Y\}^{\otimes S}$ with $S\in \{1,2,3,4\}$. We use the loss function~(\ref{vqacr_qntk_loss_eq}) with the observable in Eq.~(\ref{vqacr_experiment_observable_eq}). Coefficients are sampled from $c_n \sim \mathcal{N}(0, 2^{S}/N)$ to mitigate the local minima issue raised in recent studies~\cite{icml_xiaodi2021localminima,arxiv_anschuetz2022beyond}. We training each optimization task with gradient descent with the learning rate $\eta=N/D$ for $T=200$ steps.

We present numerical results in Figures~\ref{vqacr_fig_bc_wine} and \ref{vqacr_fig_bc_wines_D_thetaandJ}. Figures~\ref{vqacr_fig_bc_wine_LA}, \ref{vqacr_fig_bc_wine_eB}, \ref{vqacr_fig_bc_wine_gradnorm}, and \ref{vqacr_fig_bc_wine_lminK} plot the training loss, the test error, the gradient norm, and the least eigenvalue of the QNTK during the training. 
Similar to Section~\ref{vqacr_exper_qdl}, the training loss with qubit embedding exhibits increasingly pronounced and rapid linear convergence with increased $S$, which is further corroborated by the $\lmin[K]$ performance in Figure~\ref{vqacr_fig_bc_wine_lminK}. However, the advantage on the training loss convergence does not always induce smaller error on the test set. As shown in Figure~\ref{vqacr_fig_bc_wine_eB}, the test error increases slightly in later training stages, which is known as the overtraining issue. For the case $S=1$, the error increases the most. Another phenomenon is the relatively worse performance of the case $S=4$ compared with that of other S values, i.e. the overfitting. One potential reason of this phenomenon is that VQCs with large $S$ could extract excessive features from input states. Thus the training of parameters could capture each training sample easily before learning enough knowledge to classify test samples. Therefore, an appropriate selection of $S$ in a QNN should strike a balance between training convergence and generalization performance.
We note that a detailed generalization analysis of QNNs with input states generated from qubit embedding or other ensembles that deviate from 2-design states is beyond the scope of this work and could be an interesting direction for future research.
Finally, we evaluate $D^{-1/2}$ and $D^{1/2}$ dependences in the frozen parameter result in Eq.~(\ref{vqacr_theorem_gd_xy_lazy_theta_and_K_eq2}) and the locally smooth assumption in Eq.~(\ref{vqacr_lemma_ntk_local_Jacobian_stability_eq}). The results are shown in Figure~\ref{vqacr_fig_bc_wines_D_thetaandJ}.
Similar to Section~\ref{vqacr_exper_qdl}, both dependences are more evident as $S$ increases to $4$.
In summary, we have verified theoretical results in Theorem~\ref{vqacr_theorem_gd_xy_lazy_theta_and_K}.

\begin{figure}[t]
\centering
\includegraphics[width=.96\linewidth]{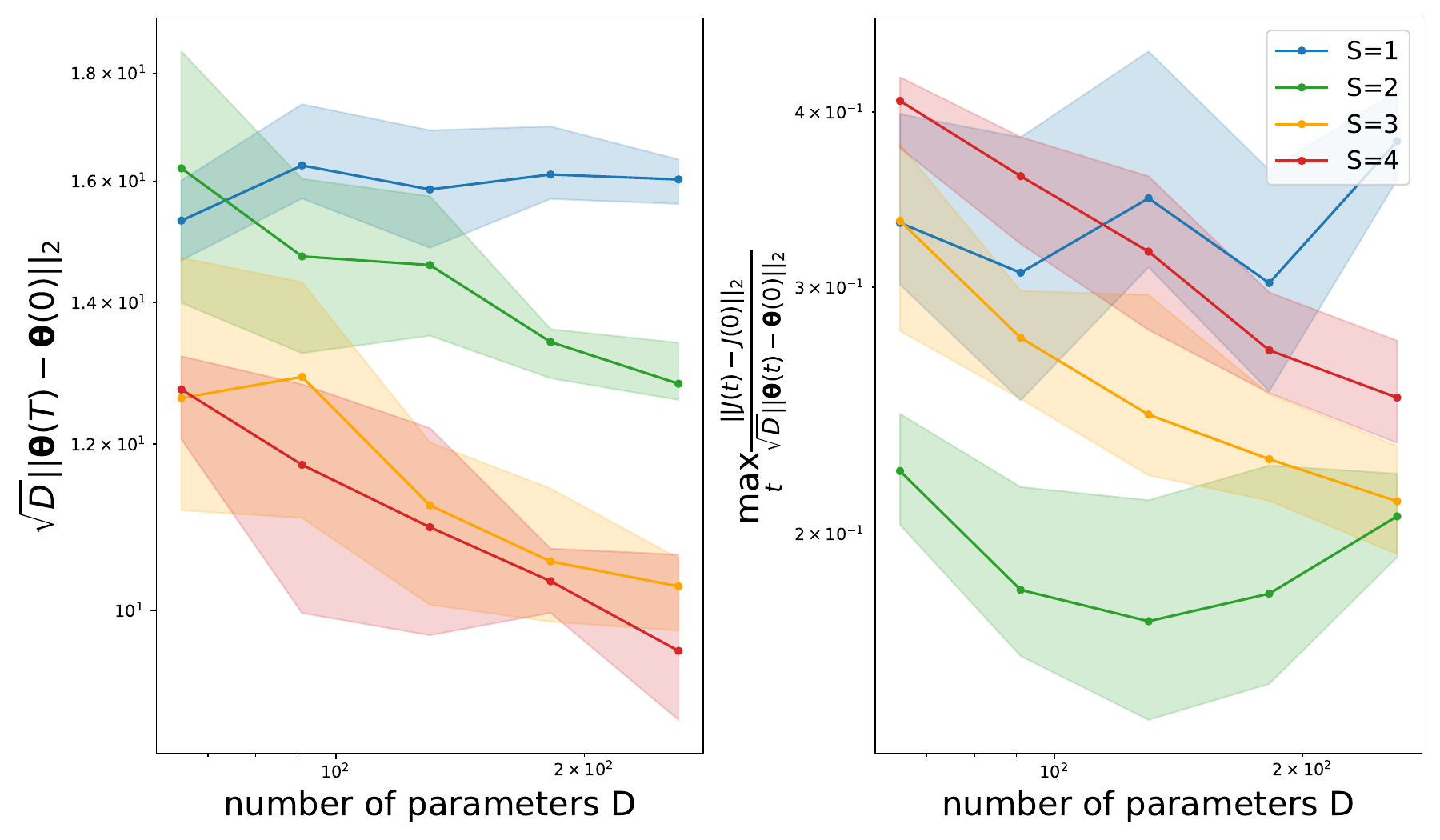}  
\caption{Numerical results of the binary classification with the wine dataset, where $N=13$ and $D \in \{65, 91, 130, 182, 260\}$. The left figure illustrates the $\ell_2$-norm distance in the parameter space from the initial to the final step. The right figure illustrates the largest coefficient in Eq.~(\ref{vqacr_lemma_ntk_local_Jacobian_stability_eq}) during the training. Each solid line denotes the average of $5$ rounds of simulations with independent circuits and parameters. }
\label{vqacr_fig_bc_wines_D_thetaandJ}
\end{figure}

\section{Discussion}

In this manuscript, we uncover the training and generalization issues of quantum machine learning with uniformly random quantum data. We demonstrate the separation performances of quantum datasets with different distributions in terms of the convergence rate and the generalization capability. Specifically, we prove fast linear convergence when optimizing the $\ell_2$-norm loss function under certain input state conditions. 
By studying the lower bound of QNTK, we derive reasonable convergence rate when the number of parameters and the number of qubits in local Hamiltonians exceed some moderate thresholds. All proposed theorems are verified by numerical experiments.

Despite these advancements, further research is imperative for a more profound understanding of QNN training. Notably, our analysis relies on the random matrix theory and assumes certain stochastic properties of datasets. Although numerical experiments indicate favorable convergence with real-world datasets, elucidating training dynamics on datasets with explicit real-world or quantum characteristics remains an area for future exploration.
Furthermore, this research primarily analyzes QNNs from the perspective of dataset properties, and VQCs are initialized with zero parameters. We expect that theories with zero initialization could be engaged as technical tools in the study with other initializations by merging the initial VQC into input states. 
Another intriguing direction for future work is to refine the parameter threshold that guarantees the linear convergence, as experiments indicate rapid convergence for circuits with the number of parameters $D\sim10^2$, which is significantly lower than our theoretical thresholds outlined in Section~\ref{vqacr_rate}. Additionally, we hope that this manuscript could inspire further research into the convergence behavior of variational models with analogous structures, such as the training of tensor networks~\cite{NIPS2016_5314b967}.

\bigskip

\textit{Acknowledgements.---}We thank Jens Eisert for useful discussions.
JL is supported in part by International Business Machines (IBM) Quantum through the Chicago Quantum Exchange, and the Pritzker School of Molecular Engineering at the University of Chicago through AFOSR MURI (FA9550-21-1-0209). LJ acknowledges support from ARO (W911NF-18-1-0020, W911NF-18-1-0212), ARO MURI (W911NF-16-1-0349, W911NF-21-1-0325), AFOSR MURI (FA9550-19-1-0399, FA9550-21-1-0209), AFRL (FA8649-21-P-0781), DoE Q-NEXT, NSF (OMA-1936118, EEC-1941583, OMA-2137642), NTT Research, and the Packard Foundation (2020-71479). This research used resources of the Oak Ridge Leadership Computing Facility, which is a DOE Office of Science User Facility supported under Contract DE-AC05-00OR22725.

\bibliographystyle{unsrt}
\bibliography{reference.bib}

\clearpage

\pagebreak

\pagenumbering{arabic}
\renewcommand*{\thepage}{S\arabic{page}}
\setcounter{lemma}{0}
\renewcommand{\thelemma}{S\arabic{lemma}}
\setcounter{theorem}{0}
\renewcommand{\thetheorem}{S\arabic{theorem}}

\onecolumngrid
\appendix

\vspace{0.5in}

\begin{center}
	{\Large \bf Supplemental Material}
\end{center}

\tableofcontents

\newpage

\section{Additional Numerical Results}

\begin{figure}[t]
\centering
\subfloat[]{
\includegraphics[width=.48\linewidth]{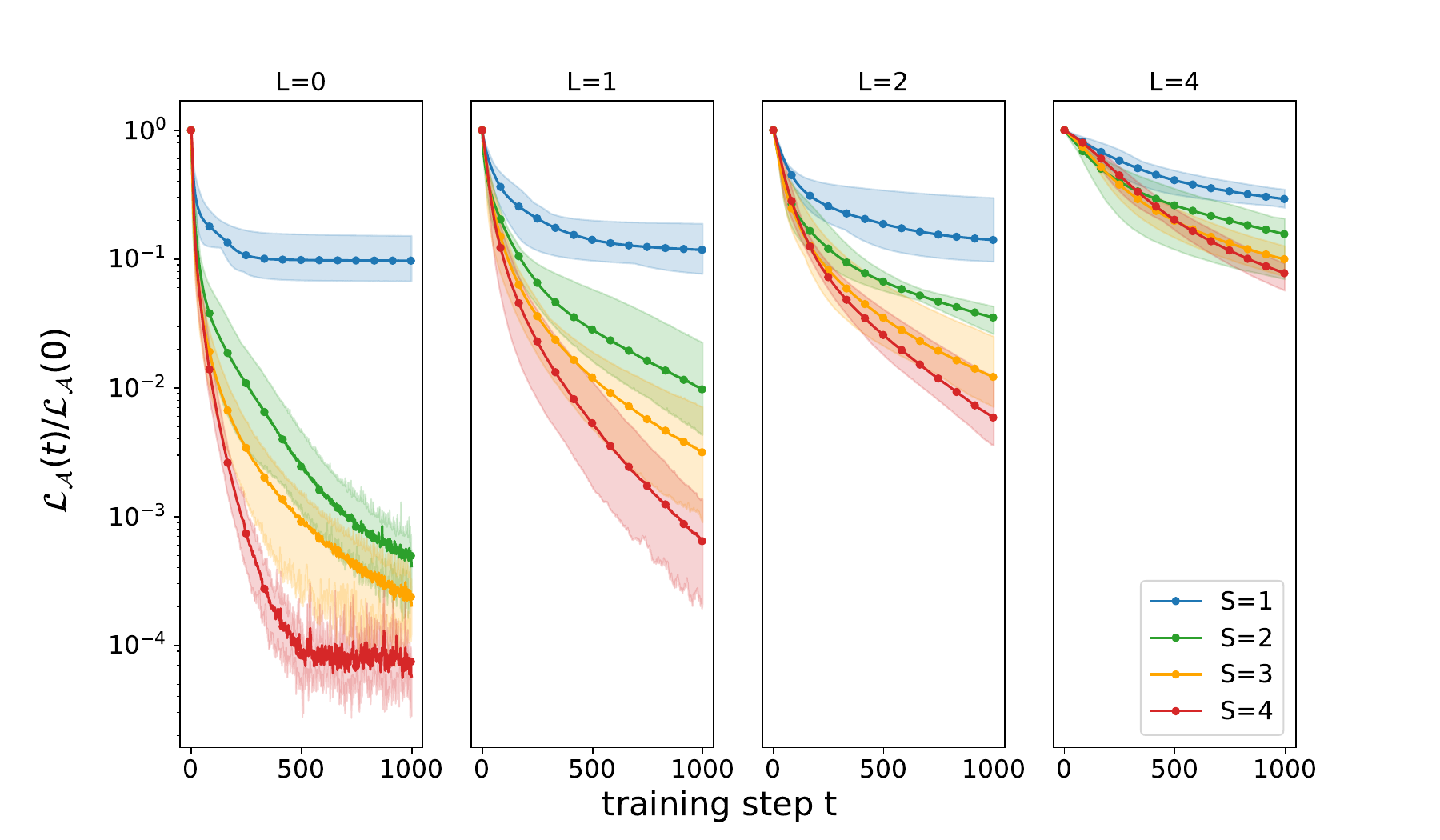}  
\label{vqacr_fig_qdl_fldc_sgd_LA}
}
\subfloat[]{
\includegraphics[width=.48\linewidth]{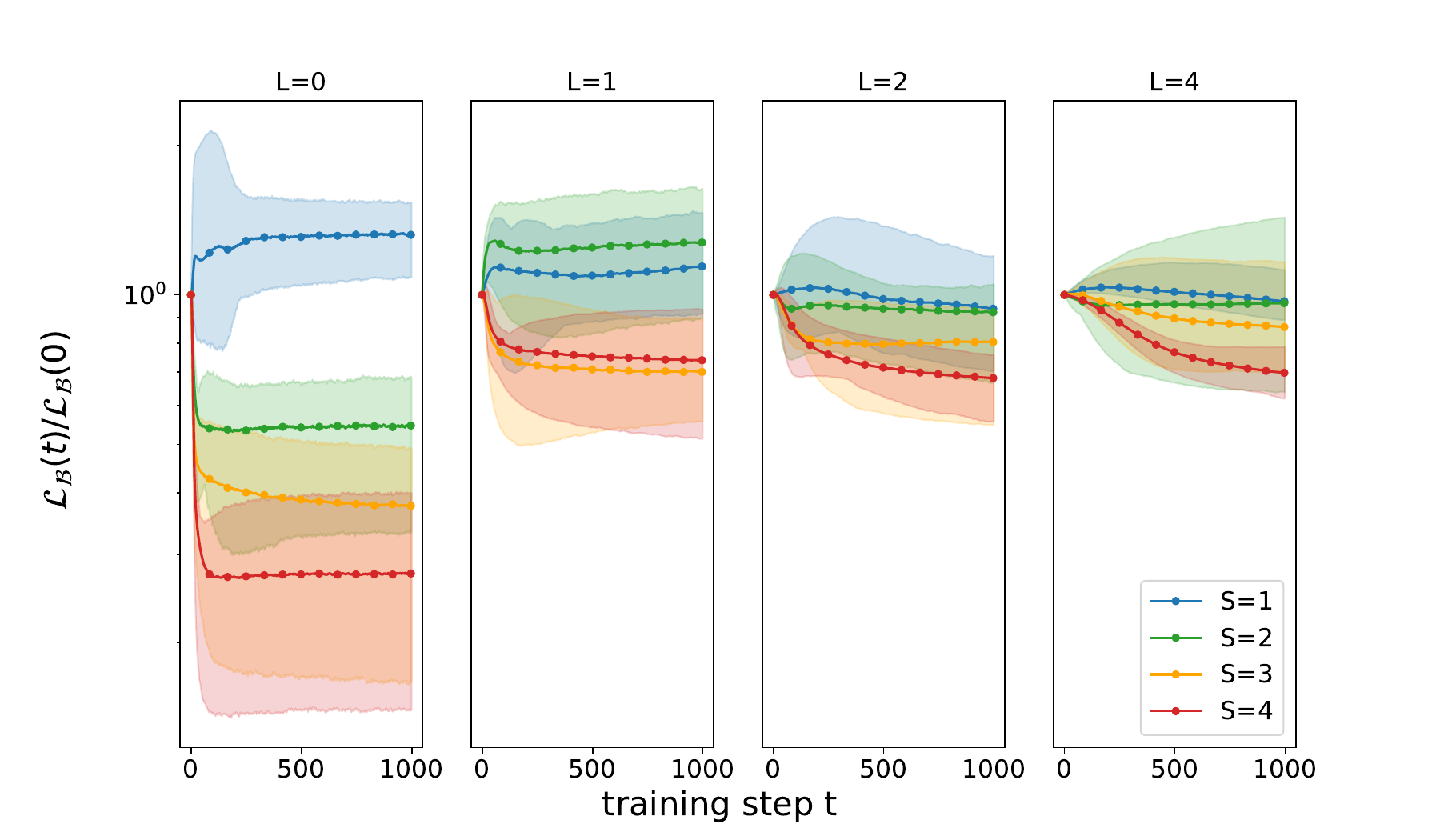}  
\label{vqacr_fig_qdl_fldc_sgd_LB}
}\\
\subfloat[]{
\includegraphics[width=.48\linewidth]{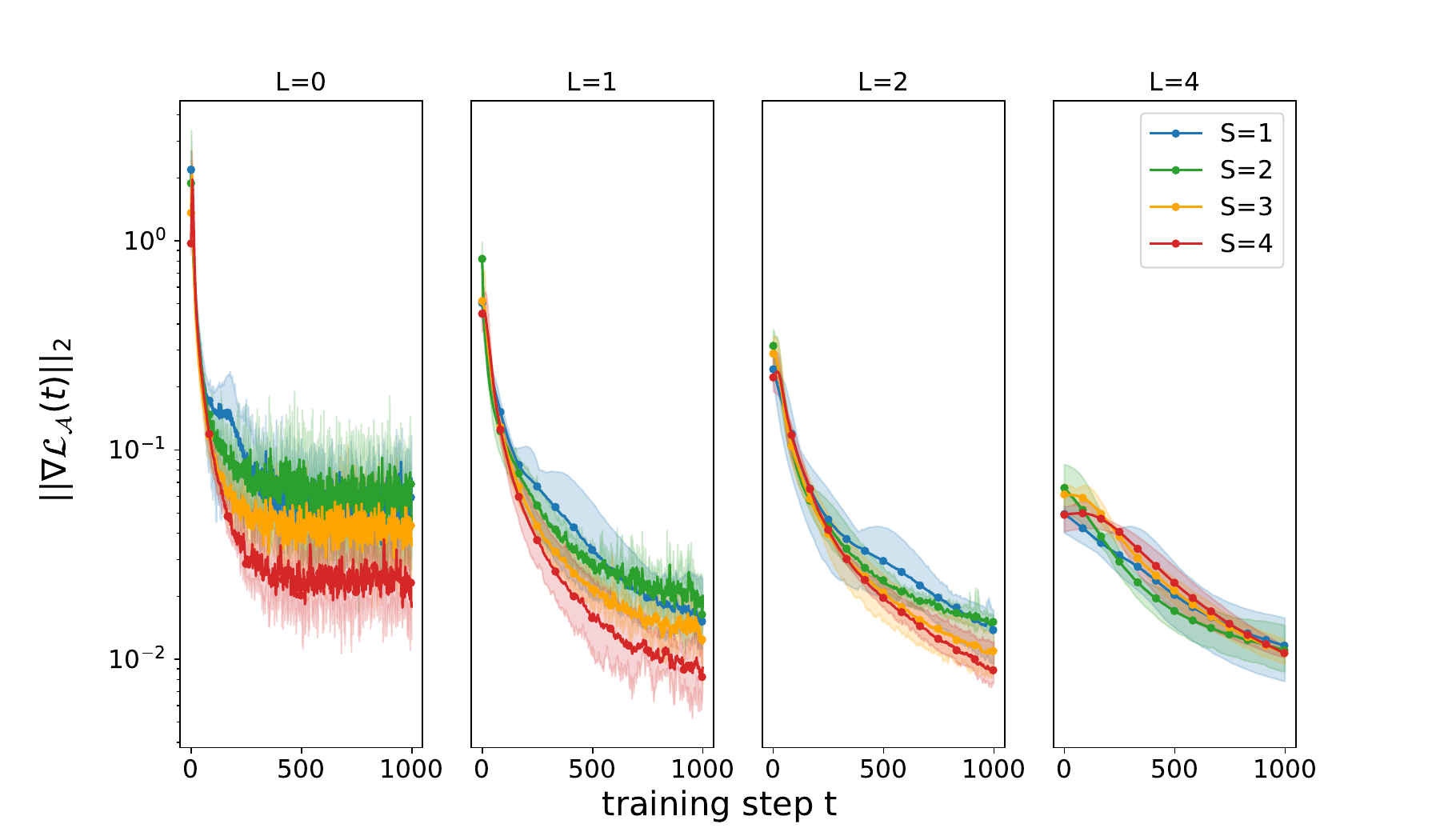}  
\label{vqacr_fig_qdl_fldc_sgd_gradnorm}
}
\subfloat[]{
\includegraphics[width=.48\linewidth]{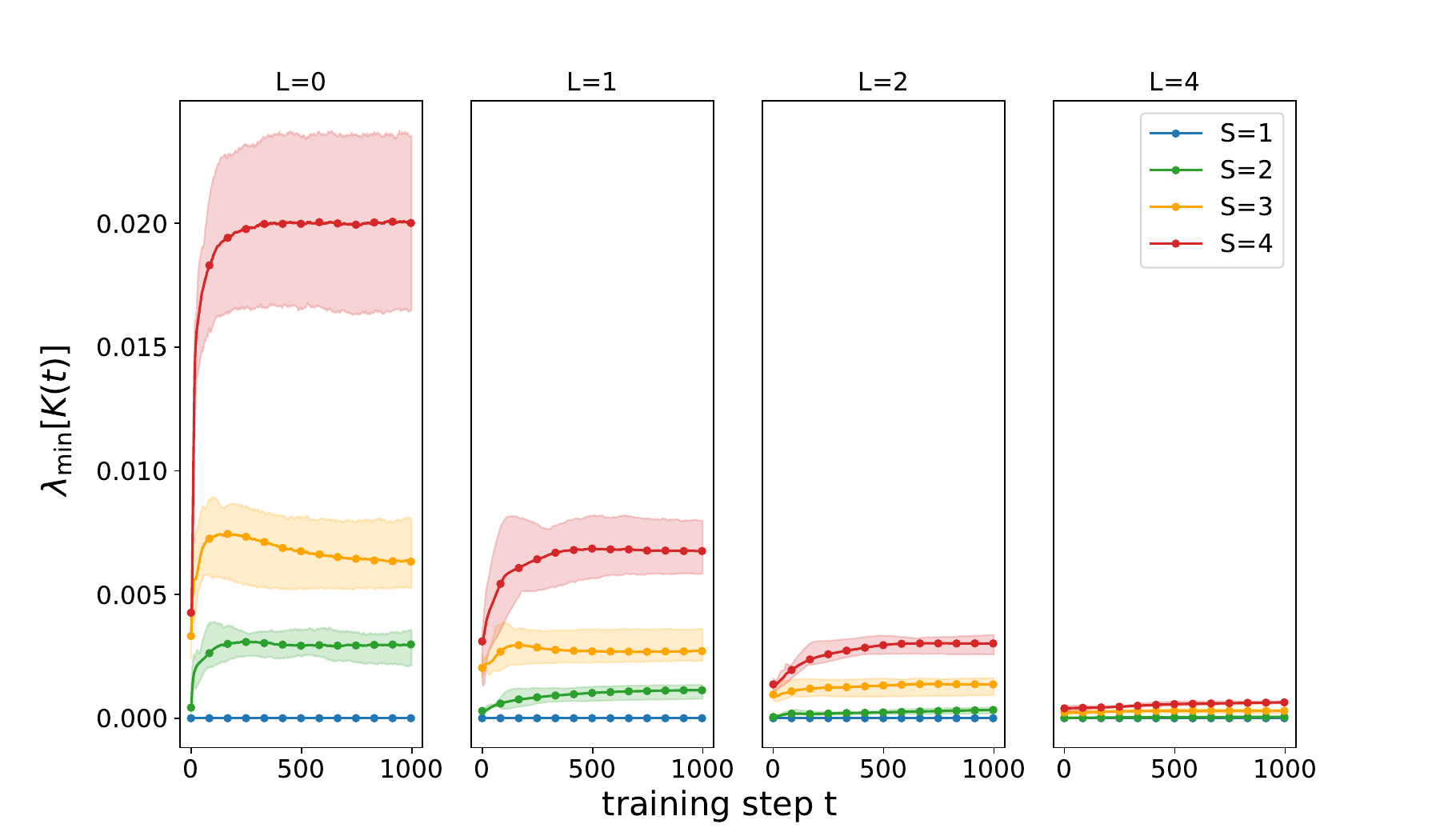}  
\label{vqacr_fig_qdl_fldc_sgd_lminK}
}
\caption{Numerical results of QDL with FLDC input states and stochastic gradient descent, where $(N,D)=(12,240)$ and $L \in \{0,1,2,4\}$. Figures~\ref{vqacr_fig_qdl_fldc_sgd_LA} and \ref{vqacr_fig_qdl_fldc_sgd_LB} show the relative loss of the training and the test dataset during the training, respectively. Figure~\ref{vqacr_fig_qdl_fldc_sgd_gradnorm} shows the $\ell_2$-norm of the gradient for the loss function. Figure~\ref{vqacr_fig_qdl_fldc_sgd_lminK} shows the least eigenvalue of the QNTK. Each solid line denotes the average of $5$ rounds of simulations with independent circuits and parameters. }
\label{vqacr_fig_qdl_fldc_sgd_L}
\end{figure}

Here we present some additional numerical results. Figure~\ref{vqacr_fig_qdl_fldc_sgd_L} shows the training and generalization performance of the quantum dynamics learning with FLDC states using stochastic gradient descent (SGD) optimizer. We add Gaussian noise into gradients to simulate the finite shots noise ($1000$ shots) when estimating measurement outcomes via the parameter-shift rule. Compared with the exact case, the training loss decays linearly towards the noise level in the SGD case. 
In terms of the test loss and the least eigenvalue of the QNTK, the exact and noisy cases exhibit comparable performance.

\section{Technical Lemmas}

In this section, we introduce some helpful lemmas that are related to unitary $t$-design distributions, variational quantum algorithms, and random matrix theory.

\subsection{Lemmas about unitary $t$-designs}

\begin{lemma}[from Ref.~\cite{nc_cerezo2020cost}]\label{vqacr_qntk_tdesign_tr_wawb}
Let $\{W_y\}_{y \in Y} \subset \mathcal{U}(d)$ form a unitary $t$-design with $t \geq 1$, and let $A,B: \mathcal{H}_w \rightarrow \mathcal{H}_w$ be arbitrary linear operators. Then
\begin{equation}\label{vqacr_qntk_tdesign_tr_wawb_eq}
\frac{1}{|Y|} \sum_{y \in Y} \Tr [W_y A W_y^\dag B] = \int d\mu(W) \Tr [W A W^\dag B] =\frac{\Tr[A]\Tr[B]}{d}.
\end{equation}
\end{lemma}

\begin{lemma}[from Ref.~\cite{nc_cerezo2020cost}]\label{vqacr_qntk_tdesign_trwawbtrwcwd}
Let $\{W_y\}_{y \in Y} \subset \mathcal{U}(d)$ form a unitary $t$-design with $t \geq 2$, and let $A,B,C,D: \mathcal{H}_w \rightarrow \mathcal{H}_w$ be arbitrary linear operators. Then
\begin{align}
\frac{1}{|Y|} \sum_{y \in Y} \Tr [W_y A W_y^\dag B ] \Tr [ W_y C W_y^\dag D] ={}& \int d\mu(W) \Tr [W A W^\dag B ] \Tr [ W C W^\dag D] \notag \\
={}& \frac{1}{d^2-1} \left( \Tr[AC] \Tr[BD] + \Tr[A]\Tr[B]\Tr[C]\Tr[D] \right) \notag \\
-{}& \frac{1}{d(d^2-1)} \left( \Tr[A]\Tr[C] \Tr[BD] + \Tr[AC]\Tr[B] \Tr[D] \right) . \label{vqacr_qntk_tdesign_trwawbtrwcwd_eq}
\end{align}
\end{lemma}

\subsection{Lemmas about variational quantum algorithms}

\begin{lemma}\label{vqacr_qntk_jac_form}
Let $J_{da}(\bm{\theta})=\frac{\partial {z}_a}{\partial \theta_d}(\bm{\theta})$ be the partial derivative of $z_a(\bm{\theta})$ in Eq.~(\ref{vqacr_qntk_loss_i_eq}), where Hamiltonians $\{H_d\}_{d=1}^{D}$ in the ansätz $V(\bm{\theta})=\prod_{d=D}^{1} \exp(-iH_{d}\theta_{d}/2) W_d $ are Hermitian unitary matrices.  Let $V_{d:d'}$  be $ \exp(-iH_{d}\theta_{d}/2) W_d \exp(-iH_{d-1}\theta_{d-1}/2) W_{d-1} \cdots \exp(-iH_{d'}\theta_{d'}/2) W_{d'}$, then we have
\begin{equation}\label{vqacr_qntk_jac_form_eq}
J_{da}(\bm{\theta}) = \frac{\partial z_a}{\partial \theta_d} (\bm{\theta})  = -\frac{i}{2} \Tr \left[  \big[ {V_{D:d+1}}^\dag O V_{D:d+1}  , H_d  \big] V_{d:1} \rho_a {V_{d:1}}^\dag  \right] . 
\end{equation}
\end{lemma}

\begin{proof}

We have
\begin{align}
J_{da} (\bmt) ={}& \frac{\partial }{ \partial \theta_{d}}  \Tr \left[ O V_{D:{d}+1} \exp(-iH_{d}\theta_{d}/2) W_d V_{d-1:1} \rho_a {V_{d-1:1}}^\dag {W_d }^\dag \exp(iH_{d}\theta_{d}/2) {V_{D:d+1}}^\dag  \right] \label{vqacr_qntk_jac_form_1_1} \\
={}& \frac{\partial }{ \partial \theta_{d}}  \Tr \left[  {V_{D:d+1}}^\dag O V_{D:d+1} \exp(-iH_{d}\theta_{d}/2) W_d V_{d-1:1} \rho_a {V_{d-1:1}}^\dag {W_d}^\dag \exp(iH_{d}\theta_{d}/2) \right] \label{vqacr_qntk_jac_form_1_2} \\
={}& \Tr \left[ {V_{D:d+1}}^\dag O V_{D:d+1} \Big( -\frac{i}{2}H_{d} \exp(-iH_{d}\theta_{d}/2)  {W_d} V_{d-1:1} \rho_a {V_{d-1:1}}^\dag {W_d}^\dag \exp(iH_{d}\theta_{d}/2) \Big) \right] \notag \\
+{}& \Tr \left[ {V_{D:d+1}}^\dag O V_{D:d+1} \Big( \exp(-iH_{d}\theta_{d}/2) {W_d} V_{d-1:1} \rho_a {V_{d-1:1}}^\dag {W_d}^\dag \exp(iH_{d}\theta_{d}/2) \frac{i}{2}H_{d} \Big) \right] \label{vqacr_qntk_jac_form_1_3} \\
={}& -\frac{i}{2} \Tr \left[ {V_{D:d+1}}^\dag O V_{D:d+1}  H_{d} V_{d:1} \rho_a {V_{d:1}}^\dag  \right] + \frac{i}{2} \Tr \left[ {V_{D:d+1}}^\dag O V_{D:d+1} V_{d:1} \rho_m {V_{d:1}}^\dag  H_{d}  \right] \notag \\
={}& -\frac{i}{2} \Tr \left[ {V_{D:d+1}}^\dag O V_{D:d+1} \big[ H_{d}  ,  V_{d:1} \rho_a {V_{d:1}}^\dag   \big] \right] \notag \\
={}& -\frac{i}{2} \Tr \left[  \big[ {V_{D:d+1}}^\dag O V_{D:d+1}  , H_{d}  \big] V_{d:1} \rho_a {V_{d:1}}^\dag  \right] .\label{vqacr_qntk_jac_form_1_5} 
\end{align}

Eq.~(\ref{vqacr_qntk_jac_form_1_1}) follows from Eq.~(\ref{vqacr_qntk_loss_i_eq}) and notations $V_{D:d+1}$, $V_{d:1}$. Eq.~(\ref{vqacr_qntk_jac_form_1_2}) is derived by using $\Tr [AB] = \Tr [BA]$. Eq.~(\ref{vqacr_qntk_jac_form_1_3}) follows from the matrix derivative. Eq.~(\ref{vqacr_qntk_jac_form_1_5}) is derived by using the trace property 
\begin{equation*}
\Tr[A[B,C]]=\Tr[ABC]-\Tr[ACB]=\Tr[ABC]-\Tr[BAC]=\Tr[[A,B]C].
\end{equation*} 
Thus, we have proved Eq.~(\ref{vqacr_qntk_jac_form_eq}).

\end{proof}

\begin{lemma}
\label{vqacr_lemma_pauli_decom_single}
Let $X=\sum_{{i} \in \{0,1,2,3\}} \frac{1}{2} c_{X,i}\sigma_{{i}}$ and $Y=\sum_{{j} \in \{0,1,2,3\}} \frac{1}{2} c_{Y,j}\sigma_{{j}}$ be the Pauli decomposition of matrices $X$ and $Y$ with dimension $(2,2)$. Denote by $\bm{c}_{XY}$ the Pauli decomposition coefficient of $XY$, then
\begin{equation}\label{vqacr_lemma_pauli_decom_single_eq}
\bm{c}_{XY} = \frac{1}{2} B \left( \bm{c}_{X} \otimes \bm{c}_{Y} \right) ,
\end{equation}
where
\begin{equation}\label{vqacr_lemma_pauli_decom_single_B_k_ij}
B = \left[ B_{ijk} \right]_{k,(ij)} ={} \left[\begin{array}{cccccccccccccccc}
1 & & & & & 1 & & & & & 1 & & & & & 1 \\
& 1 & & & 1 & & & & & & & i & & & -i & \\
& & 1 & & & & & -i & 1 & & & & & i & & \\
& & & 1 & & & i & & & -i & & & 1 & & & 
\end{array} \right].
\end{equation}
\end{lemma}

\begin{proof}

We have 

\begin{align*}
XY ={}& \frac{1}{4} \sum_{i,j=0}^{3} x_i \sigma_i y_j \sigma_j \\
={}& \frac{1}{4} \sum_{i,j,k=0}^{3} x_i y_j B_{ijk} \sigma_k,
\end{align*}
where $B$ is a tensor with dimension $(4,4,4)$. By using the property of Pauli matrices, we can obtain entries of $B$ as
\begin{equation}\label{vqacr_lemma_pauli_decom_single_B_ijk}
B_{ijk} ={} \left\{
\begin{aligned}
{}& 1,\ & \text{when one of $\{i,j,k\}$ is $0$ and the other two index equals to each other;} \\
{}& i \epsilon_{ijk}, \ & \text{when $\{i,j,k\}$ differs from each other;} \\
{}& 0 , \ & \text{other cases;}
\end{aligned}
\right.
\end{equation}
where $\epsilon_{ijk}$ is the Levi-Civita symbol. 
By rewriting the tensor $B$ as a matrix with dimension $(4,16)$, we obtain Eq.~(\ref{vqacr_lemma_pauli_decom_single_B_k_ij}).

\end{proof}

\begin{lemma}
\label{vqacr_lemma_pauli_decom}
Let $X=\sum_{\bm{i} \in \{0,1,2,3\}^{\otimes N}} \frac{1}{2^N} c_{X,\bm{i}}\sigma_{\bm{i}}$ and $Y=\sum_{\bm{j} \in \{0,1,2,3\}^{\otimes N}} \frac{1}{2^N} c_{Y,\bm{j}}\sigma_{\bm{j}}$ be the Pauli decomposition of matrices $X$ and $Y$ with dimension $(2^N,2^N)$. Denote by $\bm{c}_{XY}$ the Pauli decomposition coefficient of $XY$, then
\begin{equation}\label{vqacr_lemma_pauli_decom_eq}
\bm{c}_{XY} = \frac{1}{2^N} B^{\otimes} \left( \bm{c}_{X} \otimes \bm{c}_{Y} \right) ,
\end{equation}
where $B^\otimes$ denotes the tensor product of $B$ in Lemma~\ref{vqacr_lemma_pauli_decom_single} under the index order $(i_1,\cdots,i_N, j_1,\cdots, j_N, k_1, \cdots, k_N)$.
\end{lemma}

\begin{proof}
Eq.~(\ref{vqacr_lemma_pauli_decom_eq}) could be proved by induction by using Lemma~\ref{vqacr_lemma_pauli_decom_single}.
\end{proof}

\subsection{Lemmas about random matrix theory}

\begin{lemma}[Covariance matrix estimation, Theorem~5.44 in \cite{arxiv_vershynin2010introduction}]\label{vqacr_lemma_qntk_operator_norm_of_covariance_matrix}
Let $X_1, \cdots, X_{d}$ be independent random vectors in $\R^{m}$ with the common second moment matrix $\Sigma=\E X_i X_i^T$. Suppose $\|X_i\|_2 \leq \sqrt{n}$ almost surely $\forall i$. Then for every $t \geq 0$, the following inequality holds with probability at least $1-m\exp(-ct^2)$:
\begin{equation*}
 \Big\| \frac{1}{d} \sum_{i=1}^{d} X_i X_i^T - \Sigma \Big\|_2 \leq \max \left( \|\Sigma\|_2^{1/2} \epsilon, \epsilon^2 \right) \ where \ \epsilon = t \sqrt{\frac{n}{d}}.
\end{equation*}
Here $c>0$ is an absolute constant.
\end{lemma}

\begin{lemma}[Derived from Theorem~5.62 in \cite{arxiv_vershynin2010introduction}]\label{vqacr_lemma_qntk_spectral_norm_of_independent_columns}
Let $A$ be a $d_1 \times d_2$ matrix ($d_1 \geq d_2$) whose columns $A_i$ are independent isotropic random vectors in $\R^{d_1}$ with $\|A_j\|_2 = \sqrt{d_1}$ almost surely. Let $$d_3:=\frac{1}{d_1} \E \max_{j \leq d_2} \sum_{k \in [d_2], k \neq j} \left( A_j^T A_k \right)^2.$$
Then $\E \| \frac{1}{d_1} A^T A - I \| \leq C_0 \sqrt{\frac{d_3 \log d_2}{d_1}}$. In particular, 
\begin{equation*}
\E \max_{j \leq d_2} \left| s_j(A) - \sqrt{d_1} \right| \leq{} C \sqrt{d_3 \log d_2} ,
\end{equation*}
where $C_0, C$ are absolute constants.
\end{lemma}

\section{Generalization error of quantum machine learning with random quantum data}

\label{vqacr_geb_qkm_haar_input}

This section is composed of two parts, i.e., the generalization error analysis of quantum machine learning using quantum kernel method and quantum neural network, respectively. Specifically, we consider the generalization performance when using $2$-design quantum states as input datasets. 

\subsection{Generalization error of quantum kernel method}

First, we study the generalization error of quantum kernel method. To begin with, we introduce some notations for convenience.
The training and the test set are denoted by $\mathcal{A}=\{a\}=\{(\rho_a, y_a)\}$ and $\mathcal{B}=\{\hat{b}\}=\{(\rho_{\hat{b}}, y_{\hat{b}})\}$, respectively.
We consider the linear model that gives the prediction
\begin{equation}\label{vqacr_geb_qntk_prediction_test_kernel}
z_{\hat{b}} (\bmt) = \sum_{a \in \A} {k}(\rho_a, \rho_{\hat{b}}) \theta_{a},
\end{equation}
where $\bmt \in \R^{|\A|}$ is the parameter and $k(\rho_a, \rho_{\hat{b}})=\Tr [\rho_a \rho_{\hat{b}}]$ is the quantum kernel. For convenience, we denote $k_{a\hat{b}} = k(\rho_a, \rho_{\hat{b}})$. We use $k$ as the $|\A|\times|\A|$ quantum kernel on the training set $\A$.

Suppose the loss is defined by the $\ell_2$-norm distance between the prediction vector $\bm{z}$ and the label vector $\bm{y}$. Let $\bmt(t)$ and $\bm{z}(t)$ be the parameter and the prediction vector on training or test datasets at the optimization step $t$, respectively. Obviously we have
\begin{equation}\label{vqacr_geb_qntk_params_zero}
\bmt(0) = {k}^{-1} \bm{z}_{\A}(0)
\end{equation}
from Eq.~(\ref{vqacr_geb_qntk_prediction_test_kernel}). After the training, the linear model in Eq.~(\ref{vqacr_geb_qntk_prediction_test_kernel}) can be optimized to achieve zero training loss with the parameter
\begin{equation}\label{vqacr_geb_qntk_params_infty}
\bmt(\infty) = {k}^{-1} \bm{y}_{\A} .
\end{equation}
At the optimal parameter point, the prediction of a sample $\hat{b}=(\rho_{\hat{b}}, y_{\hat{b}})$ from the test dataset is given by
\begin{align}
z_{\hat{b}} (\infty) ={}& \sum_{a'\in \A} {k}_{\hat{b} a'}  \theta_{a'} (\infty) \tag{from Eq.~(\ref{vqacr_geb_qntk_prediction_test_kernel})} \\
={}& \sum_{a,a'\in \A} {k}_{\hat{b} a'} [{k}^{-1} ]_{a'a} {y}_{a}  \tag{from Eq.~(\ref{vqacr_geb_qntk_params_infty})} \\
={}& \sum_{a,a'\in \A} {k}_{\hat{b} a'} [{k}^{-1} ]_{a'a} \left({z}_{a}(0)-{r}_{a}(0) \right)  \tag{from $r=z-y$} \\
={}& \sum_{a,a'\in \A} {k}_{\hat{b} a'} \theta_{a'}(0) - \sum_{a,a' \in \mathcal{A}} {k}_{\hat{b}a'} [{k}^{-1}]_{a'a} r_{a} (0)  \tag{from Eq.~(\ref{vqacr_geb_qntk_params_zero})} \\
={}& z_{\hat{b}} (0) - \sum_{a,a' \in \mathcal{A}} {k}_{\hat{b}a'} [{k}^{-1}]_{a'a} r_{a} (0) ,
\label{vqacr_geb_qntk_prediction_test}
\end{align}
where $\bm{r}_{\A}(t):= \bm{z}_{\A}(t)-\bm{y}_{\A}$ is the training residual vector at the optimization step $t$.
So, the residual of the sample $\hat{b}$ after training is given by,
\begin{equation}\label{vqacr_geb_qntk_res_error_test}
r_{\hat{b}} (\infty) = z_{\hat{b}} (\infty) - y_{\hat{b}} = r_{\hat{b}} (0) - \sum_{a,a' \in \mathcal{A}} {k}_{\hat{b}a'} [{k}^{-1}]_{a'a} r_{a} (0) .
\end{equation}

We denote $\L(t):=\L(\bmt(t))$ for convenience. Therefore, we obtain the generalization error,
\begin{align}
\mathcal{E}_{\mathcal{B}|\mathcal{A}} ={}& \mathcal{L}_{\mathcal{B}} (\infty) - \mathcal{L}_{\mathcal{A}} (\infty) = \mathcal{L}_{\mathcal{B}} (\infty) = \frac{1}{2|\mathcal{B}|} \sum_{\hat{b} \in \mathcal{B}} r_{\hat{b}}^2 (\infty)  \notag \\
={}& \frac{1}{2|\mathcal{B}|} \sum_{\hat{b} \in \mathcal{B}} \left( r_{\hat{b}} (0) - \sum_{a,a' \in \mathcal{A}} {k}_{\hat{b}a'} [{k}^{-1}]_{a'a} r_{a} (0) \right)^2 \\
={}& \frac{1}{2|\mathcal{B}|} \sum_{\hat{b} \in \mathcal{B}} r_{\hat{b}}^2 (0) + \frac{1}{2|\mathcal{B}|} \sum_{\hat{b} \in \mathcal{B}} \left( \sum_{a,a' \in \mathcal{A}} {k}_{\hat{b}a'} [{k}^{-1}]_{a'a} r_{a} (0) \right)^2 \\
-{}& \frac{1}{|\mathcal{B}|} \sum_{\hat{b} \in \mathcal{B}} \sum_{a,a' \in \mathcal{A}} {k}_{\hat{b}a'} [{k}^{-1}]_{a'a} r_{\hat{b}} (0) r_{a} (0) .
\end{align}
We remark that the first term is the loss on the set $\mathcal{B}$ before the training. To show the improvement of the prediction on the test dataset, we define the following quantity as the relative generalization error,
\begin{align}
\mathcal{R}_{\mathcal{B}|\mathcal{A}} :={}& \mathcal{E}_{\mathcal{B}|\mathcal{A}} - \mathcal{L}_{\mathcal{B}} (0) \label{vqacr_geb_qntk_re_error_eq_1} \\
={}& \frac{1}{2|\mathcal{B}|} \sum_{\hat{b} \in \mathcal{B}} \left( \sum_{a,a' \in \mathcal{A}} {k}_{\hat{b}a'} [{k}^{-1}]_{a'a} r_{a} (0) \right)^2 - \frac{1}{|\mathcal{B}|} \sum_{\hat{b} \in \mathcal{B}} \sum_{a,a' \in \mathcal{A}} {k}_{\hat{b}a'} [{k}^{-1}]_{a'a} r_{\hat{b}} (0) r_{a} (0) \label{vqacr_geb_qntk_re_error_eq_2} \\
={}& \frac{1}{2|\mathcal{B}|} \sum_{\hat{b} \in \mathcal{B}, a_1, a_2 \in \mathcal{A}} g_{\hat{b} a_1} g_{\hat{b} a_2} r_{a_1}(0) r_{a_2}(0) - \frac{1}{|\mathcal{B}|} \sum_{\hat{b} \in \mathcal{B}, a \in \mathcal{A}} g_{\hat{b} a} r_{\hat{b}} (0) r_{a}(0) , \label{vqacr_geb_qntk_re_error_eq_3}
\end{align}
where we use the metric
\begin{equation}\label{vqacr_geb_qntk_metric_g_eq}
g_{\hat{b} a} = \sum_{a'\in \mathcal{A}} {k}_{\hat{b} a'} [{k}^{-1}]_{a'a}.
\end{equation}
When $\mathcal{R}_{\mathcal{B}|\mathcal{A}} \geq 0$, the element in set $\mathcal{B}$ has equal or worse prediction averagely after the training, i.e., we face the overfitting issue. Therefore we hope $\mathcal{R}_{\mathcal{B}|\mathcal{A}} \leq -C$, where the constant $C>0$ could be as large as possible.

Next, we analyze the scaling of generalization error when input datasets are sampled from state $2$-designs. Specifically, suppose states $\rho_a=|\psi_a\>\<\psi_a|$ in training and test datasets are independent $2$-design states, i.e., $|\psi_a\>=U_a |0\>$ and $U_a$s are randomly chosen from $\mathcal{U}(2^N)$ according to the unitary $2$-design. The quantum kernel method gives 
\begin{align*}
k_{ab}={}& \Tr [\rho_a \rho_b] = |\<\psi_a|\psi_b\>|^2.
\end{align*}

Since input states are engaged in the quantum kernel, different state distributions could naturally induce different training and generalization performances. 
Here, we repeat the generalization error result in the main text in Theorem~\ref{vqacr_geb_qkm_haar_input_theorem} for convenience, where the notation $\L(t):=\L(\bmt(t))$ is employed. The generalization error only has exponentially small improvement from the extremely overfitting situation.

\begin{theorem}\label{vqacr_geb_qkm_haar_input_theorem}
Suppose all $N$-qubit quantum states in the training and test datasets $\A$ and $\B$ are independently sampled from state 2-designs and the size of datasets is smaller than $2^{N/2}$. Let $\L_{\A}(t)$ and $\L_{\B}(t)$ be the training and the test loss function at step $t$ of the quantum kernel method, respectively. Then, with high probability,
\begin{equation}
\mathop{\E} \L_{\B}(\infty) \gtrsim{} \mathop{\E} \mathcal{L}_{\B} (0) - \frac{|\A|}{2^{N-1}} \mathop{\E} \sqrt{ \mathcal{L}_{\A}(0) \mathcal{L}_{\B}(0) } ,
\end{equation}
where the expectation is taken under state $2$-designs for states in $\A$ and $\B$.
\end{theorem}

\begin{proof}

The main idea is to find that the metric $g_{\hat{b}a}$ in Eq.~(\ref{vqacr_geb_qntk_metric_g_eq}) is exponentially small on average. To achieve this, we first prove that for Haar-randomly distributed input states, the quantum kernel is close to the identity matrix with an additional exponentially small term. We refer to the Pauli decomposition form of quantum states with random matrix theory. Specifically, let $A\in\R^{(4^N-1) \times |\A|}$ be the coefficient matrix that stores coefficients for input states $\{|\psi_a\>\}_{a \in \A}$,
\begin{equation}\label{vqacr_geb_qkm_haar_input_theorem_A_ia}
A_{\bm{i},a} = \frac{1}{2^N} \Tr \left[ |\psi_a\>\<\psi_a| \sigma_{\bm{i}} \right],
\end{equation}
where $\bm{i}\in \{0,1,2,3\}^{\otimes N}/\bm{0}$ denotes the index of Pauli matrices except the identity matrix: $\bm{i} = (i_1, i_2, \cdots, i_N)$. We remark that the coefficient for the identity matrix is always $1/2^N$ since $\Tr [|\psi_a\>\<\psi_a|]=1$, and the density matrix could be recovered as
\begin{equation}\label{vqacr_geb_qkm_haar_input_theorem_psi_a}
| \psi_a\>\<\psi_a| = \frac{1}{2^N} \sigma_{\bm{0}} + \sum_{\bm{i} \neq \bm{0}} A_{\bm{i},a} \sigma_{\bm{i}}.
\end{equation}
Thus, the quantum kernel can be written as 
\begin{align}
{k}_{a'a} ={}& \Tr \left[ |\psi_{a'}\>\<\psi_{a'}| \psi_a\>\<\psi_a| \right] \notag \\
={}& \frac{1}{2^{N}} + \sum_{\bm{i}, \bm{j}} A_{\bm{i},a'} A_{\bm{j},a} \Tr [\sigma_{\bm{i}}\sigma_{\bm{j}}] \tag{using Eq.~(\ref{vqacr_geb_qkm_haar_input_theorem_psi_a})} \\
={}& \frac{1}{2^{N}} + 2^N \left[ A^T A \right]_{a'a} . \label{vqacr_geb_qkm_haar_input_theorem_K_aprimea}
\end{align}
The matrix formulation of Eq.~(\ref{vqacr_geb_qkm_haar_input_theorem_K_aprimea}) is
\begin{align}
{k} ={}& 2^N A^T A + \frac{1}{2^{N}} \bm{e} \bm{e}^T, \label{vqacr_geb_qkm_haar_input_theorem_K_aprimea_matrix}
\end{align}
where $\bm{e}=(1,1,\cdots,1)^T \in \R^{|\mathcal{A}|}$.

For random input states sampled independently from state $2$-designs, the matrix $A$ is a random matrix with independent columns, so we could apply Lemma~\ref{vqacr_lemma_qntk_spectral_norm_of_independent_columns} for analyzing the kernel property. To match the isotropic column condition in Lemma~\ref{vqacr_lemma_qntk_spectral_norm_of_independent_columns}, we introduce the matrix 
\begin{equation}\label{vqacr_geb_qkm_haar_input_theorem_B}
B = 2^N \sqrt{2^N+1} A.
\end{equation}
Matrix $B$ has independent isotropic columns since
\begin{align}
\mathop{\E}_{\rm 2-design} \left[B_{a'} B_{a}^T \right]_{\bm{i}\bm{j}} ={}& \mathop{\E}_{\rm 2-design} 2^{2N} (2^N+1) A_{\bm{i},a'} A_{\bm{j},a} \tag{using Eq.~(\ref{vqacr_geb_qkm_haar_input_theorem_B})} \\
={}& \mathop{\E}_{\rm 2-design} (2^N+1) \Tr \left[ |\psi_{a'}\>\<\psi_{a'}| \sigma_{\bm{i}} \right] \Tr \left[ |\psi_a\>\<\psi_a| \sigma_{\bm{j}} \right] \tag{using Eq.~(\ref{vqacr_geb_qkm_haar_input_theorem_A_ia})} \\
={}& \delta_{a'a} \delta_{\bm{i} \bm{j}} . \tag{using Lemmas~\ref{vqacr_qntk_tdesign_tr_wawb} and \ref{vqacr_qntk_tdesign_trwawbtrwcwd} }
\end{align}
The $\ell_2$-norm of $B_{a}$ is $\sqrt{4^N-1}$ since
\begin{align}
\| B_{a} \|_2^2 ={}& 4^{N} (2^N+1) \sum_{\bm{i}} A_{\bm{i},a}^2 \tag{using Eq.~(\ref{vqacr_geb_qkm_haar_input_theorem_B})} \\
={}& 4^{N} (2^N+1) \frac{1}{2^N} \left( \Tr \left[|\psi_a\>\<\psi_a|^2\right] - \frac{1}{2^N}\right) \tag{using Eq.~(\ref{vqacr_geb_qkm_haar_input_theorem_psi_a})} \\
={}& 4^N -1 . \notag
\end{align}
Moreover, we have
\begin{align}
{}& \frac{1}{4^N-1} \mathop{\E}_{\rm 2-design} \max_{a} \sum_{a'\neq a} \left( B_{a}^T  B_{a'} \right)^2 \notag \\
\leq{}& \frac{1}{4^N-1}\mathop{\E}_{\rm 2-design} \sum_{a\in \A} \sum_{a'\neq a} \left( B_{a}^T  B_{a'} \right)^2 \\
={}& \frac{16^{N}(2^N+1)}{2^N-1} \mathop{\E}_{\rm 2-design} \sum_{a\in \A} \sum_{a'\neq a} \left( \sum_{\bm{i}} A_{\bm{i}, a}  A_{\bm{i}, a'} \right)^2 \tag{using Eq.~(\ref{vqacr_geb_qkm_haar_input_theorem_B})} \\
={}& \frac{16^{N}(2^N+1)}{2^N-1} \mathop{\E}_{\rm 2-design} \sum_{a\in \A} \sum_{a'\neq a} \frac{1}{4^N} \left( \Tr \left[ |\psi_a\>\<\psi_a|\psi_{a'}\>\<\psi_{a'}| \right] - \frac{1}{2^N} \right)^2 \tag{using Eq.~(\ref{vqacr_geb_qkm_haar_input_theorem_psi_a})} \\
={}& \frac{4^{N}(2^N+1)}{2^N-1} \sum_{a\in \A} \sum_{a'\neq a} \left[ \frac{2}{2^N(2^N+1)} - \frac{2}{2^{2N}} + \frac{1}{2^{2N}} \right] \tag{using Lemmas~\ref{vqacr_qntk_tdesign_tr_wawb} and \ref{vqacr_qntk_tdesign_trwawbtrwcwd} } \\
={}& |\A|^2-|\A|  \label{vqacr_geb_qkm_haar_input_theorem_cond3}
\end{align}

By using Lemma~\ref{vqacr_lemma_qntk_spectral_norm_of_independent_columns} with Eq.~(\ref{vqacr_geb_qkm_haar_input_theorem_cond3}), the following holds,
\begin{equation}\label{vqacr_geb_qkm_haar_input_theorem_B_bound}
\mathop{\E}_{\rm 2-design} \left\| \frac{1}{4^{N}-1} B^T B - I \right\|_F \leq C_0 \sqrt{\frac{(|\A|^2-|\A|) \ln |\A| }{4^{N}-1}},
\end{equation}
where $C_0$ is an absolute constant. Let the matrix $X=\frac{1}{4^{N}-1} B^T B - I$. Using Eq.~(\ref{vqacr_geb_qkm_haar_input_theorem_B_bound}) in the quantum kernel formulation, we have
\begin{equation}\label{vqacr_geb_qkm_haar_input_theorem_K_X}
{k} = I + \frac{1}{2^N} \left( \bm{e} \bm{e}^T - I \right) + \left( 1 - \frac{1}{2^N} \right) X,
\end{equation}
where $\E \|X\|_F \leq C_0 \sqrt{\frac{(|\A|^2-|\A|) \ln |\A| }{4^{N}-1}}$ by using Eqs.~(\ref{vqacr_geb_qkm_haar_input_theorem_K_aprimea_matrix}) and (\ref{vqacr_geb_qkm_haar_input_theorem_B}). Since $|\A|^2 \leq 2^N$, we use the norm bound of $X$ as $\|X\|_F \leq \O(2^{-N})$ with high probability. Therefore the inverse
\begin{equation}\label{vqacr_geb_qkm_haar_input_theorem_K_inverse}
\left[{k}^{-1}\right]_{a'a} = \delta_{a'a} - \frac{1}{2^N} \left( 1- \delta_{a'a} \right) - X_{a'a} + \O \left(\frac{1}{4^N} \right) = 2 \delta_{a'a} - {k}_{a'a}  + \O \left(\frac{1}{4^N} \right)
\end{equation}
with high probability.
By employing Eq.~(\ref{vqacr_geb_qkm_haar_input_theorem_K_inverse}) in the metric $g$, we have
\begin{align}
g_{\hat{b} a} ={}& \sum_{a' \in \A} {k}_{\hat{b}a'} \left( 2\delta_{a'a} - {k}_{a'a} + \O \left(\frac{1}{4^N} \right) \right) \notag \\
={}& \O \left(\frac{|\A|}{4^N} \right) + \sum_{a' \in \A} {k}_{\hat{b}a'} \left( 2\delta_{a'a} - {k}_{a'a} \right)  \label{vqacr_geb_qkm_haar_input_theorem_g_ba_fin}
\end{align}
which holds with high probability for all $\hat{b}$ and $a$.

Then, we proceed to analyze the relative generalization error. By using Eq.~(\ref{vqacr_geb_qntk_re_error_eq_3}), we have
\begin{align}
\mathop{\E} \left[ \mathcal{E}_{\mathcal{B}|\mathcal{A}} - \mathcal{L}_{\B} (0) \right ]
= \frac{1}{2|\mathcal{B}|} \sum_{\hat{b} \in \mathcal{B}, a_1, a_2 \in \mathcal{A}} \mathop{\E}_{\hat{b}} \mathop{\E}_{a_1} \mathop{\E}_{a_2} \left[ g_{\hat{b} a_1} g_{\hat{b} a_2} r_{a_1}(0) r_{a_2}(0) \right] - \frac{1}{|\mathcal{B}|} \sum_{\hat{b} \in \mathcal{B}, a \in \mathcal{A}} \mathop{\E}_{\hat{b}} \mathop{\E}_{a} \left[ g_{\hat{b} a} r_{\hat{b}} (0) r_{a}(0) \right] , \label{vqacr_geb_qkm_haar_input_theorem_r_1_1}
\end{align}
where $\E_{a}$ or $\E_{\hat{b}}$ denotes the expectation of unitary 2-design distribution with respect to the unitary $U$ that generates $U|0\>=|\psi_a\>$ or $U|0\>=|\psi_{\hat{b}}\>$. 

Next, we analyze two terms in Eq.~(\ref{vqacr_geb_qkm_haar_input_theorem_r_1_1}) separately. First, we integrate the distribution with respect to $\hat{b}$ in Eq.~(\ref{vqacr_geb_qkm_haar_input_theorem_r_1_1}). With high probability, we have
\begin{align}
\mathop{\E}_{\hat{b}} \left[ g_{\hat{b} a_1} g_{\hat{b} a_2} \right] ={}& \O \left(\frac{|\A|^2}{16^N} \right) + \O \left(\frac{|\A|}{4^N} \right) \sum_{a \in \{a_1,a_2\}} \sum_{a'\in\A} \mathop{\E}_{\hat{b}} \left[ {k}_{\hat{b}a'} \right] \left( 2\delta_{a'a} - {k}_{a'a} \right) \notag \\
+{}& \sum_{a_1', a_2' \in \A} \mathop{\E}_{\hat{b}} \left[ {k}_{\hat{b} a_1'} {k}_{\hat{b} a_2'} \right] \left( 2\delta_{a_1' a_1} - {k}_{a_1'a_1} \right) \left( 2\delta_{a_2' a_2} - {k}_{a_2'a_2} \right) \label{vqacr_geb_qkm_haar_input_theorem_r_2_1} \\
={}& \O \left(\frac{|\A|^2}{16^N} \right) + \O \left(\frac{|\A|^2}{8^N} \right)  \sum_{a \in \{a_1,a_2\}} \sum_{a'\in\A} \left( 2\delta_{a'a} - {k}_{a'a} \right) + \frac{1}{2^N(2^N+1)} \left[ (2I-{k}) {k} (2I-{k}) \right]_{a_1 a_2} \notag \\
+{}& \frac{1}{2^N(2^N+1)} \sum_{a_1', a_2' \in \A}   \left( 2\delta_{a_1' a_1} - {k}_{a_1'a_1} \right) \left( 2\delta_{a_2' a_2} - {k}_{a_2'a_2} \right)
\label{vqacr_geb_qkm_haar_input_theorem_r_2_2} \\
\simeq{}& \frac{1}{4^N}  \left( 2\delta_{a_1 a_2} - {k}_{a_1 a_2} \right)  +  \frac{1}{4^N} \sum_{a_1', a_2' \in \A}   \left( 2\delta_{a_1' a_1} - {k}_{a_1'a_1} \right) \left( 2\delta_{a_2' a_2} - {k}_{a_2'a_2} \right)  . \label{vqacr_geb_qkm_haar_input_theorem_r_2_3}
\end{align}
Eq.~(\ref{vqacr_geb_qkm_haar_input_theorem_r_2_1}) follows from Eq.~(\ref{vqacr_geb_qkm_haar_input_theorem_g_ba_fin}). Eq.~(\ref{vqacr_geb_qkm_haar_input_theorem_r_2_2}) follows from Lemmas~\ref{vqacr_qntk_tdesign_tr_wawb} and \ref{vqacr_qntk_tdesign_trwawbtrwcwd}, which give $\mathop{\E}_{\hat{b}} \left[ {k}_{\hat{b}a'} \right]=2^{-N}$ and $\mathop{\E}_{\hat{b}} \left[ {k}_{\hat{b} a_1'} {k}_{\hat{b} a_2'} \right]=\frac{1}{2^N(2^N+1)} \left[{k}_{a_1' a_2'} +1 \right]$, respectively. Eq.~(\ref{vqacr_geb_qkm_haar_input_theorem_r_2_3}) follows from ${k}^{-1} \approx 2I-{k}$ and the condition $|\A|^2 \leq 2^N$. Thus, by writing the summation as the matrix-vector production form, we have
\begin{align}
{}& \sum_{a_1,a_2 \in \A} \mathop{\E}_{\hat{b}} \mathop{\E}_{a_1} \mathop{\E}_{a_2} \left[ g_{\hat{b} a_1} g_{\hat{b} a_2} r_{a_1}(0) r_{a_2}(0) \right] \notag \\
\simeq{}& \frac{1}{4^N} \mathop{\E}_{\A} \bm{r}_{\A}(0)^T \left( 2I -{k} \right) \bm{r}_{\A}(0) + \frac{1}{4^N} \left[ \mathop{\E}_{\A} \bm{e}^T \left( 2I-{k} \right) \bm{r}_{\A}(0) \right]^2 \tag{using Eq.~(\ref{vqacr_geb_qkm_haar_input_theorem_r_2_3})} \\
={}& \frac{1}{4^N} \mathop{\E}_{\A} \left( \bmt(0)^T {k} - \bm{y}_{\A}^T \right) \left( 2I -{k} \right) \left( {k} \bmt(0) - \bm{y}_{\A} \right) + \frac{1}{4^N} \left[ \mathop{\E}_{\A} \bm{e}^T \left( 2I-{k} \right) \left( {k} \bmt(0) - \bm{y}_{\A} \right) \right]^2 \tag{using $\bm{r}=\bm{z}-\bm{y}$ and Eq.~(\ref{vqacr_geb_qntk_params_zero})} \\
\simeq{}& \frac{1}{4^N} \mathop{\E}_{\A} \left[ \bmt(0)^T {k} \bmt(0) - \bm{y}_{\A}^T {k} \bm{y}_{\A} +2 \bm{y}_{\A}^T \bm{y}_{\A} -2\bmt(0)^T \bm{y}_{\A} \right] \notag \\
+{}& \frac{1}{4^N} \left\{ \bm{e}^T \bmt(0) -2 \bm{e}^T \bm{y}_{\A} +  \mathop{\E}_{\A} \bm{e}^T {k} \bm{y}_{\A} \right\}^2 \tag{using $k^{-1} \approx 2I-k$} \\
\simeq{}& \frac{1}{4^N} \left\{  \bmt(0)^T \left[ \left(1-\frac{1}{2^N} \right) I + \frac{1}{2^N} \bm{e} \bm{e}^T \right] \bmt(0) - \bm{y}_{\A}^T \left[ \left(1-\frac{1}{2^N} \right) I + \frac{1}{2^N} \bm{e} \bm{e}^T \right] \bm{y}_{\A} +2 \bm{y}_{\A}^T \bm{y}_{\A} -2\bmt(0)^T \bm{y}_{\A} \right\} \notag \\
+{}& \frac{1}{4^N} \left\{ \bm{e}^T \bmt(0) -2 \bm{e}^T \bm{y}_{\A} +  \mathop{\E}_{\A} \bm{e}^T \left[ \left(1-\frac{1}{2^N} \right) I + \frac{1}{2^N} \bm{e} \bm{e}^T \right] \bm{y}_{\A} \right\}^2 \tag{using Eq.~(\ref{vqacr_geb_qkm_haar_input_theorem_K_X}) and $|\A|^2 \leq 2^N$} \\
\simeq{}& \frac{1}{4^N} \left\{ \|\bmt(0)-\bm{y}_{\A}\|_2^2 + \left[\bm{e}^T (\bmt(0)-\bm{y}_{\A}) \right]^2 \right\} . \label{vqacr_geb_qkm_haar_input_theorem_r_2_9}
\end{align}

Finally we proceed to the second term in Eq.~(\ref{vqacr_geb_qkm_haar_input_theorem_r_1_1}). With high probability, we have
\begin{align}
\mathop{\E}_{\hat{b}} \left[ g_{\hat{b} a} r_{\hat{b}} (0) \right] ={}& \mathop{\E}_{\hat{b}} \left[ g_{\hat{b} a} \left( - y_{\hat{b}} + \sum_{a'\in \A} {k}_{\hat{b}a'} \theta_{a'} (0)\right) \right] \tag{using $r_{\hat{b}}(0)=z_{\hat{b}}(0) - y_{\hat{b}} = \sum_{a'\in \A} {k}_{\hat{b}a'} \theta_{a'}(0) - y_{\hat{b}}$} \\
={}& \mathop{\E}_{\hat{b}} \left[ \left( \O \left(\frac{|\A|}{4^N} \right) + \sum_{a'' \in \A} {k}_{\hat{b}a''} \left( 2\delta_{a''a} - {k}_{a''a} \right) \right)
\left( - y_{\hat{b}} + \sum_{a'\in \A} {k}_{\hat{b}a'} \theta_{a'} (0)\right) \right] \label{vqacr_geb_qkm_haar_input_theorem_r_3_2} \\
={}& - \O \left(\frac{|\A|}{4^N} \right) y_{\hat{b}} + \O \left(\frac{|\A|}{4^N} \right) \sum_{a' \in \A} \theta_{a'}(0) \mathop{\E}_{\hat{b}} \left[ {k}_{\hat{b} a'}  \right] - y_{\hat{b}}\sum_{a'' \in \A} \mathop{\E}_{\hat{b}} \left[ {k}_{\hat{b} a''}  \right] \left( 2\delta_{a''a} - {k}_{a''a} \right) \notag \\
+{}& \sum_{a'a'' \in \A} \mathop{\E}_{\hat{b}} \left[ {k}_{\hat{b} a'} {k}_{\hat{b} a''} \right] \left( 2\delta_{a''a} - {k}_{a''a} \right) \theta_{a'}(0) \notag \\
={}& - \O \left(\frac{|\A|}{4^N} \right) y_{\hat{b}} + \O \left(\frac{|\A|}{8^N} \right) \sum_{a' \in \A} \theta_{a'}(0)  - y_{\hat{b}}\sum_{a'' \in \A} \frac{1}{2^N} \left( 2\delta_{a''a} - {k}_{a''a} \right) \notag \\
+{}& \sum_{a'a'' \in \A} \frac{1}{2^N(2^N+1)} \left[ 1+ k_{a'a''} \right] \left( 2\delta_{a''a} - {k}_{a''a} \right) \theta_{a'}(0) \label{vqacr_geb_qkm_haar_input_theorem_r_3_4} \\
\simeq{}& - y_{\hat{b}}\sum_{a'' \in \A} \frac{1}{2^N} \left( 2\delta_{a''a} - {k}_{a''a} \right) + \sum_{a'a'' \in \A} \frac{1}{4^N} \left[ 1+ k_{a'a''} \right] \left( 2\delta_{a''a} - {k}_{a''a} \right) \theta_{a'}(0) . \label{vqacr_geb_qkm_haar_input_theorem_r_3_5}
\end{align}
Eq.~(\ref{vqacr_geb_qkm_haar_input_theorem_r_3_2}) follows from Eq.~(\ref{vqacr_geb_qkm_haar_input_theorem_g_ba_fin}). Eq.~(\ref{vqacr_geb_qkm_haar_input_theorem_r_3_4}) is derived by using Lemmas~\ref{vqacr_qntk_tdesign_tr_wawb} and \ref{vqacr_qntk_tdesign_trwawbtrwcwd}.
By writing the summation as the matrix-vector production form, we have
\begin{align}
{}& \frac{1}{|\mathcal{B}|} \sum_{\hat{b} \in \mathcal{B}, a \in \mathcal{A}} \mathop{\E}_{\hat{b}} \mathop{\E}_{a} \left[ g_{\hat{b} a} r_{\hat{b}} (0) r_{a}(0) \right] \notag \\
\simeq{}& - \frac{1}{2^N|\B|} \left( \bm{e}^T \bm{y}_{\B} \right) \mathop{\E}_{\A} \left[ \bm{e}^T \left( 2I -{k} \right) \bm{r}_{\A} (0) \right] + \frac{1}{4^N |\B|} \mathop{\E}_{\A} \left[ \bmt(0)^T \left( \bm{e} \bm{e}^T + {k} \right) \left( 2I-{k} \right) \bm{r}_{\A} (0) \right] \tag{using Eq.~(\ref{vqacr_geb_qkm_haar_input_theorem_r_3_5})} \\
={}& - \frac{1}{2^N|\B|} \left( \bm{e}^T \bm{y}_{\B} \right) \mathop{\E}_{\A} \left[ \bm{e}^T \left( 2I -{k} \right) \left( {k} \bmt(0) - \bm{y}_{\A} \right) \right] + \frac{1}{4^N |\B|} \mathop{\E}_{\A} \left[ \bmt(0)^T \left( \bm{e} \bm{e}^T + {k} \right) \left( 2I-{k} \right) \left( {k} \bmt(0) - \bm{y}_{\A} \right) \right] \label{vqacr_geb_qkm_haar_input_theorem_r_3_8} \\
\simeq{}& - \frac{1}{2^N|\B|} \left( \bm{e}^T \bm{y}_{\B} \right) \left[ \bm{e}^T \bmt(0) - \left( 1-\frac{|\A|-1}{2^N}\right) \bm{e}^T \bm{y}_{\A} \right] \notag \\
+{}& \frac{1}{4^N |\B|} \left[ \left(1+\frac{1}{2^N} \right) \left( \bm{e}^T \bmt(0) \right)^2 - \left(1-\frac{|\A|-1}{2^N} \right)  \bm{e}^T \bmt(0)  \bm{e}^T \bm{y}_{\A} + \left(1-\frac{1}{2^N} \right)  \bmt(0)^T \bmt(0) - \bmt(0)^T \bm{y}_{\A}  \right] \label{vqacr_geb_qkm_haar_input_theorem_r_3_9} \\
\simeq{}& - \frac{1}{2^N|\B|} \left( \bm{e}^T \bm{y}_{\B} \right) \left[ \bm{e}^T \bmt(0) - \bm{e}^T \bm{y}_{\A} \right] . \label{vqacr_geb_qkm_haar_input_theorem_r_3_10}
\end{align}
Eq.~(\ref{vqacr_geb_qkm_haar_input_theorem_r_3_8}) is derived by using $\bm{r}=\bm{z}-\bm{y}$ and Eq.~(\ref{vqacr_geb_qntk_params_zero}). Eq.~(\ref{vqacr_geb_qkm_haar_input_theorem_r_3_9}) is derived by using $k^{-1} \approx 2I-k$ and 
\begin{align}
\mathop{\E} k_{aa'} ={}& \frac{1}{2^N} + \left(1-\frac{1}{2^N} \right) \delta_{aa'} \label{vqacr_geb_qkm_haar_input_theorem_r_3_11},
\end{align}
where Eq.~(\ref{vqacr_geb_qkm_haar_input_theorem_r_3_11}) follows from Lemma~\ref{vqacr_qntk_tdesign_tr_wawb}.
We remark that
\begin{align}
\mathop{\E}_{\A} \bm{e}^T \bm{r}_{\A} (0) ={}& \mathop{\E}_{\A} \bm{e}^T {k} \bmt(0) - \bm{e}^T \bm{y}_{\A} \simeq \bm{e}^T \bmt(0) - \bm{e}^T \bm{y}_{\A} , \notag \\
\mathop{\E}_{\B} \bm{e}^T \bm{r}_{\B} (0) ={}& 
\mathop{\E}_{\B} \sum_{\hat{b}\in \B} \sum_{a\in A} {k}_{\hat{b}a} \theta_a(0) - y_{\hat{b}} \simeq - \bm{e}^T \bm{y}_{\B} .
\end{align}
Using Eqs.~(\ref{vqacr_geb_qkm_haar_input_theorem_r_2_9}) and (\ref{vqacr_geb_qkm_haar_input_theorem_r_3_10}) in Eq.~(\ref{vqacr_geb_qkm_haar_input_theorem_r_1_1}), we have
\begin{align}
\mathop{\E} \left[ \mathcal{E}_{\mathcal{B}|\mathcal{A}} - \mathcal{L}_{\B} (0) \right] \simeq{}& \frac{1}{2^N|\B|} \left( \bm{e}^T \bm{y}_{\B} \right) \left[ \bm{e}^T \bmt(0) - \bm{e}^T \bm{y}_{\A} \right] \\
\simeq{}& - \frac{1}{2^N|\B|} \mathop{\E}_{\A} \bm{e}^T \bm{r}_{\A} (0) \mathop{\E}_{\B} \bm{e}^T \bm{r}_{\B} (0) \\
\geq{}& - \frac{\sqrt{|\A||\B|}}{2^N|\B|} \mathop{\E}_{\A} \|\bm{r}_{\A}(0)\|_2 \mathop{\E}_{\B} \|\bm{r}_{\B}(0)\|_2 \\
={}& - \frac{|\A|}{2^{N-1}} \mathop{\E} \sqrt{ \mathcal{L}_{\A}(0) \mathcal{L}_{\B}(0) } .
\end{align}

\end{proof}

\subsection{Generalization error of quantum neural network}

Next, we consider the quantum machine learning model utilizing variational quantum circuit to generate predications for input samples. Specifically, we denote by 
\begin{equation}\label{vqacr_geb_qntk_haar_input_v_theta}
V(\bmt)=\prod_{d=D}^{1} G_d(\theta_d) W_d =\prod_{d=D}^{1} \exp(-iH_d \theta_d/2) W_d
\end{equation}
the variational quantum circuit parameterized by $\bmt \in \R^D$, where $W_d$ and $H_d$ denote the fixed unitary gate and the Hamiltonian, respectively. We assume that all Hamiltonians $\{H_d\}$ are unitary Hermitians. Let the observable be $O$. For the sample $\rho_a \in \A$, the predication at training step $t$ is given by 
\begin{equation}\label{vqacr_geb_qntk_haar_input_z_a_t}
z_a(t) = \Tr \left[ O V(\bmt(t)) \rho_a V(\bmt(t))^\dag  \right] .
\end{equation}
The QNTK is 
\begin{equation}
K(t) := K(\bmt(t)) = \frac{1}{D} {J \left(\bmt(t)\right)}^T {J \left(\bmt(t)\right)},
\end{equation}
where the Jacobian matrix $J(\bmt) \in \R^{D\times |\A|}$ with elements $J_{da}(\bmt)=\frac{\partial z_a}{\partial \theta_d}(\bmt)$. For convenience, we assume the value of initial parameter $\theta_d(0)$ is absorbed into the fixed unitary $W_d$, i.e., $\theta_d(0) \equiv 0$. By using Lemma~\ref{vqacr_qntk_jac_form}, we denote
\begin{equation}\label{vqacr_geb_qntk_haar_input_J_da}
J_{da} := J_{d a} (\0) = -\frac{i}{2} \Tr \left[  \big[ {W_{D:d+1}}^\dag O W_{D:d+1}  , H_d  \big] W_{d:1} \rho_a {W_{d:1}}^\dag  \right] , 
\end{equation}
where $W_{d:d'}=W_d W_{d-1} \cdots W_{d'}$. We assume the frozen QNTK during training, i.e. 
\begin{equation}\label{vqacr_geb_qntk_haar_input_frozen_K_JTJ}
K(t) = K(0) = K = \frac{1}{D} J^T J.
\end{equation}

Similar to the quantum kernel method, property of input states could affect QNTK and the corresponding generalization behavior. Here we provide a negative example under the frozen QNTK regime for datasets from state 2-designs. The generalization error result in the main text is repeated in Theorem~\ref{vqacr_geb_qntk_haar_input_theorem} for convenience.

\begin{theorem}\label{vqacr_geb_qntk_haar_input_theorem}
Suppose all $N$-qubit quantum states in the training and test datasets $\A$ and $\B$ are independently sampled from state 2-designs. Let $\L_{\A}(t)$ and $\L_{\B}(t)$ be the training and the test loss function of the quantum neural network (Eqs.~(\ref{vqacr_main_backqml_za}) and (\ref{vqacr_main_backqml_qnn_o})) at the $t$-th step, where the label is given as $y:=\Tr[O U \rho U^\dag]$ for a target unitary $U$ with the zero trace observable $O$. Then, under the frozen QNTK regime,
\begin{equation}\label{vqacr_geb_qntk_haar_input_theorem_eq}
\mathop{\E} \L_{\B}(\infty) \geq{} \left( 1 - \frac{|\A|}{2^{2N}-1} \right) \mathop{\E} \mathcal{L}_{\B} (0) ,
\end{equation}
where the expectation is taken under $2$-designs distributions for states in $\A$ and $\B$ and the target unitary $U$.
\end{theorem}

\begin{proof}

After the training, the loss on the training set $\A$ should be zero. 
By using Eq.~(34) in Ref.~\cite{liu2021representation}, we have
\begin{align}
z_{\hat{b}} (\infty) ={}& z_{\hat{b}} (0) - \sum_{a,a' \in \mathcal{A}} \tilde{K}_{\hat{b}a'} [{K}^{-1}]_{a'a} r_{a} (0) 
\label{vqacr_geb_qntk_haar_input_z_b_infty}
\end{align}
under the frozen QNTK regime, where $\tilde{K}$ denotes the QNTK for the whole dataset $\A \cup \B$. Notice that Eq.~(\ref{vqacr_geb_qntk_haar_input_z_b_infty}) shares the same formulation with Eq.~(\ref{vqacr_geb_qntk_prediction_test}). After some calculation, we could derive the generalization error,
\begin{align}
\mathcal{R}_{\mathcal{B}|\mathcal{A}} :={}& \mathcal{E}_{\mathcal{B}|\mathcal{A}} - \mathcal{L}_{\mathcal{B}} (0) \label{vqacr_geb_qntk_ntk_re_error_eq_1} \\
={}& \frac{1}{2|\mathcal{B}|} \sum_{\hat{b} \in \mathcal{B}, a_1, a_2 \in \mathcal{A}} g_{\hat{b} a_1} g_{\hat{b} a_2} r_{a_1}(0) r_{a_2}(0) - \frac{1}{|\mathcal{B}|} \sum_{\hat{b} \in \mathcal{B}, a \in \mathcal{A}} g_{\hat{b} a} r_{\hat{b}} (0) r_{a}(0) , \label{vqacr_geb_qntk_ntk_re_error_eq_3}
\end{align}
where the metric 
\begin{equation}\label{vqacr_geb_qntk_ntk_metric_g_eq}
g_{\hat{b} a} = \sum_{a'\in \mathcal{A}} \tilde{K}_{\hat{b} a'} [{K}^{-1}]_{a'a}.
\end{equation}
We remark that $g_{\hat{b} a} \in \R$ since all entries in the QNTK are real elements. Therefore, the first term in Eq.~(\ref{vqacr_geb_qntk_ntk_re_error_eq_3}) is a positive semidefinite quadratic form, which is greater than or equal to zero. Thus, it suffices to show the following result:
\begin{align}
\mathop{\E}_{\A} \mathop{\E}_{\B} \mathop{\E}_{U} \left[ g_{\hat{b} a} r_{\hat{b}} (0) r_{a}(0) \right] \leq{}& \frac{1}{2^{2N}-1} \mathop{\E}_{\B} \mathop{\E}_{U} \mathcal{L}_{\B} (0) . \label{vqacr_geb_qntk_ntk_target_eq}
\end{align}

The rest part of the proof is to derive Eq.~(\ref{vqacr_geb_qntk_ntk_target_eq}). First, we calculate the expectation of the initial loss on the test set. The main idea is to conduct the average over the distribution of the target unitary $U$, and then deal with the sample $\hat{b}$. We have 
\begin{align}
\mathop{\E}_{\B} \mathop{\E}_{U} \L_{\B}(0) ={}& \frac{1}{2|\B|} \sum_{\hat{b} \in \B} \mathop{\E}_{\hat{b}} \left[ \mathop{\E}_{U} \left[ \left( z_{\hat{b}} (0) - y_{\hat{b}} \right)^2 \right] \right] \tag{using $r=z(0)-y$} \\
={}& \frac{1}{2|\B|} \sum_{\hat{b} \in \B} \mathop{\E}_{\hat{b}} \left[ \mathop{\E}_{U} \left[ \left( \Tr [O {W_{D:1}} \rho_{\hat{b}} {W_{D:1}}^\dag ] - \Tr [OU \rho_{\hat{b}} U^\dag ] \right)^2 \right] \right] \label{vqacr_geb_qntk_ntk_LB_exp_2} \\
={}& \frac{1}{2|\B|} \sum_{\hat{b} \in \B} \mathop{\E}_{\hat{b}} \left[ \left( \Tr [O {W_{D:1}} \rho_{\hat{b}} {W_{D:1}}^\dag ] \right)^2 + \frac{1}{2^N(2^N+1)} \Tr[O^2] \right] \label{vqacr_geb_qntk_ntk_LB_exp_3} \\
={}& \frac{1}{2^N(2^N+1)} \Tr[O^2] . \label{vqacr_geb_qntk_ntk_LB_exp_4}
\end{align}
Eq.~(\ref{vqacr_geb_qntk_ntk_LB_exp_2}) follows by using the detailed formulation of $z_{\hat{b}}(0)$ and $y_{\hat{b}}$ in terms of measurement results. The variational quantum circuit is denoted as $W_{D:1}$. Eq.~(\ref{vqacr_geb_qntk_ntk_LB_exp_3}) is derived by taking the average over the distribution of $U$ using Lemma~\ref{vqacr_qntk_tdesign_trwawbtrwcwd} and the condition $\Tr[O]=0$. Eq.~(\ref{vqacr_geb_qntk_ntk_LB_exp_4}) is derived similarly by taking the average over the distribution of $U_{\hat{b}}$ in $\rho_{\hat{b}}:=U_{\hat{b}}|0\>\<0|{U_{\hat{b}}}^\dag$ using Lemma~\ref{vqacr_qntk_tdesign_trwawbtrwcwd} and the condition $\Tr[O]=0$. 

Next, we calculate the expectation in the left part of Eq.~(\ref{vqacr_geb_qntk_ntk_target_eq}). The main idea is to first decompose $g_{\hat{b}a}$ into items with independent $\hat{b}$ or $a$. Next, we conduct the average over the distribution of the target unitary $U$, and then deal with the sample $\hat{b}$. We finish the decomposition and the expectation of the $U$ part as follows,
\begin{align}
\mathop{\E}_{\A} \mathop{\E}_{\B} \mathop{\E}_{U} \left[ g_{\hat{b} a} r_{\hat{b}} (0) r_a(0) \right] ={}& \mathop{\E}_{\A} \mathop{\E}_{\hat{b}} \left[ \sum_{a' \in \A}  [K^{-1}]_{a'a} \tilde{K}_{\hat{b}a'} \mathop{\E}_{U} \left[ r_{\hat{b}} (0) r_a(0) \right] \right]
\tag{using  Eq.~(\ref{vqacr_geb_qntk_ntk_metric_g_eq})} \\
={}& \frac{1}{D} \sum_{d=1}^{D} \mathop{\E}_{\A} \left[ \sum_{a' \in \A}  [K^{-1}]_{a'a} J_{da'} \mathop{\E}_{\hat{b}} \left[  J_{d\hat{b}} \mathop{\E}_{U} \left[ r_{\hat{b}} (0) r_a(0) \right] \right] \right]
\tag{using  Eq.~(\ref{vqacr_geb_qntk_haar_input_frozen_K_JTJ})} \\
={}& \frac{1}{D} \sum_{d=1}^{D} \mathop{\E}_{\A} \left[ \sum_{a' \in \A}  [K^{-1}]_{a'a} J_{da'} z_a(0) \mathop{\E}_{\hat{b}} \left[  J_{d\hat{b}} z_{\hat{b}} (0) \right] \right] \notag \\
+{}& \frac{1}{(2^{2N}-1)D} \Tr[O^2] \sum_{d=1}^{D} \mathop{\E}_{\A} \left[ \sum_{a' \in \A}  [K^{-1}]_{a'a} J_{da'} \mathop{\E}_{\hat{b}} \left[  J_{d\hat{b}} \Tr[\rho_a \rho_{\hat{b}}] \right] \right] \notag \\
-{}& \frac{1}{2^N(2^{2N}-1)D} \Tr[O^2] \sum_{d=1}^{D} \mathop{\E}_{\A} \left[ \sum_{a' \in \A}  [K^{-1}]_{a'a} J_{da'} \mathop{\E}_{\hat{b}} \left[  J_{d\hat{b}} \right] \right] . \label{vqacr_geb_qntk_ntk_grr_exp_U_3} 
\end{align}
Eq.~(\ref{vqacr_geb_qntk_ntk_grr_exp_U_3}) is obtained by taking the expectation of $U$ as follows,
\begin{align}
\mathop{\E}_{U} \left[ r_{\hat{b}} (0) r_a(0) \right] ={}& \mathop{\E}_{U} \left(z_{\hat{b}} (0) - y_{\hat{b}} \right) \left( z_a(0) - y_a \right) \tag{using $r(0)=z(0) - y$} \\
={}& \mathop{\E}_{U} \left(z_{\hat{b}} (0) - \Tr [OU \rho_{\hat{b}} U^\dag ] \right) \left( z_a(0) - \Tr [OU \rho_a U^\dag ] \right) \notag \\
={}& z_{\hat{b}} (0) z_a(0) + \frac{1}{2^{2N}-1} \Tr[O^2] \Tr[\rho_a \rho_{\hat{b}}] - \frac{1}{2^N(2^{2N}-1)} \Tr[O^2] , \label{vqacr_geb_qntk_ntk_grr_exp_U_6}
\end{align}
where Eq.~(\ref{vqacr_geb_qntk_ntk_grr_exp_U_6}) follows from Lemmas~\ref{vqacr_qntk_tdesign_tr_wawb} and \ref{vqacr_qntk_tdesign_trwawbtrwcwd} and the condition $\Tr[O]=0$. 

Next, we deal with the three $\E_{\hat{b}}$-related terms in Eq.~(\ref{vqacr_geb_qntk_ntk_grr_exp_U_3}) separately. The third $\E_{\hat{b}}$-related term in Eq.~(\ref{vqacr_geb_qntk_ntk_grr_exp_U_3}) is
\begin{align}
\mathop{\E}_{\hat{b}} \left[  J_{d\hat{b}} \right] ={}& \frac{-i}{2} \mathop{\E}_{\hat{b}} \left[  \Tr \left[  \big[ {W_{D:d+1}}^\dag O W_{D:d+1}  , H_d  \big] W_{d:1} \rho_{\hat{b}} {W_{d:1}}^\dag  \right] \right] \label{vqacr_geb_qntk_ntk_grr_exp_b_1_1} \\
={}& \frac{-i}{2\times 2^N} \Tr \left[ {W_{d:1}}^\dag \big[ {W_{D:d+1}}^\dag O W_{D:d+1}  , H_d  \big] W_{d:1} \right] \label{vqacr_geb_qntk_ntk_grr_exp_b_1_2} \\
={}& \frac{-i}{2\times 2^N(2^N+1)} \Tr \left[ \big[ {W_{D:d+1}}^\dag O W_{D:d+1}  , H_d  \big] \right] \label{vqacr_geb_qntk_ntk_grr_exp_b_1_3} \\
={}& 0 . \label{vqacr_geb_qntk_ntk_grr_exp_b_1_4}
\end{align}
Eq.~(\ref{vqacr_geb_qntk_ntk_grr_exp_b_1_1}) follows from the formulation of the partial derivative in Eq.~(\ref{vqacr_geb_qntk_haar_input_J_da}). Eq.~(\ref{vqacr_geb_qntk_ntk_grr_exp_b_1_2}) follows by taking the average over the distribution of $U_{\hat{b}}$ in $\rho_{\hat{b}}:=U_{\hat{b}}|0\>\<0|{U_{\hat{b}}}^\dag$ using Lemma~\ref{vqacr_qntk_tdesign_tr_wawb}. Eq.~(\ref{vqacr_geb_qntk_ntk_grr_exp_b_1_3}) is derived by using $\Tr[ABC]=\Tr[BCA]$. Eq.~(\ref{vqacr_geb_qntk_ntk_grr_exp_b_1_4}) is derived by using $\Tr[[A,B]]=\Tr[AB-BA]=0$. 

The second $\E_{\hat{b}}$-related term in Eq.~(\ref{vqacr_geb_qntk_ntk_grr_exp_U_3}) is
\begin{align}
\mathop{\E}_{\hat{b}} \left[  J_{d\hat{b}} \Tr[\rho_a \rho_{\hat{b}}] \right] ={}& \frac{-i}{2} \mathop{\E}_{\hat{b}} \left[  \Tr \left[  \big[ {W_{D:d+1}}^\dag O W_{D:d+1}  , H_d  \big] W_{d:1} \rho_{\hat{b}} {W_{d:1}}^\dag  \right] \Tr \left[ \rho_a \rho_{\hat{b}} \right] \right] \tag{using Eq.~(\ref{vqacr_geb_qntk_haar_input_J_da})} \\
={}& \frac{-i}{2\times 2^N(2^N+1)} \Tr \left[ {W_{d:1}}^\dag \big[ {W_{D:d+1}}^\dag O W_{D:d+1}  , H_d  \big] W_{d:1} \right] \Tr \left[ \rho_a \right] \notag \\
+{}& \frac{-i}{2\times 2^N(2^N+1)} \Tr \left[ {W_{d:1}}^\dag \big[ {W_{D:d+1}}^\dag O W_{D:d+1}  , H_d  \big] W_{d:1}  \rho_a \right] \label{vqacr_geb_qntk_ntk_grr_exp_b_2_2} \\
={}& \frac{-i}{2\times 2^N(2^N+1)} \Tr \left[ {W_{d:1}}^\dag \big[ {W_{D:d+1}}^\dag O W_{D:d+1}  , H_d  \big] W_{d:1}  \rho_a \right] . \label{vqacr_geb_qntk_ntk_grr_exp_b_2_3}
\end{align}
Eq.~(\ref{vqacr_geb_qntk_ntk_grr_exp_b_2_2}) follows by taking the average over the distribution of $U_{\hat{b}}$ in $\rho_{\hat{b}}:=U_{\hat{b}}|0\>\<0|{U_{\hat{b}}}^\dag$ using Lemma~\ref{vqacr_qntk_tdesign_trwawbtrwcwd}. Eq.~(\ref{vqacr_geb_qntk_ntk_grr_exp_b_2_3}) follows by noticing that the first term in Eq.~(\ref{vqacr_geb_qntk_ntk_grr_exp_b_2_2}) equals to zero by using the derivation in Eqs.~(\ref{vqacr_geb_qntk_ntk_grr_exp_b_1_2}-\ref{vqacr_geb_qntk_ntk_grr_exp_b_1_4}). 

The first $\E_{\hat{b}}$-related term in Eq.~(\ref{vqacr_geb_qntk_ntk_grr_exp_U_3}) is
\begin{align}
\mathop{\E}_{\hat{b}} \left[ J_{d\hat{b}} z_{\hat{b}} (0) \right] ={}& \frac{-i}{2} \mathop{\E}_{\hat{b}} \left[  \Tr \left[  \big[ {W_{D:d+1}}^\dag O W_{D:d+1}  , H_d  \big] W_{d:1} \rho_{\hat{b}} {W_{d:1}}^\dag  \right] \Tr \left[ O {W_{D:1}} \rho_{\hat{b}}  {W_{D:1}}^\dag \right] \right] \tag{using Eq.~(\ref{vqacr_geb_qntk_haar_input_J_da})} \\
={}& \frac{-i}{2\times 2^N(2^N+1)} \Tr \left[ {W_{d:1}}^\dag \big[ {W_{D:d+1}}^\dag O W_{D:d+1}  , H_d  \big] W_{d:1} \right] \Tr \left[ {W_{D:1}}^\dag O {W_{D:1}} \right] \notag \\
+{}& \frac{-i}{2\times 2^N(2^N+1)} \Tr \left[ {W_{d:1}}^\dag \big[ {W_{D:d+1}}^\dag O W_{D:d+1}  , H_d  \big] W_{d:1}  {W_{D:1}}^\dag O {W_{D:1}} \right] \label{vqacr_geb_qntk_ntk_grr_exp_b_3_2} \\
={}& \frac{-i}{2\times 2^N(2^N+1)} \Tr \left[ \big[ {W_{D:d+1}}^\dag O W_{D:d+1}  , H_d  \big] {W_{D:d+1}}^\dag O {W_{D:d+1}} \right] \label{vqacr_geb_qntk_ntk_grr_exp_b_3_3} \\
={}& 0. \label{vqacr_geb_qntk_ntk_grr_exp_b_3_4}
\end{align}
Eq.~(\ref{vqacr_geb_qntk_ntk_grr_exp_b_3_2}) follows by taking the average over the distribution of $U_{\hat{b}}$ in $\rho_{\hat{b}}:=U_{\hat{b}}|0\>\<0|{U_{\hat{b}}}^\dag$ using Lemma~\ref{vqacr_qntk_tdesign_trwawbtrwcwd}. Eq.~(\ref{vqacr_geb_qntk_ntk_grr_exp_b_3_3}) follows by noticing that the first term in Eq.~(\ref{vqacr_geb_qntk_ntk_grr_exp_b_3_2}) equals to zero by using the derivation in Eqs.~(\ref{vqacr_geb_qntk_ntk_grr_exp_b_1_2}-\ref{vqacr_geb_qntk_ntk_grr_exp_b_1_4}). Eq.~(\ref{vqacr_geb_qntk_ntk_grr_exp_b_3_4}) follows by noticing that $\Tr[[A,B]A]=\Tr[ABA-BAA]=0$.

By using Eq.~(\ref{vqacr_geb_qntk_ntk_grr_exp_b_1_4}), (\ref{vqacr_geb_qntk_ntk_grr_exp_b_2_3}), and (\ref{vqacr_geb_qntk_ntk_grr_exp_b_3_4}) jointly in Eq.~(\ref{vqacr_geb_qntk_ntk_grr_exp_U_3}), we obtain
\begin{align}
{}& \mathop{\E}_{\A} \mathop{\E}_{\B} \mathop{\E}_{U} \left[ g_{\hat{b} a} r_{\hat{b}} (0) r_a(0) \right] \notag \\
={}& \frac{-i}{2^{N+1}(2^N+1)(2^{2N}-1)D} \Tr[O^2] \sum_{d=1}^{D} \mathop{\E}_{\A} \left[ \sum_{a' \in \A}  [K^{-1}]_{a'a} J_{da'} \Tr \left[ {W_{d:1}}^\dag \big[ {W_{D:d+1}}^\dag O W_{D:d+1}  , H_d  \big] W_{d:1}  \rho_a \right] \right] \notag \\
={}& \frac{1}{2^{N}(2^N+1)(2^{2N}-1)D} \Tr[O^2] \sum_{d=1}^{D} \mathop{\E}_{\A} \left[ \sum_{a' \in \A}  [K^{-1}]_{a'a} J_{da'} J_{da} \right] \tag{using Eq.~(\ref{vqacr_geb_qntk_haar_input_J_da})} \\
={}& \frac{1}{2^{N}(2^N+1)(2^{2N}-1)} \Tr[O^2] \mathop{\E}_{\A} \left[ \sum_{a' \in \A}  [K^{-1}]_{a'a} K_{aa'} \right] \tag{using Eq.~(\ref{vqacr_geb_qntk_haar_input_frozen_K_JTJ})} \\
={}& \frac{1}{2^{N}(2^N+1)(2^{2N}-1)} \Tr[O^2] . \label{vqacr_geb_qntk_ntk_grr_exp_O2}
\end{align}
Comparing Eqs.~(\ref{vqacr_geb_qntk_ntk_LB_exp_4}) and (\ref{vqacr_geb_qntk_ntk_grr_exp_O2}), we obtain Eq.~(\ref{vqacr_geb_qntk_ntk_target_eq}). Thus, we have proved Theorem~\ref{vqacr_geb_qntk_haar_input_theorem}.

\end{proof}

\section{Spectral analysis of the quantum neural tangent kernel}
\label{vqacr_qntk_app_lmin_qntk}

We begin with a warm-up example that shows the exponentially small spectrum of QNTK with $2$-designs states as datasets.

\begin{lemma}\label{vqacr_lemma_qntk_spectrum_qntk_haar}
Let $\A\{(\rho_a, y_a)\}$ be the training datasets with independent $N$-qubit $2$-design states $\rho_a$. Denote by $z_a=\Tr[OV(\bmt)\rho_a V(\bmt)^\dag]$ the QNN with the observable $O$ and the variational circuit $V(\bmt)=\prod_{d=D}^{1} \exp[-iH_d \theta_d/2] W_d$, where the hamiltonian $H_d$ is Hermitian unitary and $W_d$ is the fixed unitary. Then the QNTK has exponentially small spectrum, i.e. there exists a constant $C>0$ such that
\begin{equation} \label{vqacr_lemma_qntk_spectrum_qntk_haar_eq}
\lmax \left[ K (\bmt) \right] \leq{} \frac{\|O\|_F^2}{4^{N}} \left( 1 + \frac{C|\A|}{2^N \delta} \sqrt{\ln |\A|} \right)
\end{equation}
with probability at least $1-\delta$.
\end{lemma}

\begin{proof}

We follow notations in Lemma~\ref{vqacr_qntk_jac_form}. For convenience, we assume the value of each parameter $\theta_d$ is absorbed into the fixed unitary $W_d$, i.e., $\theta_d \equiv 0$. Denote by
\begin{align}
O_d :={}& - \frac{i}{2} {W_{d:1}}^\dag \big[ {W_{D:d+1}}^\dag O W_{D:d+1}  , H_d  \big] W_{d:1} . \label{vqacr_lemma_qntk_spectrum_qntk_haar_Od}
\end{align}
We can write the partial derivative of $z_a$ in the formulation with Pauli decomposition coefficients:
\begin{align}
\frac{\partial z_a}{\partial \theta_d} ={}& \Tr \left[  O_d \rho_a \right] \tag{using Lemma~\ref{vqacr_qntk_jac_form}} \\
={}& \frac{1}{2^N} \left[ B^{\otimes} \left( \bm{c}_{O_d} \otimes \bm{c}_{\rho_a} \right) \right]_{\bm{0}} \tag{using notation in Lemma~\ref{vqacr_lemma_pauli_decom}} \\
={}& \frac{1}{2^N} \bm{c}_{O_d}^T \bm{c}_{\rho_a}. \tag{calculating $B^\otimes$ in Lemma~\ref{vqacr_lemma_pauli_decom}} 
\end{align}
Therefore, the Jacobian matrix is
\begin{align}
J = \frac{1}{2^N} C_{O}^T A , \label{vqacr_lemma_qntk_spectrum_qntk_haar_J}	
\end{align}
where $C_{O} \in \R^{(4^N-1)\times D}$ stores the Pauli decomposition of matrices $\{O_d\}_{d=1}^{D}$ for $d \in \{1,2,\cdots,D\}$. We remark that the each matrix $O_d$ only corresponds to $4^N-1$ coefficients since the coefficient $c_{O_d, \bm{0}}=\frac{1}{2^N} \Tr [O_d I]=0$. Similarly the matrix $A \in \R^{(4^N-1)\times |\A|}$ stores the Pauli decomposition of states $\{\rho_a\}_{a\in\A}$, which is the same with the notation in Eq.~(\ref{vqacr_geb_qkm_haar_input_theorem_A_ia}).
By using Eq.~(\ref{vqacr_qntk_eq}), the QNTK is
\begin{align}
K(\bmt) ={}& \frac{1}{D} J^T J \notag \\
={}& \frac{1}{4^N D} A^T C_{O} C_{O}^T A . \notag
\end{align}
Thus, the largest eigenvalue is upper bounded as
\begin{align}
\lmax \left[ K(\bmt) \right]	 \leq{}& \frac{1}{4^N D} \lmax \left[ A^T A \right] \lmax \left[ C_{O} C_{O}^T \right] . \label{vqacr_lemma_qntk_spectrum_qntk_haar_lmax_KCA}
\end{align}
We remark that by using Eqs.~(\ref{vqacr_geb_qkm_haar_input_theorem_B}-\ref{vqacr_geb_qkm_haar_input_theorem_B_bound}), there exists a constant $C>0$ such that
\begin{align}
\lmax \left[ A^T A \right] \leq{}& \frac{1}{2^N} + \frac{C|\A|}{4^N \delta} \sqrt{\ln |\A|} \label{vqacr_lemma_qntk_spectrum_qntk_haar_lmax_A}
\end{align}
with probability at least $1-\delta$.
Besides, we have
\begin{align}
{}& \lmax \left[ C_{O} C_{O}^T \right] \notag \\
\leq{}& \left\| C_O \right\|_F^2 \notag \\
={}& \sum_{d=1}^{D} \sum_{\bm{i} \in \{0,1,2,3\}^N/\bm{0}} \left( c_{O_d, \bm{i}} \right)^2 \notag \\
={}& \sum_{d=1}^{D} 2^N \Tr \left[ O_d^2 \right] \tag{using Lemma~\ref{vqacr_lemma_pauli_decom}} \\
={}& - 2^{N-2} \sum_{d=1}^{D} \Tr \left[ \left( {W_{d:1}}^\dag \big[ {W_{D:d+1}}^\dag O W_{D:d+1}  , H_d  \big] W_{d:1} \right)^2 \right] \tag{using Eq.~(\ref{vqacr_lemma_qntk_spectrum_qntk_haar_Od})} \\
={}& - 2^{N-2} \sum_{d=1}^{D} \Tr \left[ \big[ {W_{D:d+1}}^\dag O W_{D:d+1}  , H_d  \big]^2 \right] \notag \\
={}& 2^{N-1} \sum_{d=1}^{D} \Tr \left[ O \left(  O -  W_{D:d+1} H_d {W_{D:d+1}}^\dag O W_{D:d+1}  H_d {W_{D:d+1}}^\dag \right) \right] \notag \\
={}& 2^{N-1} \sum_{d=1}^{D} \left\| O \right\|_F \left\| O -  W_{D:d+1} H_d {W_{D:d+1}}^\dag O W_{D:d+1}  H_d {W_{D:d+1}}^\dag \right\|_F  \notag \\
\leq{}& 2^{N-1} \sum_{d=1}^{D} \left\| O \right\|_F \left( \left\| O \right\|_F + \left\|  W_{D:d+1} H_d {W_{D:d+1}}^\dag O W_{D:d+1}  H_d {W_{D:d+1}}^\dag \right\|_F \right) \notag \\
={}& 2^{N-1} \sum_{d=1}^{D} \left\| O \right\|_F \left( \left\| O \right\|_F + \sqrt{ \Tr \left[ \left( W_{D:d+1} H_d {W_{D:d+1}}^\dag O W_{D:d+1}  H_d {W_{D:d+1}}^\dag \right)^2 \right] } \right) \tag{$\Tr[X^\dag X] = \|X\|_F^2$} \\
={}& 2^{N-1} \sum_{d=1}^{D} \left\| O \right\|_F \left( \left\| O \right\|_F + \sqrt{ \Tr \left[ O^2 \right] } \right) \tag{$\Tr[XY]=\Tr[YX]$} \\
={}& 2^N D \|O\|_F^2 . \label{vqacr_lemma_qntk_spectrum_qntk_haar_lmax_C}
\end{align}

Collaborating Eqs.~(\ref{vqacr_lemma_qntk_spectrum_qntk_haar_lmax_KCA}), (\ref{vqacr_lemma_qntk_spectrum_qntk_haar_lmax_A}), and (\ref{vqacr_lemma_qntk_spectrum_qntk_haar_lmax_C}), we obtain Eq.~(\ref{vqacr_lemma_qntk_spectrum_qntk_haar_eq}).

\end{proof}

The remaining part of this section focuses on the spectrum of the asymptotic QNTK with non-uniformly distributed quantum states, i.e. the QNTK with infinite number of gates and parameters, which is denoted by 
\begin{equation} \label{vqacr_qntk_infinite_eq}
K_{\infty} (\bm{\theta}) = \lim_{D \rightarrow + \infty} K(\bm{\theta}) = \lim_{D \rightarrow + \infty} \frac{1}{D} J(\bm{\theta})^{T} J(\bm{\theta}) .
\end{equation} 
Specifically, we consider the $N$-qubit ansätz $V(\bm{\theta})=\prod_{d=D}^{1} G_{d} (\theta_{d})$, where each gate $G_{d} (\theta_{d})=\exp[-iH_{d}\theta_{d}/2]$ is implemented on qubits $\mathcal{S}_{d} \subseteq \{1,2,\cdots,N\}$ with $|\mathcal{S}_{d}| = S$. Each Hamiltonian is sampled from the set $\H$. We assume general assignments of the qubits set $\Q:=\{\S_d\}$ if there has no additional instruction. One example of $\Q$ is shown in Fig.~\ref{vqacr_Q_Sd_he_circuit}, where each $\S_{d}$ is deployed in the hardware-efficient manner~\cite{PMID32874524} with $N$ different choices.

\subsection{Bounds of $K_{\infty}(\bm{0})$}

Let $A \in \R^{4^N \times |\A|}$ be the matrix that stores the Pauli decomposition coefficients of input states $\{\rho_a\}_{a\in\A}$, i.e. $A_{\bm{p}a} = \Tr[\sigma_{\bm{p}}\rho_a ]$. As shown in Lemmas~\ref{vqacr_lemma_qntk_lower_bound} and \ref{vqacr_lemma_qntk_lower_bound_xy}, $K_{\infty}(\bm{0})$ could be bounded by the linear sum of the gram of the sub-matrix of $A$, where we consider $\mathcal{H} =  \{X,Y,Z\}^{\otimes S}$ and $\mathcal{H} = \{X,Y\}^{\otimes S}$, respectively.

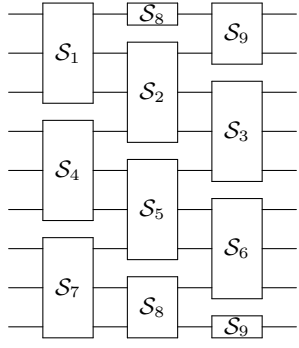
\begin{figure}
\centerline{
\Qcircuit @C=1.4em @R=.7em {
\lstick{} & \multigate{2}{\mathcal{S}_1} & \multigate{0}{\mathcal{S}_8} & \multigate{1}{\mathcal{S}_9} & \qw \\
\lstick{} & \ghost{\mathcal{S}_1} & \multigate{2}{\mathcal{S}_2} & \ghost{\mathcal{S}_9} & \qw \\
\lstick{} & \ghost{\mathcal{S}_1} & \ghost{\mathcal{S}_2} & \multigate{2}{\mathcal{S}_3} & \qw \\
\lstick{} & \multigate{2}{\mathcal{S}_4} & \ghost{\mathcal{S}_2} & \ghost{\mathcal{S}_3} & \qw \\
\lstick{} & \ghost{\mathcal{S}_4} & \multigate{2}{\mathcal{S}_5} & \ghost{\mathcal{S}_3} & \qw \\
\lstick{} & \ghost{\mathcal{S}_4} & \ghost{\mathcal{S}_5} & \multigate{2}{\mathcal{S}_6} & \qw \\
\lstick{} & \multigate{2}{\mathcal{S}_7} & \ghost{\mathcal{S}_5} & \ghost{\mathcal{S}_6} & \qw \\
\lstick{} & \ghost{\mathcal{S}_7} & \multigate{1}{\mathcal{S}_8} & \ghost{\mathcal{S}_6} & \qw \\
\lstick{} & \ghost{\mathcal{S}_7} & \ghost{\mathcal{S}_8} & \multigate{0}{\mathcal{S}_9} & \qw
}
}
\caption{An assignment of $\mathcal{Q}=\{\mathcal{S}_1,\cdots,\mathcal{S}_{N}\}$ in the hardware-efficient manner for $(N,S)=(9,3)$.}
\label{vqacr_Q_Sd_he_circuit}
\end{figure}

\begin{lemma}\label{vqacr_lemma_qntk_lower_bound}
Let the observable be $O=\sum_{k=1}^{N} o_k Z_k$, where each $Z_k$ is the product of Pauli matrices that is $Z$ on the $k$-th qubit and $I$ on other qubits. Denote by $H_d$ the Hamiltonian of the $d$-th quantum gate in the circuit $V(\bmt)=\prod_{d=D}^{1} \exp[-iH_d \theta_d/2]$ that acts on qubits $\S_d \subseteq \{1,2,\cdots,N\}$. We assume that $H_d$ is randomly sampled from the set $\H:=\{H , H \in \{X,Y,Z\}^{\otimes S}\}$. Let $\Q:=\{\S_d\}$. Let $A \in \R^{4^N \times |\A|}$ be the matrix with elements $A_{\bm{p}a} = \Tr[\sigma_{\bm{p}}\rho_a ]$. Then the asymptotic QNTK at $\bm{\theta}=\bm{0}$ could be bounded as 
\begin{equation} \label{vqacr_lemma_qntk_lower_bound_eq}
\frac{1}{3^{S} |\mathcal{Q}|} \sum_{w=0}^{S-1} \|O\|_2^2 A_{\mathcal{Q},w}^T A_{\mathcal{Q},w} \succcurlyeq{} K_{\infty} (\bm{0}) \succcurlyeq{}  \frac{1}{3^{S} |\mathcal{Q}|} \sum_{w=0}^{S-1} \Delta_{S-w} A_{\mathcal{Q},w}^T A_{\mathcal{Q},w} ,
\end{equation}
where we denote by $\succcurlyeq$ the semi-definite partial order and $\Delta_{j} := \min_{\bm{g}\in \{0,\pm 1\}^{\otimes N}, 1 \leq \|\bm{g}\|_1 \leq j} (\bm{g}^T \bm{o} )^2 $. The matrix $A_{\mathcal{Q},w}$ is formed by extracting rows from matrix $A$ corresponding to indices $\bm{p}$, if there exists some $\S_d \in \Q$ that matches the non-zero elements in $\bm{p}$ with $w$ numbers of element $3$.
\end{lemma}

\begin{proof}

Denote by $\sigma_{\bm{i}}$ the Hamiltonian $H_d$. Denote by $K_{d} \in \R^{|\A| \times |\A|}$ the QNTK corresponding to the gate $G_{d}$ at $\bm{\theta}=\bm{0}$. By using Lemma~\ref{vqacr_qntk_jac_form}, we have 
\begin{equation}
[K_{d}]_{aa'} = \frac{\partial z_a}{\partial \theta_{d}} (\bm{0}) \frac{\partial z_{a'}}{\partial \theta_{d}} (\bm{0}) = - \frac{1}{4} \Tr \Big[ [O, \sigma_{\bm{i}}] \rho_a \Big] \Tr \Big[ [O, \sigma_{\bm{i}}] \rho_{a'} \Big] . \label{vqacr_lemma_qntk_lower_bound_Ka_eq}
\end{equation}
We remark that $\sigma_{\bm{i}}$ and each $Z_i$ in the linear sum of the observable are products of Pauli matrices, so $[Z_i, \sigma_{\bm{i}}]$ could be $0$ or $2Z_i \sigma_{\bm{i}}$. 
For convenience, we proceed to analyze the term $\Tr \left[ [O, \sigma_{\bm{i}}] \rho_a \right]$ first:
\begin{align}
\Tr \Big[ [O, \sigma_{\bm{i}}] \rho_a \Big] ={}& \sum_{k=1}^{N} o_k \Tr \Big[ [Z_k, \sigma_{\bm{i}}] \rho_a \Big] \tag{using $O=\sum_{k=1}^{N} o_k Z_k$} \\
={}& \sum_{k=1}^{N} o_k \Tr \left[ \frac{1}{2^N} \sum_{\bm{p} \in \{0,1,2,3\}^{\otimes N}} c_{[Z_k, \sigma_{\bm{i}}], \bm{p}} \sigma_{\bm{p}} \frac{1}{2^N} \sum_{\bm{p}' \in \{0,1,2,3\}^{\otimes N}} c_{\rho_a, \bm{p}'} \sigma_{\bm{p}'}  \right]  \label{vqacr_lemma_qntk_lower_bound_1_2} \\
={}& \sum_{k=1}^{N} o_k \frac{1}{2^N} \sum_{\bm{p} \in \{0,1,2,3\}^{\otimes N}} c_{[Z_k, \sigma_{\bm{i}}], \bm{p}} A_{\bm{p} a} \label{vqacr_lemma_qntk_lower_bound_1_21} \\
={}& \sum_{k=1}^{N} o_k \frac{1}{4^N} \sum_{\bm{p} \in \{0,1,2,3\}^{\otimes N}} \left[ B^{\otimes} \left( \bm{c}_{Z_k} \otimes \bm{c}_{\sigma_{\bm{i}}} \right) - B^{\otimes} \left( \bm{c}_{\sigma_{\bm{i}}} \otimes \bm{c}_{Z_k} \right) \right]_{\bm{p}}   A_{\bm{p} a} \label{vqacr_lemma_qntk_lower_bound_1_22} \\
={}& \sum_{k=1}^{N} o_k \sum_{\bm{p} \in \{0,1,2,3\}^{\otimes N}} \left[ 2i B_k \bm{e}_{\bm{i}'} \right]_{\bm{p}}   A_{\bm{p} a} \label{vqacr_lemma_qntk_lower_bound_1_23} \\
={}& \sum_{k=1}^{N} 2 o_k i {[A^T B_k]}_{a\bm{i}'}. \label{vqacr_lemma_qntk_lower_bound_1_3}
\end{align} 
Eq.~(\ref{vqacr_lemma_qntk_lower_bound_1_2}) is derived by using the Pauli decomposition with coefficients: $X=\frac{1}{2^N} \sum_{\bm{p}} c_{X,\bm{p}} \sigma_{\bm{p}}$. Eq.~(\ref{vqacr_lemma_qntk_lower_bound_1_21}) is derived by tracing out the index $\bm{p}'$ and using the notation $A_{\bm{p}a}$. Eq.~(\ref{vqacr_lemma_qntk_lower_bound_1_22}) follows from Lemma~\ref{vqacr_lemma_pauli_decom}, where $B^{\otimes}$ follows the notation in Lemma~\ref{vqacr_lemma_pauli_decom}. 
Eq.~(\ref{vqacr_lemma_qntk_lower_bound_1_23}) is derived by using the formulation of $B^{\otimes}$, where the notation $\bm{i}'$ is defined by zero-padding the vector $\bm{i} \in \{1,2,3\}^{\otimes S}$ in the space $\{0,1,2,3\}^{\otimes N}$. We denote $\bm{e}_{\bm{i}'}:=\otimes_{q=1}^{N} e_{i_q'}$, which is a $4^N$-dimensional vector, and
\begin{align} \label{vqacr_lemma_qntk_lower_bound_B_k}
B_k = I_4^{\otimes (k-1)} \otimes
\begin{bmatrix}
& & & 0 \\
& & -1 & \\
& 1 & & \\
0 & & & 
\end{bmatrix}
\otimes I_4^{\otimes (N-k)} .
\end{align}
Thus, the average of $K_{d}$ w.r.t. the Hamiltonian set $\H=\{X,Y,Z\}^{\otimes S}$ is 
\begin{align}
\mathop{\mathbb{E}}_{H_{d} \in \H} [K_{d}]_{aa'} ={}& \frac{1}{3^{S}} \sum_{H_{d} \in \{X,Y,Z\}^{\otimes S}} -\frac{1}{4}  \Tr \Big[ [O, H_{d}] \rho_a \Big] \Tr \Big[ [O, H_{d}] \rho_{a'} \Big]  \label{vqacr_lemma_qntk_lower_bound_2_1} \\
={}& \frac{1}{3^{S}} \sum_{\bm{i} \in \{1,2,3\}^{\otimes S}} \sum_{k,k'=1}^{N} o_{k} o_{k'} [A^T B_{k}]_{a\bm{i}'} [A^T B_{k'}]_{a'\bm{i}'} .\label{vqacr_lemma_qntk_lower_bound_2_2}
\end{align}
Eqs.~(\ref{vqacr_lemma_qntk_lower_bound_2_1}) and (\ref{vqacr_lemma_qntk_lower_bound_2_2}) are derived by using Eqs.~(\ref{vqacr_lemma_qntk_lower_bound_Ka_eq}) and (\ref{vqacr_lemma_qntk_lower_bound_1_3}), respectively.

Next, we proceed with the matrix formulation $K_{d}$. We have 
\begin{align}
\mathop{\mathbb{E}}_{H_{d} \in \H} K_{d} ={}&  \frac{1}{3^{S}} \sum_{w=0}^{S-1} \sum_{ \substack{ \bm{z}=(z_1,\cdots,z_{w}) \\ \{z_1,\cdots, z_{w}\} \subset \mathcal{S}_{d} } } \sum_{\bm{i}/\bm{i}_{\bm{z}} \in \{1,2\}^{\otimes (S-w)}} \sum_{k,k'=1}^{N} o_{k} o_{k'} \left\{ [A^T B_{k}]_{a\bm{i}'} [A^T B_{k'}]_{a'\bm{i}'} \right\}_{aa'} \label{vqacr_lemma_qntk_lower_bound_3_1} \\
={}& \frac{1}{3^{S}} \sum_{k,k' \in \mathcal{S}_{d}} o_{k} o_{k'} A_{d}^T C_{k} C_{k'}^{T} A_{d} + \frac{1}{3^{S}} \sum_{w=1}^{S-1} \sum_{ \substack{ \bm{z}=(z_1,\cdots,z_{w}) \\ \{z_1,\cdots, z_{w}\} \subset \mathcal{S}_{d} } } \sum_{k,k' \in \mathcal{S}_{d}/ \{z_1,\cdots, z_{w}\} } o_{k} o_{k'} A_{d,\bm{z}}^T C_{k,\bm{z}} C_{k',\bm{z}}^{T} A_{d,\bm{z}}. \label{vqacr_lemma_qntk_lower_bound_3_2}
\end{align}
We obtain Eq.~(\ref{vqacr_lemma_qntk_lower_bound_3_1}) by rearranging the summation with kernels corresponding to Hamiltonians with $w \in \{0,1,\cdots,S\}$ number of $Z$ components, where the case $w=S$ induces a zero term by noticing $[Z^{\otimes S},O]=0$. Eq.~(\ref{vqacr_lemma_qntk_lower_bound_3_2}) follows from using notations $A_{d}$, $A_{d,\bm{z}}$, $C_k$, and $C_{k,\bm{z}}$. The matrix $A_{d} \in \R^{2^{S} \times |\A|}$ is derived by extracting rows from matrix $A$ corresponding to indices from the set $\{\bm{i}| i_x \in \{1,2\} \ \forall x \in \mathcal{S}_{d}, i_x = 0 \ \forall x \notin \mathcal{S}_{d}\}$. The matrix $A_{d,\bm{z}} \in \R^{2^{S-w} \times |\A|}$ is derived by extracting rows from matrix $A$ corresponding to indices from the set 
\begin{equation} \label{vqacr_lemma_qntk_lower_bound_index_z_set}
\Big\{ \bm{i}=(i_1, \cdots, i_N) | i_x \in \{1,2\} \ \forall x \in \mathcal{S}_{d}/\{z_1,\cdots,z_{w}\}, i_x=3 \ \forall x \in \{z_1,\cdots,z_{w}\}, i_x = 0 \ \forall x \notin \mathcal{S}_{d} \Big\}.
\end{equation}
The matrix $C_k \in \R^{2^{S} \times 2^{S}}$ is derived by extracting columns and rows from $B_k$ corresponding to indices from the set mentioned above, i.e., 
\begin{align} \label{vqacr_lemma_qntk_lower_bound_C_k}
C_k = I_2^{\otimes (k_{d}-1)} \otimes
\begin{bmatrix}
& -1 \\
1 &
\end{bmatrix}
\otimes I_2^{\otimes (S-k_{d})},
\end{align}
where $k_{d}$ is the index of $k$ in the set $\mathcal{S}_{d}$. The matrix $C_{k,\bm{z}} \in \R^{2^{S-w} \times 2^{S-w}}$ is obtained from matrix $C_k$ by dropping tensor product components corresponding to indices in $\bm{z}$. 

We remark that both $C_k$ and $C_{k,\bm{z}}$ are of full rank with singular values $\pm 1$. Therefore the eigenvalue of the positive semi-definite matrix $\sum_{k,k' \in \mathcal{S}_{d}} o_{k} o_{k'} C_{k} C_{k'}^{T}$ can be bounded as follows
\begin{align} \label{vqacr_lemma_qntk_lower_bound_C_k_least_eiganvalue}
\|O\|_2^2 = \left( \sum_{k=1}^{N} |o_k| \right)^2 \geq{} \sum_{k,k' \in \S_d} |o_k| |o_{k'}| \geq{} \lambda \bigg[ \sum_{k,k' \in \mathcal{S}_{d}} o_{k} o_{k'} C_{k} C_{k'}^{T} \bigg] \geq  \min_{\bm{g}\in \{-1,1\}^{\otimes S}} (\bm{g}^T \bm{o}_{d})^2 \geq \Delta_{S},
\end{align}
where we denote by $\bm{o}_{d} \in \R^{S}$ the vector formed by extracting elements in $\bm{o}$ corresponding to indices in $\mathcal{S}_{d}$. 
Similarly, we have 
\begin{align} \label{vqacr_lemma_qntk_lower_bound_C_kz_least_eiganvalue}
\|O\|_2^2 \geq{} \lambda \bigg( \sum_{k,k' \in \mathcal{S}_{d}/\{z_1,\cdots,z_{w}\}} o_{k} o_{k'} C_{k,\bm{z}} C_{k',\bm{z}}^{T} \bigg) \geq  \min_{\bm{g}\in \{-1,1\}^{\otimes (S-w)}} (\bm{g}^T \bm{o}_{d,\bm{z}})^2 \geq \Delta_{S-w},
\end{align}
where we denote by $\bm{o}_{d,\bm{z}} \in \R^{S-w}$ the vector formed by extracting elements in $\bm{o}$ corresponding to indices in $\mathcal{S}_{d}/\{z_1,\cdots,z_{w}\}$.

Next, we proceed to take the average w.r.t. the qubit group $\mathcal{Q}$. Denote by $\succcurlyeq$ the semi-definite partial order, i.e., $X \succcurlyeq Y \iff \bm{a}^T (X-Y) \bm{a} \geq 0$, we have
\begin{align}
K_{\infty}(\bm{0}) = \mathop{\mathbb{E}}_{\mathcal{S}_{d} \in \mathcal{Q} } \mathop{\mathbb{E}}_{H_{d} \in \mathcal{H}} K_{d} \succcurlyeq{}& \frac{1}{3^{S} |\mathcal{Q}|} \sum_{\mathcal{S}_{d} \in \mathcal{Q}} \sum_{w=0}^{S-1} \sum_{ \substack{ \bm{z}=(z_1,\cdots,z_{w}) \\ \{z_1,\cdots, z_{w}\} \subset \mathcal{S}_{d} } } \Delta_{S-w} A_{d,\bm{z}}^T A_{d,\bm{z}},\label{vqacr_lemma_qntk_lower_bound_4_1} \\
={}& \frac{1}{3^{S} |\mathcal{Q}|} \sum_{w=0}^{S-1} \Delta_{S-w} A_{\mathcal{Q},w}^T A_{\mathcal{Q},w}.\label{vqacr_lemma_qntk_lower_bound_4_2}
\end{align}
Eq.~(\ref{vqacr_lemma_qntk_lower_bound_4_1}) is obtained from Eqs.~(\ref{vqacr_lemma_qntk_lower_bound_3_2}) and (\ref{vqacr_lemma_qntk_lower_bound_C_kz_least_eiganvalue}) by noticing that $XYX^T$ is positive semi-definite if $Y$ is positive semi-definite. Eq.~(\ref{vqacr_lemma_qntk_lower_bound_4_2}) is obtained by using the notation $A_{\mathcal{Q},w}$, which is the matrix with dimension $(2^{S-w} \binom{S}{w}|\mathcal{Q}|, |\A|)$ formed by vertically concatenating matrices $A_{d,\bm{z}}$ for all $\mathcal{S}_{d} \in \mathcal{Q}$ and $ {\dim}(\bm{z}) = w$. 
Similarly we have
\begin{align}
K_{\infty}(\bm{0}) = \mathop{\mathbb{E}}_{\mathcal{S}_{d} \in \mathcal{Q} } \mathop{\mathbb{E}}_{H_{d} \in \mathcal{H}} K_{d} \preccurlyeq{}& \frac{1}{3^{S} |\mathcal{Q}|} \sum_{\mathcal{S}_{d} \in \mathcal{Q}} \sum_{w=0}^{S-1} \sum_{ \substack{ \bm{z}=(z_1,\cdots,z_{w}) \\ \{z_1,\cdots, z_{w}\} \subset \mathcal{S}_{d} } } \|O\|_2^2 A_{d,\bm{z}}^T A_{d,\bm{z}}, \notag \\
={}& \frac{1}{3^{S} |\mathcal{Q}|} \sum_{w=0}^{S-1} \|O\|_2^2 A_{\mathcal{Q},w}^T A_{\mathcal{Q},w}.\notag 
\end{align}
Thus, we obtain Eq.~(\ref{vqacr_lemma_qntk_lower_bound_eq}).

\end{proof}

\begin{lemma}\label{vqacr_lemma_qntk_lower_bound_xy}
Let the observable be $O=\sum_{k=1}^{N} o_k Z_k$, where each $Z_k$ is the product of Pauli matrices that is $Z$ on the $k$-th qubit and $I$ on other qubits. Denote by $H_d$ the Hamiltonian of the $d$-th quantum gate in the circuit $V(\bmt)=\prod_{d=D}^{1} \exp[-iH_d \theta_d/2]$ that acts on qubits $\S_d \subseteq \{1,2,\cdots,N\}$. We assume that $H_d$ is randomly sampled from the set $\H:=\{X,Y\}^{\otimes S}$. Let $\Q:=\{\S_d\}$. Let $A \in \R^{4^N \times |\A|}$ be the matrix with elements $A_{\bm{p}a} = \Tr[\sigma_{\bm{p}}\rho_a ]$. Then the asymptotic QNTK at $\bm{\theta}=\bm{0}$ could be bounded as 
\begin{equation} \label{vqacr_lemma_qntk_lower_bound_xy_eq}
\frac{1}{2^{S} |\mathcal{Q}|} \|O\|_2^2 A_{\mathcal{Q}}^T A_{\mathcal{Q}} \succcurlyeq{}K_{\infty} (\bm{0}) \succcurlyeq{} \frac{1}{2^{S} |\mathcal{Q}|} \Delta_{S} A_{\mathcal{Q}}^T A_{\mathcal{Q}},
\end{equation}
where we denote by $\succcurlyeq$ the semi-definite partial order and $\Delta_{j} := \min_{\bm{g}\in \{0,\pm 1\}^{\otimes N}, 1 \leq \|\bm{g}\|_1 \leq j} (\bm{g}^T \bm{o} )^2 $. The matrix $A_{\mathcal{Q}}$ is formed by extracting rows from matrix $A$ corresponding to indices $\bm{p}$, if there exists some $\S_d \in \Q$ that matches the non-zero elements in $\bm{p}$ with elements in $\{1,2\}$.
\end{lemma}

\begin{proof}

We follow notations and derivations in Eqs.~(\ref{vqacr_lemma_qntk_lower_bound_Ka_eq}-\ref{vqacr_lemma_qntk_lower_bound_B_k})  in Lemma~\ref{vqacr_lemma_qntk_lower_bound}.
Thus, the average of $K_{d}$ w.r.t. the Hamiltonian set $\mathcal{H}=\{X,Y\}^{\otimes S}$ is 
\begin{align}
\mathop{\mathbb{E}}_{H_{d} \in \mathcal{H}} [K_{d}]_{aa'} ={}& \frac{1}{2^{S}} \sum_{H_{d} \in \{X,Y\}^{\otimes S}} -\frac{1}{4} \Tr \Big[ [O, H_{d}] \rho_a \Big] \Tr \Big[ [O, H_{d}] \rho_{a'} \Big] \label{vqacr_lemma_qntk_lower_bound_xy_2_1} \\
={}& \frac{1}{2^{S}} \sum_{\bm{i} \in \{1,2\}^{\otimes S}} \sum_{k,k'=1}^{N} o_{k} o_{k'} [A^T B_{k}]_{a\bm{i}'} [A^T B_{k'}]_{a'\bm{i}'} .\label{vqacr_lemma_qntk_lower_bound_xy_2_2}
\end{align}
Eqs.~(\ref{vqacr_lemma_qntk_lower_bound_xy_2_1}) and (\ref{vqacr_lemma_qntk_lower_bound_xy_2_2}) are derived by using Eqs.~(\ref{vqacr_lemma_qntk_lower_bound_Ka_eq}) and (\ref{vqacr_lemma_qntk_lower_bound_1_3}), respectively. The notation $\bm{i}'$ is defined by zero-padding the vector $\bm{i} \in \{1,2\}^{\otimes S}$ in the space $\{0,1,2,3\}^{\otimes N}$.

Next, we proceed with the matrix formulation $K_{d}$. We have 
\begin{align}
\mathop{\mathbb{E}}_{H_{d} \in \mathcal{H}} K_{d} ={}&  \frac{1}{2^{S}} \sum_{\bm{i} \in \{1,2\}^{\otimes S}} \sum_{k,k'=1}^{N} o_{k} o_{k'} \left[ [A^T B_{k}]_{a\bm{i}'} [A^T B_{k'}]_{a'\bm{i}'} \right]_{aa'} \label{vqacr_lemma_qntk_lower_bound_xy_3_1} \\
={}& \frac{1}{2^{S}} \sum_{k,k' \in \mathcal{S}_{d}} o_{k} o_{k'} A_{d}^T C_{k} C_{k'}^{T} A_{d}. \label{vqacr_lemma_qntk_lower_bound_xy_3_2}
\end{align}
The derivation  of Eq.~(\ref{vqacr_lemma_qntk_lower_bound_xy_3_1}) is similar to that of Eq.~(\ref{vqacr_lemma_qntk_lower_bound_3_1}) with $w=0$, since the notation $\bm{i}'$ in this Lemma has no $z$ terms.
Eq.~(\ref{vqacr_lemma_qntk_lower_bound_xy_3_2}) follows from using notations $A_{d}$ and $C_k$. The matrix $A_{d} \in \R^{2^{S} \times |\A|}$ is derived by extracting rows from matrix $A$ corresponding to indices from the set $\{\bm{i}| i_x \in \{1,2\} \ \forall x \in \mathcal{S}_{d}, i_x = 0 \ \forall x \notin \mathcal{S}_{d}\}$. The matrix $C_k \in \R^{2^{S} \times 2^{S}}$ is derived by extracting columns and rows from $B_k$ corresponding to indices from the set mentioned above, i.e., 
\begin{align} \label{vqacr_lemma_qntk_lower_bound_xy_C_k}
C_k = I_2^{\otimes (k_{d}-1)} \otimes
\begin{bmatrix}
& -1 \\
1 &
\end{bmatrix}
\otimes I_2^{\otimes (S-k_{d})},
\end{align}
where $k_{d}$ is the index of $k$ in the set $\mathcal{S}_{d}$. 

We remark that $C_k$ is of full rank with singular values $\pm 1$. The eigenvalue of the positive semi-definite matrix $\sum_{k,k' \in \mathcal{S}_{d}} c_{k} c_{k'} C_{k} C_{k'}^{T}$ can be bounded as follows
\begin{align} \label{vqacr_lemma_qntk_lower_bound_xy_C_k_least_eiganvalue}
\|O\|_2^2 = \left( \sum_{k=1}^{N} |o_k| \right)^2 \geq{} \lambda \bigg[ \sum_{k,k' \in \mathcal{S}_{d}} o_{k} o_{k'} C_{k} C_{k'}^{T} \bigg] \geq  \min_{\bm{g}\in \{-1,1\}^{\otimes S}} (\bm{g}^T \bm{o}_{d})^2 \geq \Delta_{S},
\end{align}
where we denote by $\bm{o}_{d} \in \R^{S}$ the vector formed by extracting elements in $\bm{c}$ corresponding to indices in $\mathcal{S}_{d}$. 

Next, we proceed to take the average w.r.t. the qubit group $\mathcal{Q}$. Denote by $\succcurlyeq$ the semi-definite partial order, i.e., $X \succcurlyeq Y \iff \bm{a}^T (X-Y) \bm{a} \geq 0$, we have
\begin{align}
K_{\infty}(\bm{0}) = \mathop{\mathbb{E}}_{\mathcal{S}_{d} \in \mathcal{Q} } \mathop{\mathbb{E}}_{H_{d} \in \mathcal{H}} K_{d} \succcurlyeq{}& \frac{1}{2^{S} |\mathcal{Q}|} \sum_{\mathcal{S}_{d} \in \mathcal{Q}} \Delta_{S} A_{d}^T A_{d},\label{vqacr_lemma_qntk_lower_bound_xy_4_1} \\
={}& \frac{1}{2^{S} |\mathcal{Q}|} \Delta_{S} A_{\mathcal{Q}}^T A_{\mathcal{Q}}.\label{vqacr_lemma_qntk_lower_bound_xy_4_2}
\end{align}
Eq.~(\ref{vqacr_lemma_qntk_lower_bound_xy_4_1}) is obtained from Eq.~(\ref{vqacr_lemma_qntk_lower_bound_xy_3_2}) by noticing that $XYX^T$ is positive semi-definite if $Y$ is positive semi-definite. Eq.~(\ref{vqacr_lemma_qntk_lower_bound_xy_4_2}) is obtained by using the notation $A_{\mathcal{Q}}$, which is the matrix with dimension $(2^{S} |\mathcal{Q}|, |\A|)$ formed by vertically concatenating matrices $A_{d}$ for all $\mathcal{S}_{d} \in \mathcal{Q}$. 
Similarly we have
\begin{align}
K_{\infty}(\bm{0}) = \mathop{\mathbb{E}}_{\mathcal{S}_{d} \in \mathcal{Q} } \mathop{\mathbb{E}}_{H_{d} \in \mathcal{H}} K_{d} \preccurlyeq{}& \frac{1}{2^{S} |\mathcal{Q}|} \sum_{\mathcal{S}_{d} \in \mathcal{Q}} \|O\|_2^2 A_{d}^T A_{d}, \notag \\
={}& \frac{1}{2^{S} |\mathcal{Q}|} \|O\|_2^2 A_{\mathcal{Q}}^T A_{\mathcal{Q}}. \notag
\end{align}
Thus, we obtain Eq.~(\ref{vqacr_lemma_qntk_lower_bound_xy_eq}).

\end{proof}

\subsection{Bounds of $\lmin \left[ K_{\infty}(\bm{0}) \right]$ and $\lmax \left[ K_{\infty}(\bm{0}) \right]$}

Different randomness within $\{\rho_a\}_{a \in \A}$ or elements of $A$ and different choice of $\H$ can lead to separate lower bound guarantees on the asymptotic QNTK. 
Here we provide two examples. The former is the random state distribution described in the main text with a bias on local terms, where the Hamiltonian set is $\H=\{X,Y,Z\}^{S}$. The latter is a state distribution based on the qubit embedding model and the Hamiltonian set is $\H=\{X,Y\}^{S}$. Roughly speaking, we found that for both cases, the least eigenvalue of the asymptotic QNTK deviates from zero with high probability when $S\geq \O(\log|\A|)$. Detailed results are presented in Lemmas~\ref{vqacr_lemma_qntk_least_eigen_gaussian_matrix} and \ref{vqacr_lemma_qntk_least_eigen_random_sphere}, respectively. Both results are derived based on the random matrix theory.

\begin{lemma}
\label{vqacr_lemma_qntk_least_eigen_gaussian_matrix}
We follow conditions and notations in Lemma~\ref{vqacr_lemma_qntk_lower_bound}. Let $\rho_{a}=\sum_{\bm{p} \in \{0,1,2,3\}^{\otimes N}} \frac{1}{2^N} A_{\bm{p}a} \sigma_{\bm{p}}$ be the Pauli decomposition of state $\rho_a$, where $A$ follows Proposition~\ref{vqacr_prop_state_distribution_local}. Denote by $A_{\Q}$ the matrix with dimension $((3^S-1)|\Q|, |\A|)$ formed by vertically concatenating matrices $A_{\Q,w}$ for $w \in \{0,1,\cdots,S-1\}$. Denote by $\alpha_{\min}:=\min_{\bm{i}} \alpha_{\bm{i}}$, $\alpha_{\max}:=\max_{\bm{i}} \alpha_{\bm{i}}$, and ${\rm Var}_{\Q} := {\rm diag}(\alpha_{\bm{i}})$, where $\bm{i}$ corresponds to the index of rows of $A_{\Q}$. Denote by $A'_{\Q} := {\rm Var}_{\Q}^{-1/2} A_{\Q}$, $ A'_{\Q,\min}:=\frac{1 }{\sqrt{(3^S-1)|\Q|}} \min_{a\in\A} \left\| [A'_{\Q}]_{\cdot,a} \right\|_2 $, and $ A'_{\Q,\max}:=\frac{1 }{\sqrt{(3^S-1)|\Q|}} \max_{a\in\A} \left\| [A'_{\Q}]_{\cdot,a} \right\|_2 $.
Then there exists two constants $C>0$ and $k<1$ such that for $S \geq{} \log_3 \left[ 1 + \frac{C^2  |\A|^2 \ln |A| }{ k^2 A_{\Q, \min}'^4 |\Q| \delta^2 } \right]$, we have
\begin{align}
(1-k) \frac{3^S-1}{3^S} \Delta_S \alpha_{\min} A_{\Q, \min}'^2 \leq{}& \lmin \left[ K_{\infty} (\bm{0}) \right] \leq{}   \lmax \left[ K_{\infty} (\bm{0}) \right] \leq{} (1+k) \frac{3^S-1}{3^S} \|O\|_2^2 \alpha_{\max} A_{\Q, \max}'^2
\label{vqacr_lemma_qntk_least_eigen_gaussian_matrix_eq_1}
\end{align}
with probability at least $1-\delta$.
\end{lemma}

\begin{proof}
Denote by $D_{\Q}$ a $|\A|\times |\A|$ diagonal matrix with elements $[D_{\Q}]_{aa}:=\|[A'_{\Q}]_{\cdot,a}\|_2$, i.e., $D_{\Q}$ stores the norm of columns of $A'_{\Q}$. Let 
\begin{equation}\label{vqacr_lemma_qntk_least_eigen_gaussian_matrix_R}
R := \sqrt{(3^S-1)|\Q|} {\rm Var}_{\Q}^{-1/2} A_{\Q} D_{\Q}^{-1} .
\end{equation}
We denote that $R$ matches the random matrix condition in Lemma~\ref{vqacr_lemma_qntk_spectral_norm_of_independent_columns}. Specifically, columns of $R$ are independent isotropic random vectors:
\begin{align}
\mathop{\E}_{\bm{p}, \bm{q}, a, b} R_{\bm{p}a} R_{\bm{q}b} ={}& (3^S-1)|\Q| \mathop{\E}_{\bm{p}, \bm{q}, a, b} \frac{[A'_\Q]_{\bm{p},a}[A'_\Q]_{\bm{q},b}}{\left\|[A'_\Q]_{\cdot,a}\right\|_2 \left\|[A'_\Q]_{\cdot,b}\right\|_2} = \delta_{\bm{p}\bm{q}} \delta_{ab} . \notag
\end{align}
All columns of $R$ have the same norm:
\begin{align}
\left\|R_{\cdot,a}\right\|_2 = \sqrt{(3^S-1)|\Q|} \frac{\left\|[A'_\Q]_{\cdot,a}\right\|_2}{\left\|[A'_\Q]_{\cdot,a}\right\|_2} = \sqrt{(3^S-1)|\Q|} .
\end{align}
We have
\begin{align}
h_{\Q,\A} := {}& \frac{1}{(3^S-1)|\Q|} \mathop{\E} \max_{a \in \A} \sum_{b\in\A, b \neq a} \left( {R_{\cdot,a}}^T R_{\cdot,b} \right)^2 \notag \\
={}& {(3^S-1)|\Q|} \mathop{\E} \max_{a \in \A} \sum_{b\in\A, b \neq a} \left( \frac{ {[A'_\Q]_{\cdot,a}}^T {[A'_\Q]_{\cdot,b}} }{\left\|[A'_\Q]_{\cdot,a}\right\|_2 \left\|[A'_\Q]_{\cdot,b}\right\|_2 }  \right)^2 \tag{using Eq.~(\ref{vqacr_lemma_qntk_least_eigen_gaussian_matrix_R})} \\
\leq{}& \frac{1}{(3^S-1)|\Q|} \mathop{\E} \sum_{a,b\in\A, b \neq a} \frac{1}{ A_{\Q,\min}'^4 } \left( \sum_{\bm{p}} [A'_\Q]_{\bm{p},a} [A'_\Q]_{\bm{p},b}   \right)^2 \tag{using the notation of $D_{\Q,\min}$} \\
={}& \frac{|\A| (|\A|-1)}{ A_{\Q, \min}'^4}  , \label{vqacr_lemma_qntk_least_eigen_gaussian_matrix_3_4}
\end{align}
where Eq.~(\ref{vqacr_lemma_qntk_least_eigen_gaussian_matrix_3_4}) is obtained by using the distribution condition of $A$ in Proposition~\ref{vqacr_prop_state_distribution_local}. 
Then by using Lemma~\ref{vqacr_lemma_qntk_spectral_norm_of_independent_columns}, there exists a constant $C$ such that 
\begin{align}
\E \left| \frac{1}{(3^S-1)|\Q|} \left( s \left[R \right] \right)^2 - 1 \right|	\leq C \sqrt{\frac{ h_{\Q,\A} \ln |\A|}{(3^S-1)|\Q|} } , \label{vqacr_lemma_qntk_least_eigen_gaussian_matrix_3_6}
\end{align}
where $s[X]$ denotes the singular value of a matrix $X$.
Thus, with probability at least $1-\delta$, we have 
\begin{align}
\lmin \left[ A_{\Q}^T A_{\Q} \right] ={}& \left( s_{\min} \left[ A_{\Q} \right] \right)^2 \notag \\
={}& \frac{1}{(3^S-1)|\Q|} \left( s_{\min} \left[ {\rm Var}_{\Q}^{1/2} R D_{\Q} \right] \right)^2 \tag{using Eq.~(\ref{vqacr_lemma_qntk_least_eigen_gaussian_matrix_R})} \\
\geq{}& \alpha_{\min} A_{\Q, \min}'^2 \left( s_{\min} \left[ R \right] \right)^2 \notag \\
\geq{}& {(3^S-1)|\Q|} \alpha_{\min} A_{\Q, \min}'^2 \left(  1 - \frac{C}{\delta} \sqrt{ \frac{h_{\Q,\A} \ln |\A|}{(3^S-1)|\Q|} } \right) \tag{using Eq.~(\ref{vqacr_lemma_qntk_least_eigen_gaussian_matrix_3_6})} \\
\geq{}& {(3^S-1)|\Q|} \alpha_{\min} A_{\Q, \min}'^2 \left(  1 - \frac{C |\A|}{\delta A_{\Q, \min}'^2} \sqrt{ \frac{\ln |\A|}{(3^S-1)|\Q|} } \right) \tag{using Eq.~(\ref{vqacr_lemma_qntk_least_eigen_gaussian_matrix_3_4})} \\
\geq{}& (1-k) {(3^S-1)|\Q|} \alpha_{\min} A_{\Q, \min}'^2 \label{vqacr_lemma_qntk_least_eigen_gaussian_matrix_4_3}
\end{align}
and
\begin{align}
\lmax \left[ A_{\Q}^T A_{\Q} \right] ={}& \left( s_{\max} \left[ A_{\Q} \right] \right)^2 \notag \\
={}& \frac{1}{(3^S-1)|\Q|} \left( s_{\max} \left[ {\rm Var}_{\Q}^{1/2} R D_{\Q} \right] \right)^2 \tag{using Eq.~(\ref{vqacr_lemma_qntk_least_eigen_gaussian_matrix_R})} \\
\leq{}& \alpha_{\max} A_{\Q, \max}'^2 \left( s_{\max} \left[ R \right] \right)^2 \notag \\
\leq{}& {(3^S-1)|\Q|} \alpha_{\max} A_{\Q, \max}'^2 \left(  1 + \frac{C}{\delta} \sqrt{ \frac{h_{\Q,\A} \ln |\A|}{(3^S-1)|\Q|} } \right) \tag{using Eq.~(\ref{vqacr_lemma_qntk_least_eigen_gaussian_matrix_3_6})} \\
\leq{}& {(3^S-1)|\Q|} \alpha_{\max} A_{\Q, \max}'^2 \left(  1 + \frac{C |\A|}{\delta A_{\Q, \min}'^2} \sqrt{ \frac{\ln |\A|}{(3^S-1)|\Q|} } \right) \tag{using Eq.~(\ref{vqacr_lemma_qntk_least_eigen_gaussian_matrix_3_4})} \\
\leq{}& (1+k) {(3^S-1)|\Q|} \alpha_{\max} A_{\Q, \max}'^2 .\label{vqacr_lemma_qntk_least_eigen_gaussian_matrix_5_3}
\end{align}
Eqs.~(\ref{vqacr_lemma_qntk_least_eigen_gaussian_matrix_4_3}) and (\ref{vqacr_lemma_qntk_least_eigen_gaussian_matrix_5_3}) follows from the condition on $S$:
\begin{equation}
S \geq{} \log_3 \left[ 1 + \frac{C^2 |\A|^2 \ln |A| }{ k^2 A_{\Q, \min}'^4 |\Q| \delta^2 } \right] . \notag
\end{equation}
By using Lemma~\ref{vqacr_lemma_qntk_lower_bound}, we could finally provide the lower bound and the upper bound on the least and the largest eigenvalue of the asymptotic QNTK, respectively. With probability at least $1-\delta$, we have
\begin{align}
\lmin \left[ K_{\infty} (\bm{0}) \right] \geq{}& \frac{1}{3^{S} |\mathcal{Q}|} \Delta_{S} \lmin \left[ A_{\mathcal{Q}}^T A_{\mathcal{Q}} \right] \geq{} (1-k) \frac{3^S-1}{ 3^S} \Delta_S \alpha_{\min} A_{\Q, \min}'^2 , \notag \\
\lmax \left[ K_{\infty} (\bm{0}) \right] \leq{}& \frac{1}{3^{S} |\mathcal{Q}|} \|O\|_2^2 \lmax \left[ A_{\mathcal{Q}}^T A_{\mathcal{Q}} \right] \leq{} (1+k) \frac{3^S-1}{3^S} \|O\|_2^2 \alpha_{\max} A_{\Q, \max}'^2 . \notag
\end{align}

\end{proof}

\begin{lemma}\label{vqacr_lemma_qntk_least_eigen_random_sphere}
We follow conditions and notations in Lemma~\ref{vqacr_lemma_qntk_lower_bound_xy}. Let $\rho_a$ be the density matrix of the qubit embedding of the vector $\bm{x}_{a}$ via Eq.~(\ref{vqacr_qubit_embedding_eq}), where each entry of $\bm{x}_{a}$ is an independent random variable distributed uniformly in $[-1,1]$. Then there exists two constants $C>0$ and $k<1$ such that for $S \geq 2+ \log_2 \frac{C^2 |\A|^2 \ln |\A|}{k^2|\Q|\delta^2}$, 
\begin{equation}\label{vqacr_lemma_qntk_least_eigen_random_sphere_eq}
(1-k) \frac{\Delta_{S} }{2^{S}} \leq{} \lmin \left[ K_{\infty} (\bm{0}) \right] \leq{} \lmax \left[ K_{\infty} (\bm{0}) \right] \leq{} (1+k)\frac{ \|O\|_2^2 }{2^{S}}
 \end{equation}
with probability at least $1-\delta$.
\end{lemma}

\begin{proof}

By using Lemma~\ref{vqacr_lemma_qntk_lower_bound_xy}, we have
\begin{align*}
\frac{\Delta_{S}}{2^{2S} |\mathcal{Q}|}  \left( s_{\min} \left[ {2^{S/2}} A_{\mathcal{Q}} \right] \right)^2 \leq{}& \lmin  \left[ K_{\infty} (\0) \right] \leq{} \lmin  \left[ K_{\infty} (\0) \right] \leq{} \frac{\|O\|_2^2}{2^{2S} |\mathcal{Q}|}  \left( s_{\max} \left[ {2^{S/2}} A_{\mathcal{Q}} \right] \right)^2 ,
\end{align*}
so it suffices to prove
\begin{align}
s_{\min} \left[ 2^{{S}/{2}} A_{\mathcal{Q}} \right]  \geq{}& \sqrt{2^S |\mathcal{Q}|} - \frac{C}{\delta} \sqrt{|\A|(|\A|-1)\ln |\A|} \geq \sqrt{(1-k)2^S |\Q|} , \label{vqacr_lemma_qntk_least_eigen_random_sphere_key} \\
s_{\max} \left[ 2^{{S}/{2}} A_{\mathcal{Q}} \right]  \leq{}& \sqrt{2^S |\mathcal{Q}|} + \frac{C}{\delta} \sqrt{|\A|(|\A|-1)\ln |\A|} \leq \sqrt{(1+k)2^S |\Q|} , \label{vqacr_lemma_qntk_least_eigen_random_sphere_key_2}
\end{align}
with probability at least $1-\delta$, where the final inequality is derived from the condition of $S$. Our main idea is to employ Lemma~\ref{vqacr_lemma_qntk_spectral_norm_of_independent_columns} to obtain bounds of singular values of the matrix $2^{{S}/{2}} A_{\mathcal{Q}}$ with $d_1=2^{S}|\mathcal{Q}|$ and $d_2=|\A|$.

By using the qubit embedding with $R_Z$ rotations on the initial state $\otimes_{n=1}^{N} \frac{|0\>+|1\>}{\sqrt{2}}$, we have
\begin{align}
\rho_a ={}& \bigotimes_{n=1}^{N} \exp[-iZ\pi x_{a,n}/2] \frac{(|0\>+|1\>)(\<0|+\<1|)}{2} \exp[iZ\pi x_{a,n}/2] \notag \\
={}& \bigotimes_{n=1}^{N} \exp[-iZ\pi x_{a,n}/2] \frac{I+X}{2} \exp[iZ\pi x_{a,n}/2] \notag \\
={}& \bigotimes_{n=1}^{N} \frac{I+X\cos(\pi x_{a,n})+Y\sin(\pi x_{a,n})}{2} .\label{vqacr_lemma_qntk_least_eigen_random_sphere_1_3}
\end{align}
Denote by $[A_{\mathcal{Q}}]_{\cdot,a}$ the $a$-th column of $A_{\mathcal{Q}}$. By using the definition of $A_{\mathcal{Q}}$ and Eq.~(\ref{vqacr_lemma_qntk_least_eigen_random_sphere_1_3}), we have
\begin{align}
[A_{\mathcal{Q}}]_{\cdot,a} ={}& \bigoplus_{\mathcal{S}_d \in \mathcal{Q}} \bigotimes_{n \in \mathcal{S}_d} 
\begin{bmatrix}
\cos(\pi x_{a,n}) \\
\sin(\pi x_{a,n})
 \end{bmatrix} \label{vqacr_lemma_qntk_least_eigen_random_sphere_1_4} \\
:={}& \bigoplus_{\mathcal{S}_d \in \mathcal{Q}} \bm{v}_{d,a} , \label{vqacr_lemma_qntk_least_eigen_random_sphere_1_5}
\end{align}
where we split $[A_{\mathcal{Q}}]_{\cdot,a}$ into $|\mathcal{Q}|$ vectors $\bm{v}_{1,a},\cdots,\bm{v}_{|\mathcal{Q}|,a}$ in Eq.~(\ref{vqacr_lemma_qntk_least_eigen_random_sphere_1_5}). The vector $\bm{v}_{d,a}$ corresponds to the qubit group $\mathcal{S}_{d}$.
Thus, we have
\begin{align}
\l\|2^{S/2} [A_{\mathcal{Q}}]_{\cdot,a} \r\|_2^2 ={}& 2^S \sum_{d=1}^{|\mathcal{Q}|} \l\| \bm{v}_{d,a} \r\|_2^2 \notag \\
={}& 2^S \sum_{d=1}^{|\mathcal{Q}|} \prod_{n \in \mathcal{S}_{d}} \left( \cos^2(\pi x_{a,n}) + \sin^2(\pi x_{a,n}) \right) \notag \\
={}& 2^S |\mathcal{Q}| .
\end{align}

Next, we prove that $2^{S/2}[A_{\mathcal{Q}}]_{\cdot,a}$ are independent isotropic random vectors, i.e.,
\begin{align}
\mathop{\E}_{\bm{x}_{a}} 	2^{{S}} [A_{\mathcal{Q}}]_{\cdot,a} [A_{\mathcal{Q}}]_{\cdot,a}^T ={}& I_{2^S|\mathcal{Q}|} . \label{vqacr_lemma_qntk_least_eigen_random_sphere_iso_cond}
\end{align}

We consider the diagonal term first. By noticing Eq.~(\ref{vqacr_lemma_qntk_least_eigen_random_sphere_1_5}), we remark that the vector $[A_{\mathcal{Q}}]_{\cdot,a}$ could be written as the direct sum of $|\Q|$ vectors $\bm{v}_{d,a}$, where the element of each $\bm{v}_{d,a}$ is assigned with
an index $\bm{p}_{d} \in \{1,2\}^{\otimes S}$. Thus, we could denote by  $v_{\bm{p}_{d},a}$ the element of $[A_{\mathcal{Q}}]_{\cdot,a}$. Suppose $\S_d=\{q_1, \cdots,q_S\}$. For any $\bm{p}_{d}$, we have
\begin{align}
\mathop{\E}_{\bm{x}_{a}} 2^S v_{\bm{p}_{d},a}^2 ={}& 2^S \mathop{\E}_{\bm{x}_{a}} \prod_{n=1}^{S} \left( \delta_{p_{n}1} \cos^2(\pi x_{a,q_n}) + \delta_{p_{n}2} \sin^2(\pi x_{a,q_n}) \right) \label{vqacr_lemma_qntk_least_eigen_random_sphere_3_1} \\
={}& 2^S \prod_{n=1}^{S} \mathop{\E}_{{x}_{a,n}} \left( \delta_{p_{n}1} \cos^2(\pi x_{a,q_n}) + \delta_{p_{n}2} \sin^2(\pi x_{a,q_n}) \right) \label{vqacr_lemma_qntk_least_eigen_random_sphere_3_2} \\
={}& 2^S \prod_{n=1}^{S} \frac{\delta_{p_{n}1} + \delta_{p_{n}2}}{2} \tag{using $\E_{x} \cos^2(\pi x) = \E_{x} \sin^2(\pi x) = \frac{1}{2}$} \\
={}& 1 . \label{vqacr_lemma_qntk_least_eigen_random_sphere_3_4} 
\end{align}
Eq.~(\ref{vqacr_lemma_qntk_least_eigen_random_sphere_3_1}) follows from Eq.~(\ref{vqacr_lemma_qntk_least_eigen_random_sphere_1_4}). Eq.~(\ref{vqacr_lemma_qntk_least_eigen_random_sphere_3_2}) is obtained by noticing that $\{x_n^{m}\}_{n=1}^{N}$ are independent random variables. Eq.~(\ref{vqacr_lemma_qntk_least_eigen_random_sphere_3_4}) is obtained by noticing that $p_{n} \in \{1,2\}$.
Next, we consider the off-diagonal term in $[A_{\mathcal{Q}}]_{\cdot,a} [A_{\mathcal{Q}}]_{\cdot,a}^T$. Similar to the diagonal case, we assign two entries $v_{\bm{p}_{d},a}$ and $v_{\bm{p}'_{d'},a}$ with two indices $\bm{p}_{d}$ and $\bm{p}'_{d'}$, respectively. 
For the case $d\neq d'$, i.e. $\S_d \neq \S_{d'}$, there exists a qubit $q \in \S_d$ and $q \notin \S_{d'}$. Then the expectation $\mathop{\E}_{\bm{x}_a} v_{\bm{p}_{d},a}v_{\bm{p}'_{d'},a}$ contains a factor $\E \cos(\pi x_{a,q})=0$ or $\E \sin(\pi x_{a,q})=0$, where both induce the zero result. For the case $d=d'$,  we remark that $\bm{p}_{d} \neq \bm{p}'_{d}$ due to the off-diagonal condition. Thus, there exists at least one index $i\in\{1,\cdots,S\}$ such that $p_i \neq p'_i$ and the qubit $q_i \in \S_d$. Therefore the expectation $\mathop{\E}_{\bm{x}_a} v_{\bm{p}_{d},a}v_{\bm{p}'_{d'},a}$ contains a factor $\E \cos(\pi x_{a,q_i}) \sin(\pi x_{a,q_i})=0$, which induces the zero result.
Thus, we have proved Eq.~(\ref{vqacr_lemma_qntk_least_eigen_random_sphere_iso_cond}).

Next, we proceed to calculate the term $d_3$ in Lemma~\ref{vqacr_lemma_qntk_spectral_norm_of_independent_columns}. We have
\begin{align}
[A_{\mathcal{Q}}]_{\cdot,a}^T [A_{\mathcal{Q}}]_{\cdot,b} ={}& \sum_{\mathcal{S}_d \in \mathcal{Q}} \prod_{n\in \mathcal{S}_d} \left( \cos\pi x_{a,n} \cos\pi x_{b,n} +  \sin\pi x_{a,n} \sin\pi x_{b,n} \right) \label{vqacr_lemma_qntk_least_eigen_random_sphere_5_1} \\
={}& \sum_{\mathcal{S}_d \in \mathcal{Q}} \prod_{n\in \mathcal{S}_d}  \cos \left( \pi x_{a,n} - \pi x_{b,n} \right) , \label{vqacr_lemma_qntk_least_eigen_random_sphere_5_2} 
\end{align}
where Eq.~(\ref{vqacr_lemma_qntk_least_eigen_random_sphere_5_1}) follows from Eq.~(\ref{vqacr_lemma_qntk_least_eigen_random_sphere_1_4}). Thus
\begin{align}
d_3 ={}& \frac{1}{2^S|\mathcal{Q}|} \E \max_{a \in \A} \sum_{b \in \A, b \neq a} \left( 2^S [A_{\mathcal{Q}}]_{\cdot,a}^T [A_{\mathcal{Q}}]_{\cdot,b} \right)^2  \notag \\
\leq{}& \frac{2^S}{|\mathcal{Q}|} \E \sum_{a,b \in \A, b \neq a}\left( \sum_{\mathcal{S}_d \in \mathcal{Q}} \prod_{n\in \mathcal{S}_d}  \cos \left( \pi x_{a,n} - \pi x_{b,n} \right) \right)^2 \label{vqacr_lemma_qntk_least_eigen_random_sphere_6_2} \\
={}& \frac{2^S}{|\mathcal{Q}|} \E \sum_{a,b \in \A, b \neq a} \sum_{\mathcal{S}_d \in \mathcal{Q}} \prod_{n\in \mathcal{S}_d}  \cos^2 \left( \pi x_{a,n} - \pi x_{b,n} \right)  \label{vqacr_lemma_qntk_least_eigen_random_sphere_6_3} \\
={}& \frac{2^S}{|\mathcal{Q}|} \sum_{a,b \in \A, b \neq a} \sum_{\mathcal{S}_d \in \mathcal{Q}} \prod_{n\in \mathcal{S}_d}  \frac{1}{2}  \notag \\
={}& |\A|(|\A|-1) . \notag
\end{align}
Eq.~(\ref{vqacr_lemma_qntk_least_eigen_random_sphere_6_2}) follows from Eq.~(\ref{vqacr_lemma_qntk_least_eigen_random_sphere_5_2}). Eq.~(\ref{vqacr_lemma_qntk_least_eigen_random_sphere_6_3}) is derived by noticing that for $\mathcal{S}_d \neq \mathcal{S}_{d'}$, the expectation  of the $\cos$ term is zero. 

Thus, by using Lemma~\ref{vqacr_lemma_qntk_spectral_norm_of_independent_columns}, there exists a constant $C$ such that
\begin{align}
\E \max_{j} \left| s_j(2^{S/2} A_{\mathcal{Q}}) - \sqrt{2^S |\mathcal{Q}|} \right| \leq{} C \sqrt{|\A|(|\A|-1) \ln |\A|} .
\label{vqacr_lemma_qntk_least_eigen_random_sphere_7_1}
\end{align}
Eqs.~(\ref{vqacr_lemma_qntk_least_eigen_random_sphere_key}) and (\ref{vqacr_lemma_qntk_least_eigen_random_sphere_key_2}) follow from Eq.~(\ref{vqacr_lemma_qntk_least_eigen_random_sphere_7_1}) by using the 
Markov's inequality. Thus, we have proved Lemma~\ref{vqacr_lemma_qntk_least_eigen_random_sphere}.

\end{proof}

\subsection{Proof of Lemma~\ref{vqacr_lemma_qntk_least_eigen_gaussian_matrix_main}}

\label{vqacr_lemma_qntk_stability_proof1}

The main idea is to analyze the stability of the QNTK when the depth is large but finite. Specifically, we are interested in the distance between the asymptotic QNTK $K_{\infty}(\bm{0})$ with the finite QNTK $K (\bm{0})$ with $D$ gates.
We follow notations in Lemmas~\ref{vqacr_lemma_qntk_lower_bound} and \ref{vqacr_lemma_qntk_least_eigen_gaussian_matrix}. 
We remark that 
\begin{equation*}
K(\bmt) = \frac{1}{D} J(\bmt)^T J(\bmt)
\end{equation*}
by Eq.~(\ref{vqacr_qntk_eq}), where the row $J_{d}(\bm{\theta}) = \{ \frac{\partial z_a}{\partial \theta_{d}} (\bm{\theta}) \}_{a \in \A}$. Thus, $K(\bm{0})$ and $K_{\infty}(\bm{0})$ are the empirical and the explicit covariance matrix of independent random vectors $\{J_{d}(\bm{0})\}$, respectively. Since each Hamiltonian $H_d$ involves $S$ qubits, we have
\begin{align}
\mathop{\E}_{a} J_{da}(\bm{0})^2 ={}& \mathop{\E}_{a} \frac{1}{4} \left( \Tr \left[ i [O, H_d] \rho_a \right] \right)^2 \leq \|O\|_2^2 \alpha_{\max}, \notag
\end{align}
where $\alpha_{\max} := \max_{\|\bm{i}\|_1=S} \alpha_{\bm{i}}$. We remark that each element in the vector $J_d(\bm{0})$ is independent from each other due to the independent sample $a \in \A$. By using the central limit theorem, there exists an absolute constant $c_1>0$, such that 
\begin{align}
P \left( \|J_{d}(\bm{0})\|_2^2 \geq{} c_1 |\A| \|O\|_2^2 \alpha_{\max} \ln \frac{D}{\delta} \right) \leq{}& 1 - \frac{\delta}{3D} . \notag
\end{align}
Thus, with probability at least $1-\delta/3$ we have that for all $d = 1,2,\cdots,D$,
\begin{align}
\|J_{d}(\bm{0})\|_2^2 \leq{} c_1 |\A| \|O\|_2^2 \alpha_{\max} \ln \frac{D}{\delta} . \label{vqacr_lemma_qntk_stability_proof1_bound_Jd_all}
\end{align}
By using Lemma~\ref{vqacr_lemma_qntk_operator_norm_of_covariance_matrix} and Eq.~(\ref{vqacr_lemma_qntk_stability_proof1_bound_Jd_all}), we have that with probability at least $1-2\delta/3$, 
\begin{align}
{}& \left\| K(\bm{0}) - K_{\infty}(\bm{0}) \right\|_2 \notag \\
={}& \left\| \frac{1}{D} \sum_{d=1}^{D} J_d(\bm{0}) J_d(\bm{0})^T - K_{\infty}(\bm{0}) \right\|_2 \notag \\
\leq{}& \epsilon \cdot \max \left( \sqrt{\left\|K_{\infty}(\bm{0})\right\|_2} , \epsilon \right) , \label{vqacr_lemma_qntk_stability_proof1_2_3}
\end{align}
where 
\begin{equation}\label{vqacr_lemma_qntk_stability_proof1_epsilon1}
\epsilon = \sqrt{\frac{c_1 |\A| \|O\|_2^2 \alpha_{\max} }{D} \ln \frac{D}{\delta}} \cdot \sqrt{\frac{1}{c_2} \ln \frac{|\A|}{\delta}}
\end{equation}
and $c_2>0$ is an absolute constant. Let $C_2:=9\max(c_1/c_2,1)$ in the condition $D_0 = C_2 \frac{ \|O\|_2^4 \alpha_{\max}^2 A_{\Q, \max}'^2 }{k^2 \Delta_S^2 \alpha_{\min}^2 A_{\Q, \min}'^4} |\A| \ln \frac{|\A|}{\delta} $ such that $D\geq D_0 \ln \frac{D_0}{\delta}$, we have
\begin{align}
\frac{\frac{D}{\delta}}{\ln \frac{D}{\delta}} \geq{}& \frac{ \frac{D_0}{\delta} \ln \frac{D_0}{\delta} }{\ln \frac{D_0}{\delta} + \ln \ln \frac{D_0}{\delta} } \tag{function $f(x)=x/\ln x$ is increasing for $x\geq e$} \\
\geq{}& \frac{\frac{D_0}{\delta} \ln \frac{D_0}{\delta}}{\ln \frac{D_0}{\delta} + \frac{1}{2} \ln \frac{D_0}{\delta} } \tag{$\ln x \leq \frac{x}{2}$ for $x>0$} \\
={}& \frac{2}{3} \frac{D_0}{\delta} = \frac{6}{k^2} \max(c_1/c_2, 1) \frac{|\A| \|O\|_2^4 \alpha_{\max}^2 A_{\Q, \max}'^2 }{\Delta_S^2 \alpha_{\min}^2 A_{\Q, \min}'^4} \frac{1}{\delta} \ln \frac{|\A|}{\delta} . \label{vqacr_lemma_qntk_stability_proof1_bound_DlnD}
\end{align}

Using Eq.~(\ref{vqacr_lemma_qntk_stability_proof1_bound_DlnD}) in Eq.~(\ref{vqacr_lemma_qntk_stability_proof1_epsilon1}), we obtain 
\begin{align}
\epsilon ={}& \sqrt{\frac{c_1 |\A| \|O\|_2^2 \alpha_{\max} }{c_2} \frac{1}{\delta} \ln \frac{|\A|}{\delta}  } \cdot  \sqrt{\frac{\ln \frac{D}{\delta}}{\frac{D}{\delta}}} \leq{} \frac{k}{\sqrt{6}} \cdot \frac{\Delta_S \alpha_{\min} A_{\Q, \min}'^2 }{ \sqrt{  \|O\|_2^2 \alpha_{\max} A_{\Q, \max}'^2 } } . \label{vqacr_lemma_qntk_stability_proof1_epsilon2}
\end{align}
Besides, by using Lemma~\ref{vqacr_lemma_qntk_least_eigen_gaussian_matrix}, 
we have
\begin{align}
\left(1-\frac{k}{2}\right) \Delta_S \alpha_{\min} A_{\Q, \min}'^2 \leq{}& \lmin \left[ K_{\infty} (\bm{0}) \right] \leq{}   \lmax \left[ K_{\infty} (\bm{0}) \right] \leq{} \left(1+\frac{k}{2}\right) \|O\|_2^2 \alpha_{\max} A_{\Q, \max}'^2
\label{vqacr_lemma_qntk_stability_proof1_spectrum_norm_K_infty}
\end{align}
with probability at least $1-\delta/3$ for $S \geq{} \log_3 \left[ 1 + \frac{C_1^2  |\A|^2 \ln |A| }{ k^2 A_{\Q, \min}'^4 |\Q| \delta^2 } \right]$, where $C_1=6C$ and the constant $C$ follows the notation in Lemma~\ref{vqacr_lemma_qntk_least_eigen_gaussian_matrix}.
Using Eqs.~(\ref{vqacr_lemma_qntk_stability_proof1_epsilon2}) and (\ref{vqacr_lemma_qntk_stability_proof1_spectrum_norm_K_infty}) in Eq.~(\ref{vqacr_lemma_qntk_stability_proof1_2_3}) could yield the result. With probability at least $1-\delta$ we have
\begin{align}
\lmin \left[K(\bm{0}) \right] \geq{}& \lmin \left[K_{\infty}(\bm{0}) \right] - \left\| K(\bm{0}) - K_{\infty}(\bm{0}) \right\|_2 \notag \\
\geq{}& \left(1-\frac{k}{2}\right) \Delta_S \alpha_{\min} A_{\Q, \min}'^2 - \frac{k}{\sqrt{6}} \cdot \frac{\Delta_S \alpha_{\min} A_{\Q, \min}'^2 }{\sqrt{ \|O\|_2^2 \alpha_{\max} A_{\Q, \max}'^2 } } \cdot \sqrt{1+\frac{k}{2}} \cdot \sqrt{ \|O\|_2^2 \alpha_{\max} A_{\Q, \max}'^2 }  \notag \\
\geq{}& (1-k) \Delta_S \alpha_{\min} A_{\Q, \min}'^2 \notag
\end{align}
and
\begin{align}
\lmax \left[K(\bm{0}) \right] \leq{}& \lmax \left[K_{\infty}(\bm{0}) \right] + \left\| K(\bm{0}) - K_{\infty}(\bm{0}) \right\|_2 \notag \\
\leq{}& \left(1+\frac{k}{2}\right) \|O\|_2^2 \alpha_{\max} A_{\Q, \max}'^2 + \frac{k}{\sqrt{6}} \cdot \frac{\Delta_S \alpha_{\min} A_{\Q, \min}'^2 }{\sqrt{ \|O\|_2^2 \alpha_{\max} A_{\Q, \max}'^2 } } \cdot \sqrt{1+\frac{k}{2}} \cdot \sqrt{ \|O\|_2^2 \alpha_{\max} A_{\Q, \max}'^2 } \notag \\
\leq{}& \left(1+\frac{k}{2}\right) \|O\|_2^2 \alpha_{\max} A_{\Q, \max}'^2 + \frac{k}{2} \Delta_S \alpha_{\min} A_{\Q, \min}'^2 \notag \\
\leq{}& (1+k) \|O\|_2^2 \alpha_{\max} A_{\Q, \max}'^2 . \notag
\end{align}

\subsection{Proof of Lemma~\ref{vqacr_lemma_qntk_least_eigen_random_sphere_main}}

\label{vqacr_lemma_qntk_stability_proof2}
The main idea is similar to Section~\ref{vqacr_lemma_qntk_stability_proof1}. We follow notations in Lemmas~\ref{vqacr_lemma_qntk_lower_bound_xy} and \ref{vqacr_lemma_qntk_least_eigen_random_sphere}. 
We remark that 
\begin{equation*}
K(\bmt) = \frac{1}{D} J(\bmt)^T J(\bmt)
\end{equation*}
by Eq.~(\ref{vqacr_qntk_eq}), where the row $J_{d}(\bm{\theta}) = \{ \frac{\partial z_a}{\partial \theta_{d}} (\bm{\theta}) \}_{a \in \A}$. Thus, $K(\bm{0})$ and $K_{\infty}(\bm{0})$ are the empirical and the explicit covariance matrix of independent random vectors $\{J_{d}(\bm{0})\}$, respectively. Since each Hamiltonian $H_d$ involves $S$ qubits, we have
\begin{align}
\mathop{\E}_{a} J_{da}(\bm{0})^2 ={}& \mathop{\E}_{a}  \frac{1}{4} \left( \Tr \left[ i [O, H_d] \rho_a \right] \right)^2 \leq \frac{\|O\|_2^2}{2^S} \notag
\end{align}
by using Eq.~(\ref{vqacr_lemma_qntk_least_eigen_random_sphere_iso_cond}). We remark that each element in the vector $J_d(\bm{0})$ is independent from each other due to the independent sample $a \in \A$. By using the central limit theorem, there exists an absolute constant $c_1>0$, such that 
\begin{align}
P \left( \|J_{d}(\bm{0})\|_2^2 \geq{} c_1 |\A| \frac{\|O\|_2^2}{2^S} \ln \frac{D}{\delta} \right) \leq{}& 1 - \frac{\delta}{3D} . \notag
\end{align}
Thus, with probability at least $1-\delta/3$ we have that for all $d = 1,2,\cdots,D$,
\begin{align}
\|J_{d}(\bm{0})\|_2^2 \leq{} c_1 |\A| \frac{\|O\|_2^2}{2^S} \ln \frac{D}{\delta} . \label{vqacr_lemma_qntk_stability_proof2_bound_Jd_all}
\end{align}
By using Lemma~\ref{vqacr_lemma_qntk_operator_norm_of_covariance_matrix} and Eq.~(\ref{vqacr_lemma_qntk_stability_proof2_bound_Jd_all}), we have that with probability at least $1-2\delta/3$, 
\begin{align}
{}& \left\| K(\bm{0}) - K_{\infty}(\bm{0}) \right\|_2 \notag \\
={}& \left\| \frac{1}{D} \sum_{d=1}^{D} J_d(\bm{0}) J_d(\bm{0})^T - K_{\infty}(\bm{0}) \right\|_2 \notag \\
\leq{}& \epsilon \cdot \max \left( \sqrt{\left\|K_{\infty}(\bm{0})\right\|_2} , \epsilon \right) , \label{vqacr_lemma_qntk_stability_proof2_2_3}
\end{align}
where 
\begin{equation}\label{vqacr_lemma_qntk_stability_proof2_epsilon1}
\epsilon = \sqrt{\frac{c_1 |\A| \|O\|_2^2 }{2^S D} \ln \frac{D}{\delta}} \cdot \sqrt{\frac{1}{c_2} \ln \frac{|\A|}{\delta}}
\end{equation}
and $c_2>0$ is an absolute constant. Let $C_2:=9\max(c_1/c_2,1)$ in the condition $D_0 = C_2 \frac{|\A| \|O\|_2^4 }{k^2 \Delta_S^2 } \ln \frac{|\A|}{\delta} $ such that $D\geq D_0 \ln \frac{D_0}{\delta}$, we have
\begin{align}
\frac{\frac{D}{\delta}}{\ln \frac{D}{\delta}} \geq{}& \frac{\frac{D_0}{\delta} \ln \frac{D_0}{\delta}}{\ln \frac{D_0}{\delta} + \ln \ln \frac{D_0}{\delta}} \tag{function $f(x)=x/\ln x$ is increasing for $x\geq e$} \\
\geq{}& \frac{\frac{D_0}{\delta} \ln \frac{D_0}{\delta}}{\ln \frac{D_0}{\delta} + \frac{1}{2} \ln \frac{D_0}{\delta}} \tag{$\ln x \leq \frac{x}{2}$ for $x>0$} \\
={}& \frac{2}{3} \frac{D_0}{\delta} = 6 \max(c_1/c_2, 1) \frac{\|O\|_2^4 }{k^2 \Delta_S^2 } \frac{|\A|}{\delta} \ln \frac{|\A|}{\delta} . \label{vqacr_lemma_qntk_stability_proof2_bound_DlnD}
\end{align}

Using Eq.~(\ref{vqacr_lemma_qntk_stability_proof2_bound_DlnD}) in Eq.~(\ref{vqacr_lemma_qntk_stability_proof2_epsilon1}), we obtain 
\begin{align}
\epsilon ={}& \sqrt{\frac{c_1 |\A| \|O\|_2^2 }{2^S c_2} \ln \frac{|\A|}{\delta} } \cdot  \sqrt{\frac{\ln \frac{D}{\delta}}{D}} \leq{}  \frac{k \Delta_S }{\sqrt{6} \|O\|_2 \sqrt{2^S}} . \label{vqacr_lemma_qntk_stability_proof2_epsilon2}
\end{align}
Besides, by using Lemma~\ref{vqacr_lemma_qntk_least_eigen_random_sphere}, 
we have
\begin{align}
\left( 1-\frac{k}{2} \right) \frac{\Delta_S}{2^{S}} \leq{}& \lmin \left[ K_{\infty} (\bm{0}) \right] \leq{}   \lmax \left[ K_{\infty} (\bm{0}) \right] \leq{} \left( 1-\frac{k}{2} \right) \frac{\|O\|_2^2 }{2^{S}}
\label{vqacr_lemma_qntk_stability_proof2_spectrum_norm_K_infty}
\end{align}
with probability at least $1-\delta/3$ for $S \geq 2+ \log_2 \frac{C_1^2 |\A|^2 \ln |\A|}{k^2 |\Q|\delta^2}$, where $C_1=6C$ and the constant $C$ follows the notation in Lemma~\ref{vqacr_lemma_qntk_least_eigen_random_sphere}.
Using Eqs.~(\ref{vqacr_lemma_qntk_stability_proof2_epsilon2}) and (\ref{vqacr_lemma_qntk_stability_proof2_spectrum_norm_K_infty}) in Eq.~(\ref{vqacr_lemma_qntk_stability_proof2_2_3}) could yield the result. With probability at least $1-\delta$ we have
\begin{align}
\lmin \left[K(\bm{0}) \right] \geq{}& \lmin \left[K_{\infty}(\bm{0}) \right] - \left\| K(\bm{0}) - K_{\infty}(\bm{0}) \right\|_2 \notag \\
\geq{}& \left( 1-\frac{k}{2} \right) \frac{\Delta_S}{2^{S}} - \frac{k \Delta_S }{\sqrt{6} \|O\|_2 \sqrt{2^S}} \cdot \sqrt{1+\frac{k}{2}}  \frac{\|O\|_2 }{\sqrt{2^S}} \notag \\
\geq{}& (1-k) \frac{\Delta_S}{2^S} \notag
\end{align}
and
\begin{align}
\lmax \left[K(\bm{0}) \right] \leq{}& \lmax \left[K_{\infty}(\bm{0}) \right] + \left\| K(\bm{0}) - K_{\infty}(\bm{0}) \right\|_2 \notag \\
\leq{}& \left( 1+\frac{k}{2} \right) \frac{\|O\|_2^2}{2^{S}} + \frac{k \Delta_S }{\sqrt{6} \|O\|_2 \sqrt{2^S}} \cdot \sqrt{1+\frac{k}{2}}  \frac{\|O\|_2 }{\sqrt{2^S}} \notag \\
\leq{}& \left( 1+\frac{k}{2} \right) \frac{ \|O\|_2^2 }{2^{S}} + \frac{k}{2} \frac{\Delta_S}{ 2^S} \notag \\
\leq{}& (1+k) \frac{\|O\|_2^2}{2^S} . \notag
\end{align}

\section{Linear convergence of QNN training}
\label{vqacr_qntk_app_qnn_linear_converge}

We first present a linear convergence theory for quantum neural networks under general settings in Section~\ref{vqacr_qntk_app_qnn_linear_converge_general}. Next, we provide the proof of the linear convergence result, i.e., Theorems~\ref{vqacr_theorem_gd_xyz_lazy_theta_and_K} and \ref{vqacr_theorem_gd_xy_lazy_theta_and_K} in the main text, in Sections~\ref{vqacr_qntk_app_qnn_linear_converge_xyz} and \ref{vqacr_qntk_app_qnn_linear_converge_xy}, respectively.

\subsection{QNN training in the locally smooth region}
\label{vqacr_qntk_app_qnn_linear_converge_general}

We begin by analyzing the general QNN case. 
In this section, we prove the convergence rate of training general QNNs and demonstrate the lazy training regime in Theorem~\ref{vqacr_lemma_gd_general_lazy_theta_and_K}. 
Analogous to Theorem~G.1 in the work~\cite{NEURIPS2019_0d1a9651}, our analysis requires the local Lipschitz continuity of the Jacobian matrix $J(\bmt)$ with constants of the order $\O({D}^{\frac{1}{2}})$ as shown in Assumption~\ref{vqacr_lemma_ntk_local_Jacobian_stability_xyz}. 
In the locally smooth region defined in Assumption~\ref{vqacr_lemma_ntk_local_Jacobian_stability_xyz}, the loss function decays linearly during the gradient descent training with high probability. The decay rate is independent of the number of parameter $D$ for any $D\geq D_0: = \poly(|\A|, \lmax, \lmin^{-1})$, where $\lmin, \lmax$ denotes the least and the largest eigenvalue of the initial QNTK, respectively. We remark that the linear convergence threshold in this work is milder than that in previous research  with $D \geq \poly(2^N)$ when $\lmin^{-1} \lesssim \poly(N)$. 

\begin{theorem}\label{vqacr_lemma_gd_general_lazy_theta_and_K}
Let $\bmt(t) \in \R^D$ be the parameter of a QNN $\{V(\bmt), O\}$ at the $t$-th step of the gradient descent training on the dataset $A=\{(\rho_a, y_a)\}$. Let $\L_{\A}(t)$ and $K(t)$ be the loss function and the QNTK at the $t$-th step, respectively. Denote by $\lmax$ and $\lmin$ the largest and the least eigenvalue of $K(0)$, respectively. Let $\eta=\frac{|\A|}{D}\eta_0$ be the learning rate, where $\eta_0 \leq \lmin^{-1} $. Then for any $k \in (0,1/2]$, there exists two constants $C_0, C_3>0$, such that if Assumption~\ref{vqacr_lemma_ntk_local_Jacobian_stability_xyz} holds for the parameter $\bmt(0)$ with the constant $C_0$ and $R_0 := \frac{5}{2} {\lmin^{-1}} \sqrt{2|\A|\lmax \L_{\A}(0)} $ for every $D \geq D_0:= C_3 k^{-2} \lmax^2 \lmin^{-4} |\A| \L_{\A}(0) $, we have,
\begin{align}
{}& \L_{\A} (T) \leq{} \left( 1 - (1-k) \eta_0 \lmin \right)^{2T} \L_{\A} (0) , \label{vqacr_lemma_gd_general_lazy_theta_and_K_eq_1} \\
{}& \left\| \bm{\theta}(T) - \bm{\theta}(0) \right\|_2  \leq \frac{R_0}{\sqrt{D}} . \label{vqacr_lemma_gd_general_lazy_theta_and_K_eq_2} 
\end{align}
\end{theorem}

\begin{proof}

For convenience, we will denote by $\bm{r}(t):=\bm{r}(\bmt(t))$ and $J(t):=J(\bmt(t))$ the residual vector and the Jacobian matrix at the parameter $\bmt(t)$, respectively. Notations such as $\bm{r}(\bmt)$ will also be employed without ambiguity. 
Denote by \begin{align}
k_1 ={}& \frac{2\lmin}{5\lmax}k .\label{vqacr_lemma_gd_general_lazy_theta_and_K_k1}
\end{align}
Choose the proper constant $C_3$ such that 
\begin{align}
D_0 \geq{}& \left( \frac{2}{k} + \frac{4}{5} \right)^2 \frac{25\lmax^2}{4\lmin^4} C_0 2|\A| \L_{\A}(0) \notag \\
\geq{}& \frac{1}{k^2} \left( \frac{1+\frac{2}{5}k}{1-k} \right)^2 \frac{25\lmax^2}{4\lmin^4} C_0 2|\A| \L_{\A}(0) \notag \\
\geq{}& \frac{1}{k_1^2} \left( \frac{1+k_1}{1-k} \right)^2 \frac{1}{\lmin^2} C_0 \|\bm{r}(0)\|_2^2 . \label{vqacr_lemma_gd_general_lazy_theta_and_K_D0}
\end{align}

Our main idea is to prove the following two equations jointly by induction:
\begin{align}
\left\| \bm{r}(t) \right\|_2 \leq{}& \left[ 1 - (1-k) \eta_0 \lmin \right]^{t} \l\| \bm{r}(0) \r\|_2 , \label{vqacr_lemma_gd_general_lazy_theta_and_K_eq_5} \\
\left\| \bmt(t+1) - \bmt(t) \right\|_2 \leq{}& (1+k_1) \eta_0 \sqrt{\frac{\lmax}{D}} \left[ 1 - (1-k) \eta_0 \lmin \right]^{t} \l\| \bm{r}(0) \r\|_2  .\label{vqacr_lemma_gd_general_lazy_theta_and_K_eq_4} 
\end{align}
Then Eqs.~(\ref{vqacr_lemma_gd_general_lazy_theta_and_K_eq_1}) and (\ref{vqacr_lemma_gd_general_lazy_theta_and_K_eq_2}) could be obtained subsequently.

Next we prove Eqs.~(\ref{vqacr_lemma_gd_general_lazy_theta_and_K_eq_4}) and (\ref{vqacr_lemma_gd_general_lazy_theta_and_K_eq_5}) by induction. For the $t=0$ case, Eq.~(\ref{vqacr_lemma_gd_general_lazy_theta_and_K_eq_5}) holds trivially, and Eq.~(\ref{vqacr_lemma_gd_general_lazy_theta_and_K_eq_4}) holds by considering the gradient updation rule:
\begin{align*}
\|\bmt(1)-\bmt(0)\|_2 = \eta \|\nabla \L_{\A} (0) \|_2 = \frac{\eta}{|\A|} \| J(0) \bm{r}(0)\|_2 \leq \frac{\eta_0}{D} \|J(0)\|_2 \|\bm{r}(0)\|_2 = \frac{\eta_0}{D} \sqrt{D\lmax} \|\bm{r}(0)\|_2 .
\end{align*}
Next we assume that Eqs.~(\ref{vqacr_lemma_gd_general_lazy_theta_and_K_eq_4}) and (\ref{vqacr_lemma_gd_general_lazy_theta_and_K_eq_5}) holds for $t=0,1,\cdots,j$ and proceed to the $t=j+1$ case. 
By using $t\in \{0,1,\cdots,j\}$ cases of Eq.~(\ref{vqacr_lemma_gd_general_lazy_theta_and_K_eq_4}), we have the following bound for any $\bmt \in \{ \bmt | \|\bmt-\bmt(j)\|_2 \leq \|\bmt(j+1)-\bmt(j)\|_2 \}$,
\begin{align}
\| J(\bmt) - J(\bmt(0)) \|_2 \leq{}& \sqrt{C_0 D} \| \bmt - \bmt(0) \|_2 \notag \\
\leq{}& \sqrt{C_0 D} \left( \|\bmt - \bmt(j)\|_2 + \sum_{t=0}^{j-1} \|\bmt(t+1) - \bmt(t)\|_2 \right) \notag \\
\leq{}& \sqrt{C_0 D} \sum_{t=0}^{j} \|\bmt(t+1) - \bmt(t)\|_2 \notag \\
\leq{}& \sqrt{C_0 D} (1+k_1) \eta_0 \sqrt{\frac{\lmax}{D}} \frac{1-\left[ 1 - (1-k) \eta_0 \lmin \right]^{j+1}}{1-\left[ 1 - (1-k) \eta_0 \lmin \right]} \l\| \bm{r}(0) \r\|_2 \notag \\
\leq{}& \frac{1+k_1}{1-k} \sqrt{C_0 \lmax} \frac{1}{\lmin} \l\| \bm{r}(0) \r\|_2 \notag \\
\leq{}& k_1 \sqrt{D\lmax} . \label{vqacr_lemma_gd_general_lazy_theta_and_K_Jthetanorm}
\end{align}
The last term follows from $D \geq D_0$ and Eq.~(\ref{vqacr_lemma_gd_general_lazy_theta_and_K_D0}).
Then the $t=j+1$ case of Eq.~(\ref{vqacr_lemma_gd_general_lazy_theta_and_K_eq_5}) could be derived as follows:
\begin{align}
\l\| \bm{r}(j+1) \r\|_2 ={}& \l\| \bm{r}(j+1) - \bm{r}(j) + \bm{r}(j) \r\|_2	\notag \\
={}& \l\| J(\bmt')^T \left[ \bmt(j+1) - \bmt(j) \right] + \bm{r}(j) \r\|_2 \notag \\
={}& \l\| - \frac{\eta}{|\A|} J(\bmt')^T J(j) \bm{r}(j) + \bm{r}(j) \r\|_2 \label{vqacr_lemma_gd_general_lazy_theta_and_K_2_3} \\
\leq{}& \l\| 1 - \frac{\eta_0}{D} J(\bmt')^T J(j) \r\|_2 \l\| \bm{r}(j) \r\|_2 \notag \\
\leq{}& \left[ 1 - (1-k) \eta_0 \lmin \right] \l\| \bm{r}(j) \r\|_2 , \label{vqacr_lemma_gd_general_lazy_theta_and_K_2_5}
\end{align}
where $\bmt'$ is a point derived from the Taylor expansion, which satisfies $\|\bmt'-\bmt(j)\|_2 \leq \|\bmt(j+1)-\bmt(j)\|$.
Eq.~(\ref{vqacr_lemma_gd_general_lazy_theta_and_K_2_3}) comes from the gradient descent rule. 
Eq.~(\ref{vqacr_lemma_gd_general_lazy_theta_and_K_2_5}) is derived as follows:
\begin{align}
{}& \l\| 1 - \frac{\eta_0}{D} J(\bmt')^T J(j) \r\|_2 \notag \\
\leq{}& \l\| 1 - \frac{\eta_0}{D} J(0)^T J(0) \r\|_2 + \frac{\eta_0}{D} \l\| J(0)^T J(0) - J(0)^T J(j) \r\|_2 + \frac{\eta_0}{D} \l\| J(0)^T J(j) - J(\bmt')^T J(j) \r\|_2 \notag \\
\leq{}& 1 - \eta_0 \lmin + \frac{\eta_0}{D} \left( \l\| J(0) \r\|_2 \l\| J(0)  - J(j) \r\|_2 + \l\| J(j)-J(0)+J(0) \r\|_2 \l\| J(0)  - J(\bmt') \r\|_2 \right) \notag \\
\leq{}&  1 - \eta_0 \lmin + \frac{\eta_0}{D} \sqrt{D\lmax} \cdot k_1 \sqrt{D\lmax} + \frac{\eta_0}{D} \left( k_1 \sqrt{D\lmax} + \sqrt{D\lmax} \right) k_1 \sqrt{D\lmax} \tag{using Eq.~(\ref{vqacr_lemma_gd_general_lazy_theta_and_K_Jthetanorm})} \\
={}& 1-\eta_0 \lmin + k_1 (2+k_1) \eta_0 \lmax \notag \\
\leq{}& 1-\eta_0 \lmin + \frac{2}{5} k (2+k_1) \eta_0 \lmin \notag \\
\leq{}& 1 - (1-k) \eta_0 \lmin , \notag
\end{align}
where the last term is derived by using $k_1 \leq 2/5$.
The $t=j+1$ case of Eq.~(\ref{vqacr_lemma_gd_general_lazy_theta_and_K_eq_4}) could be obtained accordingly:
\begin{align}
\|\bmt(j+2)-\bmt(j+1)\|_2 ={}& \eta \left\|\nabla \L_{\A} (j+1) \right\|_2 \notag \\
\leq{}& \frac{|\A|}{D} \eta_0 \frac{1}{|\A|} \left\|J(j+1) \right\|_2 \left\|\bm{r}(j+1) \right\|_2 \label{vqacr_lemma_gd_general_lazy_theta_and_K_4_2} \\
={}& \frac{\eta_0}{D} \left\|J(j+1) - J(0) + J(0)\right\|_2 \left\|\bm{r}(j+1) \right\|_2 \notag \\
\leq{}& (1+k_1) \frac{\eta_0}{D} \sqrt{D\lmax} \left\|\bm{r}(j+1) \right\|_2 \label{vqacr_lemma_gd_general_lazy_theta_and_K_4_4} \\
\leq{}& (1+k_1) \eta_0 \sqrt{\frac{\lmax}{D}} \left[ 1 - (1-k) \eta_0 \lmin \right]^{j+1} \l\| \bm{r}(0) \r\|_2 .\label{vqacr_lemma_gd_general_lazy_theta_and_K_4_5} 
\end{align}
Eq.~(\ref{vqacr_lemma_gd_general_lazy_theta_and_K_4_2}) follows from $\eta=\frac{|\A|}{D} \eta_0$ and $\L_{\A} = \frac{1}{2|\A|} \|\bm{r}\|_2^2$. Eq.~(\ref{vqacr_lemma_gd_general_lazy_theta_and_K_4_4}) is derived by using the local smoothness bound of $J$ in Eq.~(\ref{vqacr_lemma_gd_general_lazy_theta_and_K_Jthetanorm}). Eq.~(\ref{vqacr_lemma_gd_general_lazy_theta_and_K_4_5}) is derived by using the $t=j+1$ case of Eq.~(\ref{vqacr_lemma_gd_general_lazy_theta_and_K_eq_5}).
Thus, we have proved Eqs.~(\ref{vqacr_lemma_gd_general_lazy_theta_and_K_eq_4}) and (\ref{vqacr_lemma_gd_general_lazy_theta_and_K_eq_5}) for $t=0,1,\cdots,+\infty$.

Finally, we proceed to derive the main result in the theorem. The linear convergence result in Eq.~(\ref{vqacr_lemma_gd_general_lazy_theta_and_K_eq_1}) could be obtained from Eq.~(\ref{vqacr_lemma_gd_general_lazy_theta_and_K_eq_5}) and $\L_{\A}(t)=\frac{1}{2|\A|} \|\bm{r}(t)\|_2^2$. The frozen parameter regime in Eq.~(\ref{vqacr_lemma_gd_general_lazy_theta_and_K_eq_2}) could be obtained from Eq.~(\ref{vqacr_lemma_gd_general_lazy_theta_and_K_eq_4}) as follows:
\begin{align}
\left\| \bmt(T) - \bmt(0) \right\|_2 \leq{}& \sum_{t=0}^{T-1} \left\| \bmt(t+1) - \bmt(t) \right\|_2 \notag \\
\leq{}& \sum_{t=0}^{T-1} (1+k_1) \eta_0 \sqrt{\frac{\lmax}{D}} \left[ 1 - (1-k) \eta_0 \lmin \right]^{t} \l\| \bm{r}(0) \r\|_2 \notag \\
\leq{}& (1+k_1) \eta_0 \sqrt{\frac{\lmax}{D}} \frac{1}{1-\left[ 1 - (1-k) \eta_0 \lmin \right]} \l\| \bm{r}(0) \r\|_2 \notag \\
={}& \frac{1+k_1}{1-k} \frac{\sqrt{\lmax}}{\lmin} \sqrt{2|\A|\L_{\A}(0)} \frac{1}{\sqrt{D}} \notag \\
\leq{}& \frac{1+\frac{2}{5}k}{1-\frac{1}{2}} \frac{1}{\lmin} \sqrt{2|\A| \lmax\L_{\A}(0)} \frac{1}{\sqrt{D}} \notag \\
\leq{}& \frac{R_0}{\sqrt{D}} .
\end{align}

\end{proof}

\subsection{Proof of Theorem~\ref{vqacr_theorem_gd_xyz_lazy_theta_and_K}}
\label{vqacr_qntk_app_qnn_linear_converge_xyz}

We choose constants $c_1=\frac{9}{4}C_2$, $c_2=4C_3$, and $c_3=9C_1^2$, where constants $C_1$ and $C_2$ follow from Lemma~\ref{vqacr_lemma_qntk_least_eigen_gaussian_matrix_main} and the constant $C_3$ follows from Theorem~\ref{vqacr_lemma_gd_general_lazy_theta_and_K}.
Then by using Lemma~\ref{vqacr_lemma_qntk_least_eigen_gaussian_matrix_main}, for $D\geq D_1 \ln \frac{D_1}{\delta}$ with $D_1 = \left( \frac{\lmax' }{\lmin' A_{\Q, \max}' } \right)^2  c_1 |\A| \ln \frac{|\A|}{\delta} $ and $S \geq{} \log_3 \left[ 1 + c_3 \frac{ |\A|^2 \ln |A| }{ A_{\Q, \min}'^4 |\Q| \delta^2 } \right]$, the following holds with probability at least $1-\delta$:
\begin{align}
\lmin' = \frac{2}{3} \Delta_S \alpha_{\min} A_{\Q, \min}'^2 \leq{}& \lambda \left[K(\bmt(0)) \right] \leq{} \frac{4}{3} \|O\|_2^2 \alpha_{\max} A_{\Q, \max}'^2 = \lmax', \label{vqacr_theorem_gd_xyz_lazy_theta_and_K_min_max_eigen}
\end{align}

Our main idea is to derive the convergence result based on Theorem~\ref{vqacr_lemma_gd_general_lazy_theta_and_K}, which holds with probability at least $1-\delta$ with eigenvalues in terms of $\lmin'$ and $\lmax'$. 
Let $k=1/2$, the threshold $D_2$ satisfies
\begin{align}
D_2 ={}& c_2 \lmax'^2 \lmin'^{-4} |\A| \L_{\A}(0)  \notag \\
={}& C_3 k^{-2} \lmax'^2 \lmin'^{-4} |\A| \L_{\A}(0). \label{vqacr_theorem_gd_xyz_lazy_theta_and_K_D2}
\end{align}

By employing the formulation of $D_2$ in Eq.~(\ref{vqacr_theorem_gd_xyz_lazy_theta_and_K_D2}) in Theorem~\ref{vqacr_lemma_gd_general_lazy_theta_and_K}, we have the following result that holds with probability $1-\delta$,
\begin{align*}
{}& \L_{\A} (T) \leq{} \left( 1 - \frac{1}{2} \eta_0 \lmin' \right)^{2T} \L_{\A} (0) ,  \\
{}& \left\| \bm{\theta}(T) - \bm{\theta}(0) \right\|_2  \leq \frac{R_0}{\sqrt{D}} .
\end{align*}

Thus, we have proved Theorem~\ref{vqacr_theorem_gd_xyz_lazy_theta_and_K}.

\subsection{Proof of Theorem~\ref{vqacr_theorem_gd_xy_lazy_theta_and_K}}
\label{vqacr_qntk_app_qnn_linear_converge_xy}

We choose constants $c_1= 9 C_2$, $c_2=36C_3$, and $c_3=9C_1^2$, where constants $C_1$ and $C_2$ follow from Lemma~\ref{vqacr_lemma_qntk_least_eigen_random_sphere_main} and the constant $C_3$ follows from Theorem~\ref{vqacr_lemma_gd_general_lazy_theta_and_K}.
Then by using Lemma~\ref{vqacr_lemma_qntk_least_eigen_random_sphere_main}, for 
$D\geq D_1 \ln \frac{D_1}{\delta}$ with $D_1 = 9 C_2 \frac{|\A| \|O\|_2^4 }{ \Delta_S^2 } \ln \frac{|\A|}{\delta} $ and $S \geq 2+ \log_2 \frac{9C_1^2 |\A|^2 \ln |\A|}{ |\Q|\delta^2}$, the following holds with probability at least $1-\delta$:
\begin{align}
\lmin' := \frac{2}{3} \frac{\Delta_S}{2^S} \leq{}& \lambda \left[K(\bmt(0)) \right] \leq{} \frac{4}{3} \frac{\|O\|_2^2}{2^S} := \lmax' .\label{vqacr_theorem_gd_xy_lazy_theta_and_K_min_max_eigen}
\end{align}

Our main idea is to derive the convergence result based on Theorem~\ref{vqacr_lemma_gd_general_lazy_theta_and_K}, which holds with probability at least $1-\delta$ with eigenvalues in terms of $\lmin'$ and $\lmax'$. 
Specifically, the term $R_0$ in the locally smooth condition in Assumption~\ref{vqacr_lemma_ntk_local_Jacobian_stability_xyz} satisfies
\begin{align}
R_0 ={}& \frac{5}{2} \frac{ \|O\|_2 }{\Delta_S } \sqrt{6 \times 2^S |\A| \L_{\A}(0)} \notag \\
={}& \frac{5}{2} \frac{ 1 }{ \frac{2}{3} \Delta_S 2^{-S} } \sqrt{\frac{4}{3} \|O\|_2^2 2^{-S} \cdot 2 |\A| \L_{\A}(0)} \notag \\
={}& \frac{5}{2} {\lmin'^{-1}} \sqrt{2|\A|\lmax' \L_{\A}(0)} . \label{vqacr_theorem_gd_xy_lazy_theta_and_K_R0}
\end{align}
Let $k=1/2$, the threshold $D_2$ satisfies
\begin{align}
D_2 ={}& 36 C_3 \frac{\|O\|_2^4 2^{2S} }{ \Delta_S^4  }  |\A| \L_{\A}(0)  \notag \\
={}& C_3 k^{-2} \frac{\frac{16}{9} \|O\|_2^4 2^{-2S} }{ \frac{16}{81} \Delta_S^4 2^{-4S} } |\A| \L_{\A}(0)  \notag \\
={}& C_3 k^{-2} \lmax'^2 \lmin'^{-4} |\A| \L_{\A}(0). \label{vqacr_theorem_gd_xy_lazy_theta_and_K_D2}
\end{align}

By employing the formulation of $R_0$ and $D_2$ in Eqs.~(\ref{vqacr_theorem_gd_xy_lazy_theta_and_K_R0}) and (\ref{vqacr_theorem_gd_xy_lazy_theta_and_K_D2}) in Theorem~\ref{vqacr_lemma_gd_general_lazy_theta_and_K}, we have the following result that holds with probability $1-\delta$,
\begin{align*}
{}& \L_{\A} (T) \leq{} \left( 1 - \frac{1}{2} \eta_0 \lmin' \right)^{2T} \L_{\A} (0) ,  \\
{}& \left\| \bm{\theta}(T) - \bm{\theta}(0) \right\|_2  \leq \frac{R_0}{\sqrt{D}} . 
\end{align*}
Theorem~\ref{vqacr_theorem_gd_xy_lazy_theta_and_K} could be easily derived by using eigenvalue bounds in Eq.~(\ref{vqacr_theorem_gd_xy_lazy_theta_and_K_min_max_eigen}).

\end{document}